\documentclass[ecta,final]{econsocart}
\usepackage{preamble}
\usepackage{preamble_local}
\usepackage{mathtools}
\usepackage[normalem]{ulem}
\usepackage{enumitem}
\setlist{itemsep=0.5em}     
\setlist{topsep=0.7em}      
\setlist{parsep=0.2em}        

\usepackage{multirow}
\usepackage{booktabs}
\usepackage{booktabs}
\usepackage{array}
\usepackage{amsmath}
\usepackage{helvet} 
\usepackage{graphicx}  
\usepackage{multirow}
\usepackage{subcaption}
\usepackage{caption}
\usepackage{bm}
\usepackage{bbm}
\usepackage{colortbl}
\usepackage{xcolor}
\usepackage{multirow}
\usepackage{booktabs}
\usepackage{bigints}
\usepackage{apptools}
\AtAppendix{\counterwithin{theorem}{section}}
\DeclareFontFamily{U}{mathx}{}
\DeclareFontShape{U}{mathx}{m}{n}{<-> mathx10}{}
\DeclareSymbolFont{mathx}{U}{mathx}{m}{n}
\DeclareMathAccent{\widecheck}{0}{mathx}{"71}

\RequirePackage[colorlinks,citecolor=blue,linkcolor=blue,urlcolor=blue,pagebackref]{hyperref}
\usepackage{cleveref}

\startlocaldefs

\theoremstyle{plain}

\newtheorem{theorem}{Theorem}

\newtheorem{lemma}{Lemma}

\theoremstyle{remark}

\renewcommand{\P}{\mathbb{P}}

\def\hatt{{\hat{\tau}}}
\def\bmu{{\boldsymbol \mu}}
\def\bSigma{{\boldsymbol \Sigma}}
\def\bx{{\bm{x}}}
\def\bX{{\bm{X}}}

\def\bU{{\bm{U}}}

\def\bR{{\bm{R}}}

\def\by{{\bm{y}}}
\def\bz{{\bm{z}}}
\def\bb{{\bm{b}}}

\def\bf{{\bm{f}}}
\def\br{{\bm{r}}}

\def\bSigma{{\boldsymbol \Sigma}}
\def\blambda{{\boldsymbol \Lambda}}
\def\bv{{\bm{v}}}
\def\bu{{\bm{u}}}
\def\bU{{\bm{U}}}

\def\bd{{\bm{d}}}
\def\br{{\bm{r}}}
\def\bt{{\bm{t}}}
\def\bV{{\bm{V}}}

\def\bI{{\bm{I}}}
\renewcommand{\bar}[1]{\overline{#1}}
\renewcommand{\tilde}[1]{\widetilde{#1}}
\renewcommand{\hat}[1]{\widehat{#1}}
\newcommand{\lam}[1]{{\boldsymbol\Lambda}_{\lambda}^{(-#1)}}

\newcommand{\norm}[2]{\left\|#1\right\|_{#2}}
\newcommand{\abs}[1]{\left|#1\right|}

\newcommand{\hb}[1]{\hat{\bbeta}^{(-i)}}
\newcommand{\pto}{\xrightarrow{p}}
\newcommand{\dto}{\xrightarrow{d}}

\newcommand{\wt}[1]{\widetilde{#1}}
\usepackage{amssymb}
\usepackage{pifont}
\usepackage{cancel}

\newcommand{\cmark}{\ding{51}} 
\newcommand{\xmark}{\ding{55}} 

\makeatletter
\newcommand{\ogeneric}[2][0.7]{%
  \vphantom{\oplus}\mathpalette\o@generic{{#1}{#2}}%
}
\newcommand{\o@generic}[2]{\o@@generic#1#2}
\newcommand{\o@@generic}[3]{%
  \begingroup
  \sbox\z@{$\m@th#1\oplus$}%
  \dimen@=\dimexpr\ht\z@+\dp\z@\relax
  \savebox\tw@[\totalheight]{$\m@th#1\bigcirc$}%
  \makebox[\wd\z@]{%
    \ooalign{%
      $#1\vcenter{\hbox{\resizebox{\dimen@}{!}{\usebox\tw@}}}$\cr
      \hidewidth
      $#1\vcenter{\hbox{\resizebox{#2\dimen@}{!}{$#1\vphantom{\oplus}{#3}$}}}$%
      \hidewidth
      \cr
    }%
  }%
  \endgroup
}
\makeatother

\newcommand{\argmin}{\mathop{\mathrm{argmin}\,}\limits}

\endlocaldefs

\begin{document}

\begin{frontmatter}

\title{Unbiased Regression-Adjusted Estimation of Average Treatment Effects in Randomized Controlled Trials
}
\runtitle{Unbiased Regression-Adjusted Estimation of ATE in RCTs}

\begin{aug}
%
%
%
\author[id=au1,addressref={add1}]{\fnms{Alberto}~\snm{Abadie}\ead[label=e1]{abadie@mit.edu}}
\author[id=au2,addressref={add2}]{\fnms{Mehrdad}~\snm{Ghadiri}\ead[label=e2]
{mehrdadg@mit.edu}}
\author[id=au3,addressref={add2}]{\fnms{Ali}~\snm{Jadbabaie}\ead[label=e3]
{jadbabai@mit.edu}}
\author[id=au4,addressref={add2}]{\fnms{Mahyar}~\snm{JafariNodeh}\ead[label=e4]{mahyarjn@mit.edu}}
\address[id=add1]{%
\orgdiv{Department of Economics}, \orgname{Massachusetts Institute of Technology}}
\address[id=add2]{%
\orgdiv{Institute for Data, Systems, and Society}, \orgname{Massachusetts Institute of Technology}}
\end{aug}

\support{Partially supported by ONR grants N00014-24-1-2687 (Abadie), N00014-23-1-2299 (Jadbabaie), and by a Michael Hammer Postdoctoral Fellowship (Ghadiri). We thank Javier Gardeazabal and Whitney Newey for helpful discussions.}
\coeditor{\fnm{[Name} \snm{Surname}; will be inserted later]}

\begin{abstract}
This article introduces a leave-one-out regression adjustment (LOORA) for estimating average treatment effects in randomized controlled trials. In finite samples, LOORA removes the bias of conventional regression adjustment and yields exact variance formulas for regression-adjusted Horvitz–Thompson and difference-in-means estimators. Ridge regularization curbs the influence of high-leverage observations, improving stability and precision in small samples. In large samples, LOORA matches the variance of the regression-adjusted estimator in \cite{lin2013agnostic} while remaining exactly unbiased. Two within-subject experimental applications, each providing a realistic joint distribution of potential outcomes as ground truth, show that LOORA removes substantial bias and achieves confidence interval coverage close to the nominal level.
\end{abstract}

\end{frontmatter}

\section{Introduction}
\label{sec:intro}

\noindent Establishing causal relationships is a central goal in economics and the social sciences. In observational studies, selection bias often obscures the effects of interventions. Randomized controlled trials (RCTs) address this problem by eliminating systematic pre-randomization differences between treatment and control groups. For this reason, RCTs occupy a central place in modern empirical research on causal inference.

As experimental practice has matured, attention has shifted toward improving efficiency and statistical power. Researchers now collect extensive pretreatment information to predict outcomes and use regression adjustment to explain part of the outcome variance. This approach reduces noise and increases the precision of treatment effect estimation without requiring complex assignment mechanisms.

Most formal guarantees for regression-adjusted estimators rely on asymptotic arguments. However, early-stage clinical trials, marketing geo-experiments, and studies of firm- or country-level policies frequently operate in small sample regimes, where standard regression-adjusted estimators become biased, with undesirable consequences for treatment effect estimation and policy design \citep{freedman2008regression}. Moreover, in small samples, high-leverage observations (that is, observations whose covariate values exert a disproportionate influence on the regression fit) further degrade the performance of regression adjustment \citep{young2019channeling}. 

To address these challenges, we propose a regression adjustment framework that is unbiased in finite samples and robust to influential observations. A key component of our analysis is a leave-one-out regression adjustment (LOORA) procedure, which removes the bias of classical regression adjustment. Within this framework, we develop two estimators of the average treatment effect (ATE): LOORA-HT, a Horvitz–Thompson estimator for simple random assignment; and LOORA-DM, a difference-in-means estimator for complete random assignment. LOORA-HT and LOORA-DM are easy to implement and deliver a performance that previously required complex assignment mechanisms \citep[see][]{harshaw2022design}.

LOORA-HT and LOORA-DM simultaneously address three core problems in regression-adjusted estimation of treatment effects. First, they answer the critique of \citet{freedman_2008_b} by providing unbiased, regression-adjusted estimators for RCTs. Second, they attain the asymptotic variance of the estimator in \citet{lin2013agnostic}, which is the efficient variance among all linearly adjusted estimators under simple and complete random assignment. Third, by stabilizing estimation in the presence of influential observations, they reduce the sensitivity to leverage documented in \citet{young2019channeling}.

\subsection{Related Work}
\label{subsec:related-work}

The practice of regression adjustment in randomized experiments dates at least to the use of analysis of covariance (ANCOVA) in classical experimental design, where a linear adjustment removes the effects of pretreatment covariates to improve precision.\footnote{See, for example, discussions in \citet[Chapter 9]{fisher1971design}.}  
Despite its pervasiveness, the formal large-sample properties of regression adjustment after randomization have been repeatedly questioned.  
In two influential critiques, \citet{freedman2008regression,freedman_2008_b} argue that ordinary least squares (OLS) adjustments can introduce bias when the regression model is misspecified, and can be less efficient than the unadjusted difference-in-means estimator, even asymptotically.

In response, \citet{lin2013agnostic} shows that, under complete random assignment, an OLS regression of outcomes on treatment, covariates, and all treatment-by-covariate interactions delivers an estimator that is consistent for the average treatment effect and asymptotically efficient among the class of linearly adjusted estimators.  
While these results clarify the large sample properties of regression adjustment after randomization, two important gaps remain: (i) \citet{lin2013agnostic} does not provide finite-sample unbiasedness guarantees for regression adjustment, and (ii) the interacted specification in \citet{lin2013agnostic} can be unstable in moderate samples because it increases the number of regressors by adding interaction terms and potentially inflating leverage scores and sensitivity to influential observations.

This article proposes two  estimators, LOORA-HT under simple random assignment and LOORA-DM under complete random assignment, that (i) are exactly unbiased for the finite-population average treatment effect, (ii) admit closed-form variance expressions in finite samples, and (iii) 
are asymptotically efficient among linearly adjusted estimators, attaining the variance bound of \citet{lin2013agnostic}. We restrict our attention to linear regression adjustments because of their ubiquity in empirical research and the potential instability of nonlinear adjustments in small samples. Linear adjustments, however, can accommodate nonlinear relationships through covariate transformations.

LOORA combines leave-one-out regression adjustment and ridge regression to remove the finite-sample bias of regression-adjusted estimators and to curb the influence of high-leverage observations. 
For the difference-in-means setting, we show that LOORA-DM admits a representation that is equivalent to a leave-two-out construction in the sense of \citet{spiess2025optimal}, thereby satisfying necessary conditions for unbiasedness under complete random assignment.  
In this way, our estimators provide constructive finite-sample analogues of the asymptotic guarantees in \citet{lin2013agnostic}, while remaining valid without the assumption of correctly-specified parametric outcome models.

In two within-subject experiments that provide a realistic ground truth for the joint distribution of potential outcomes, LOORA-HT and LOORA-DM substantially reduce estimation bias and yield confidence intervals with close-to-nominal coverage.

A related line of work examines precision gains in randomized studies through experimental design rather than regression adjustment.
\citet{harshaw2022design} proposes a Gram--Schmidt walk design, which chooses treatment assignments in a way that approximately balances covariates. 
Their goal is to optimize the randomization design to deliver low-variance Horvitz--Thompson estimators in finite samples. 
In contrast, we consider simple or complete random assignment and instead adjust outcomes through leave-one-out regressions.  
This distinction matters in practice, because our methods can be applied to commonly used experimental designs without the need to change the assignment rule. 

\citet{ghadiri2023finite} introduces a cross-fitted, regression-adjusted Horvitz--Thompson estimator under simple random assignment. It splits the sample into two folds, estimates regression coefficients on one fold, and applies those coefficients to adjust outcomes in the other. This construction yields an exactly unbiased estimator with an explicit nonasymptotic variance bound. However, the bound is looser than the \citet{lin2013agnostic} asymptotic efficiency bound, and the article does not provide a procedure for constructing confidence intervals.
\citet{lei2021regression} proposes a bias-correction procedure for regression adjustment when the number of covariates grows with the sample size. The method removes higher-order bias terms but does not yield exact unbiasedness in finite samples.
\citet{chang2024exact} develops an exact bias-correction approach that eliminates the bias, but does not derive the finite-sample variance of the resulting estimator.
 \cite{gu2025assumption} studies regression adjustment with high-dimensional covariates under a superpopulation framework and develops consistent estimators in this setting. \cite{zhao2024covariate} considers a related framework and proposes a debiasing approach based on higher-order influence functions that corrects own-observation bias arising from regression adjustment. Their theoretical development primarily focuses on unregularized linear adjustment, or regimes in which the ridge regularization parameter vanishes asymptotically, which is equivalent to using the Moore--Penrose pseudoinverse of the Gram matrix in the associated regression problem. In contrast, our analysis accommodates arbitrary regularization parameters. This flexibility enables explicit damping of high-leverage observations, which can otherwise induce unstable or unreliable inference, especially in finite-population or small-sample settings.

A large concurrent literature studies regression adjustment and inference in more structured designs: stratified or blocked experiments, matched pairs, cluster randomization, or covariate-adaptive randomization.  
For example, \citet{bai2024primer} survey modern design-based analysis of randomized experiments, emphasizing the role of stratification, regression adjustment and cluster randomized experiments.  
\citet{bai2025new,bai2025inference} develop design-based variance estimators and inference procedures for finely stratified and matched-pair experiments, including settings with imperfect compliance.
\citet{cytrynbaum2024covariate} characterizes asymptotically optimal linear covariate adjustment for general stratified randomization schemes.

Our approach provides an appealing alternative in settings where stratification becomes impractical or infeasible. In particular, stratification requires covariate information at the design stage, which is often unavailable in online experiments and field experiments with asynchronous enrollment. By contrast, simple and complete randomization designs apply naturally to these settings.

\Cref{sec:detailed-related-work} provides a more detailed comparison of our estimators with related methods.

\bigskip

\subsection{Notation}

\paragraph{Vectors and matrices.}
We denote matrices and vectors by bold uppercase and lowercase letters, respectively. The $i^{\text{th}}$ entry of a vector $\bu$ is denoted by $u_i$. The transposed $i^{\text{th}}$ row of a matrix $\bX$ is denoted by $\bx_i$, and its $(i,j)^{\text{th}}$ entry by $x_{ij}$. For a constant $c$ and a vector $\bu$, expressions such as $c + \bu$ and $\bu^{-1}$ are interpreted entrywise. When $\bu$ and $\bv$ are vectors of equal dimension, $\bu \bv$ and $\bu / \bv$ denote entrywise (Hadamard) product and division, respectively.  
A vector’s associated diagonal matrix appears in uppercase; for instance, $\bT$ denotes the diagonal matrix $\bt$ as its main diagonal.  
For any vector $\bv \in \R^n$, the notation $\bv_{-i}$ refers to the vector in $\R^{n-1}$ obtained by removing the $i^{\text{th}}$ entry of $\bv$. Likewise, $\bX_{-i}$ denotes the matrix obtained by deleting the $i^{\text{th}}$ row of $\bX$. We denote by $\boldsymbol{1}$ the vector with all entries equal to one. 
For a matrix $\bX \in \R^{n\times k}$, let $\overline{\bX}$ denote the matrix whose rows are all equal to the average row of $\bX$, that is, each row of $\overline{\bX}$ equals $ \boldsymbol{1}^\top \bX/n$.
Similarly, for a vector $\by \in \R^n$, let $\overline{\by}$ denote the vector whose entries are all equal to the average of the entries in $\by$.

{\em Norms and asymptotic notation.}
For a vector $\bu \in \R^n$, $\norm{\bu}{2} = \sqrt{\bu^\top \bu}$ denotes the Euclidean norm, and $\norm{\bu}{\infty} = \max_{i \in [n]} |u_i|$ denotes the $\ell_\infty$ norm.  
The operator $\tr(\cdot)$ denotes the trace of a square matrix.  
For a matrix $\bX \in \R^{n \times k}$, $\norm{\bX}{F} = \sqrt{\tr(\bX^\top \bX)}$ denotes the Frobenius norm, and $\norm{\bX}{2,\infty} = \max_{i \in [n]} \norm{\bx_i}{2}$ denotes its $(2,\infty)$-operator norm.  
Let $\mathbb{N}$ be the set of natural numbers. For any two functions $f, g \colon \mathbb{N} \to \R$, we write $f(n) = O(g(n))$ if there exist constants $C > 0$ and $n_0 \in \mathbb{N}$ such that for all $n \ge n_0$, $\lvert f(n) \rvert \le C \lvert g(n) \rvert$. Similarly, we write $f(n) = \Omega(g(n))$ if there exist constants $C > 0$ and $n_0 \in \mathbb{N}$ such that for all $n \ge n_0$, $\lvert f(n) \rvert \ge C \lvert g(n) \rvert$.

{\em Projection matrices.}
Consider a full-rank matrix $\bX \in \R^{n \times k}$ with $n \ge k$.  
We denote the projection matrix by $\bH = \bX(\bX^\top \bX)^{-1}\bX^\top$, and its diagonal entries by $h_{ii} = \bx_i^\top (\bX^\top \bX)^{-1} \bx_i$, which represent the leverage scores of rows $i \in [n]$.  
The identity matrix is denoted by $\bI$. For $\lambda \ge 0$, the ridge projection matrix is $\bH_\lambda = \bX(\bX^\top \bX + \lambda \bI)^{-1}\bX^\top$, whose diagonal entries $h_{\lambda ii} = \bx_i^\top (\bX^\top \bX + \lambda \bI)^{-1} \bx_i$ are the ridge leverage scores.

\section{Finite Population Framework}
We operate within the Neyman--Rubin potential outcomes framework \citep{neyman1923application,rubin1974estimating} for a randomized controlled trial over a set of units indexed by $[n] = \{1, \ldots, n\}$.
Each unit is assigned to either treatment or control. The vector $\bd \in \R^{n}$ denotes the treatment assignment, with $d_i = 1$ if unit $i$ is in treatment and $d_i = 0$ if it is in control. We also define $\bz = 2(\bd - 1/2)$, so that $z_i = 1$ for treated units and $z_i = -1$ otherwise.
Finally, we define $\bq$ as the vector with $i^{\text{th}}$ element equal to $q_i = p_i d_i + (1-p_i)(1-d_i)$, where $p_i$ is the treatment-assignment probability for unit $i$.

For each unit $i$, the potential outcomes under treatment and control are denoted by $y^{(1)}_i$ and $y^{(0)}_i$, respectively. We treat potential outcomes as fixed quantities; the only source of randomness in our setting is the treatment assignment. The observed outcome for unit $i$ is
\[
y_i = y^{(1)}_i d_i + y^{(0)}_i (1 - d_i).
\]
We collect the observed and potential outcomes across all units in the vectors $\by = (y_1, \ldots, y_n)^\top$, $\by^{(1)} = (y^{(1)}_1, \ldots, y^{(1)}_n)^\top$, and $\by^{(0)} = (y^{(0)}_1, \ldots, y^{(0)}_n)^\top$.

The average treatment effect (ATE) is defined as
\[
\tau = \frac{1}{n} \sum_{i=1}^n \bigl(y^{(1)}_i - y^{(0)}_i\bigr).
\]

Each unit $i \in [n]$ comes with a vector of fixed pretreatment covariates $\bx_i \in \R^k$. The matrix $\bX = (\bx_1, \ldots, \bx_n)^\top \in \R^{n \times k}$ collects covariate values for all units. For simplicity of exposition, we assume throughout the article that $\bX$ has full column rank. 
However, since our estimators are based on ridge regression, the results extend directly to cases in which $\bX$ is rank-deficient.

We consider two standard treatment assignment mechanisms. Under \textit{simple random assignment}, each unit is independently assigned to treatment with probability $p_i \in (0,1)$, with $\bp = (p_1, \ldots, p_n)^\top$. 
Under \textit{complete random assignment}, the number of treated units is fixed at $n_T \in [n - 1]$, and the number of control units is $n_C = n - n_T$. Treatment is assigned by selecting uniformly at random from all possible subsets of $n_T$ units in the sample.

Under simple random assignment, the Horvitz--Thompson (HT) estimator \citep{horvitz1952generalization} is defined as
\[
\hat{\tau}_{\text{HT}} = \frac{1}{n} \sum_{i=1}^n \frac{d_i y_i}{p_i} \;-\; \frac{1}{n} \sum_{i=1}^n \frac{(1 - d_i) y_i}{1 - p_i}.
\]

Under complete random assignment, the difference-in-means (DM) estimator is given by
\[
\hat{\tau}_{\text{DM}} = \frac{1}{n_T} \sum_{i=1}^n d_i y_i \;-\; \frac{1}{n_C} \sum_{i=1}^n (1 - d_i) y_i.
\]

\section{Regression Adjustment for Horvitz-Thompson Estimator}
\label{sec:reg-adj-ht}

In this section, we first review some known results on the variance of the Horvitz--Thompson (HT) estimator. We then introduce LOORA-HT, our leave-one-out regression-adjusted version of the HT estimator (\Cref{sec:loora-ht-estimator}). We establish its unbiasedness and derive an exact expression for its variance in the finite-population setting (\Cref{sec:loora-ht-var}). Next, we analyze the asymptotic behavior of LOORA-HT (\Cref{sec:asymptotic-ht}) and prove its asymptotic efficiency (\Cref{sec:loora-ht-efficiency}). Finally, we describe our approach to variance estimation and confidence interval construction for LOORA-HT (\Cref{sec:loora-ht-conf-int}).

It is well known that, under simple random assignment, the classical Horvitz--Thompson (HT) estimator is unbiased for $\tau$. 
The variance of the HT estimator under simple random assignment is given by
\begin{equation}\label{eq:ht_var}
    \frac{1}{n^2} \norm{\bmu}{2}^2,
\end{equation}
where $\bmu = (\mu_1, \ldots, \mu_n)^\top$, and
\begin{align*}
\mu_i = \sqrt{\frac{1 - p_i}{p_i}}\, y^{(1)}_i + \sqrt{\frac{p_i}{1 - p_i}}\, y^{(0)}_i.
\end{align*}

Let $\breve{y}_i = y_i - \bx_i^\top \bb$ denote the covariate-adjusted outcomes, where $\bb$ is a fixed vector of coefficients. The corresponding potential outcomes are
\begin{align*}
\breve{y}^{(1)}_i = y^{(1)}_i - \bx_i^\top \bb
\quad \text{and} \quad
\breve{y}^{(0)}_i = y^{(0)}_i - \bx_i^\top \bb. 
\end{align*}
Because $\breve{y}^{(1)}_i - \breve{y}^{(0)}_i = y^{(1)}_i - y^{(0)}_i$, it follows 
that the HT estimator applied to the adjusted outcomes $\breve{y}_i$ is also unbiased for $\tau$.  
Since $\bb$ is fixed, applying \eqref{eq:ht_var} to the adjusted potential outcomes implies that the variance of the covariate-adjusted HT estimator is
\begin{align}
\label{eq:fixed-reg-adj-ht-var}
\frac{1}{n^2} \norm{\bmu - \bR^{-1} \bX \bb}{2}^2,
\end{align}
where $\bR$ is a diagonal matrix associated with the vector $\br \in \R^{n}$ defined by $r_i = \sqrt{p_i (1 - p_i)}$. Let $\widetilde{\bX} = \bR^{-1} \bX$. Then the variance in \eqref{eq:fixed-reg-adj-ht-var} is minimized by the choice of coefficients
\begin{align}
\label{eq:fixed-reg-adj-ht-var-prob}
\bbeta^* = \argmin_{\bb \in \R^k} \norm{\bmu - \widetilde{\bX} \bb}{2}^2.
\end{align}

Equation \eqref{eq:fixed-reg-adj-ht-var} implies that regression adjustment can reduce the variance of the HT estimator.  
In practice, however, only one potential outcome is observed for each unit, so $\bmu$ cannot be computed directly, and the minimization problem in \eqref{eq:fixed-reg-adj-ht-var-prob} cannot be solved exactly.
We address this challenge in LOORA-HT by replacing $\bmu$ with a proxy vector that is equal to $\bmu$ in expectation.

\subsection{Leave-One-Out Regression-Adjusted Horvitz-Thompson}
\label{sec:loora-ht-estimator}

\begin{algorithm}[t] 
\caption{LOORA-HT estimator for binary treatment experiments}
\label{alg:ht-loo-reg-adj}
\begin{algorithmic}[1]
\State {\bfseries Input:} Covariates $\bX \in \R^{n \times k}$, outcome vector $\by \in \R^{n}$, treatment assignment vector $\bd \in \{0,1\}^{n}$, vector of treatment assignment probabilities $\bp \in (0,1)^n$, regularization factor $\lambda \geq 0$.

\State \label{alg-step:loo-ht-def-r} Calculate $\br \in \R^{n}$ with $r_i = \sqrt{p_i(1-p_i)}$ and  $\widetilde{\bX} = \bR^{-1} \bX$.

\State \label{alg-step:loo-ht-def-q} Calculate $\bq \in \R^{n}$ with 
\[
q_i =\left\{\begin{array}{cl}p_i& \mbox{ if } d_i=1,\\
1-p_i& \mbox{ if } d_i=0.
\end{array}\right.
\]

\State \label{alg-step:loo-ht-def-d}
Calculate $\bz = 2(\bd -1/2)\in \R^{n}$.

\State \label{alg-step:loo-ht-def-y-tild} Calculate $\widetilde{\by}  \in \R^{n}$ with
\[
\widetilde y_i = \frac{1}{q_i}\left(\frac{1-p_i}{p_i}\right)^{z_i/2} y_i.
\]

\State \label{alg-step:loo-ht-init-est} Set $S = 0$.

\For{$i \in [n]$} \label{alg-step:loo-ht-loop}
\State \label{alg-step:loo-ht-reg-solve} Set 
\begin{align*}
\widehat{\bbeta}_{\lambda}^{(-i)} = \argmin_{\bb\in \mathbb R^k} \norm{\widetilde{\by}_{-i}-\widetilde{\bX}_{-i} \bb}{2}^2 + \lambda\norm{\bb}{2}^2.
\end{align*}

\State \label{alg-step:loo-ht-update-est} Set $$S = S + \frac{z_i}{q_i}  (y_i - \bx_{i}^\top \widehat{\bbeta}_{\lambda}^{(-i)}).$$
\EndFor\\

\Return \label{alg-step:loo-ht-return-output} $$\widehat{\tau}_{\textup{LHT}}= S/n.$$
\end{algorithmic}
\end{algorithm}

\Cref{alg:ht-loo-reg-adj} describes LOORA-HT.
Rather than using a single value for $\bb$, LOORA-HT applies a different vector of regression coefficients, $\widehat{\bbeta}_{\lambda}^{(-i)}$, to each unit $i$ in the sample.
LOORA-HT constructs its adjustment vectors using a proxy vector $\widetilde{\by}$ that coincides with $\bmu$ in expectation,
\[
\widetilde y_i =
\begin{cases}
\dfrac{(1 - p_i)^{1/2}}{p_i^{3/2}}\, y_i, & \text{if } d_i = 1, \\[1em]
\dfrac{p_i^{1/2}}{(1 - p_i)^{3/2}}\, y_i, & \text{if } d_i = 0.
\end{cases}
\]
LOORA-HT is defined as
\[
\widehat{\tau}_{\textup{LHT}} = \frac{1}{n}\sum_{i=1}^n \frac{z_i}{q_i}  (y_i - \bx_{i}^\top \widehat{\bbeta}_{\lambda}^{(-i)}),
\]
where $\widehat{\bbeta}_{\lambda}^{(-i)}$ is obtained from a regression on the leave-one-out data, $(\widetilde{\by}_{-i}, \widetilde{\bX}_{-i})$.  
In general, leave-one-out coefficients can be unstable, particularly for rows with high leverage scores (see \Cref{cor:sum-sqr-err-with-removal}).  
To address this challenge, we use ridge regression to mitigate the impact of high leverage scores.  
In \Cref{lem:ridge-lev-upper-bound} below, we provide an upper bound on the ridge leverage scores as a function of the ridge regularization parameter.

\subsection{The Variance of LOORA-HT}
\label{sec:loora-ht-var}

Let the ridge projection matrix of $\widetilde{\bX}$ be
\[
\widetilde{\bH}_\lambda
= \widetilde{\bX} (\widetilde{\bX}^\top \widetilde{\bX} + \lambda \bI)^{-1} \widetilde{\bX}^\top.
\]
We denote its $(i,j)^{\text{th}}$ element by $\widetilde{h}_{\lambda ij}$ and its $i^{\text{th}}$ diagonal element by $\widetilde{h}_{\lambda ii}$.  
Next, we define the vector
\[
\bt = \frac{(1 - \bp)^2 \by^{(1)} - \bp^2 \by^{(0)}}{\br},
\]
which characterizes how far $\widetilde{\by}$ deviates from $\bmu$,
\[
\widetilde{\by} - \bmu = \frac{\bz \bt}{\bq}.
\]
The ridge regression coefficient on $(\bmu, \widetilde{\bX})$ is
\begin{align}
\label{eq:real_lambda_ridge}
\bbeta_\lambda
= \argmin_{\bb \in \R^k} \norm{\bmu - \widetilde{\bX} \bb}{2}^2 + \lambda \norm{\bb}{2}^2.
\end{align}
The next theorem establishes the unbiasedness of LOORA-HT and provides an exact expression for its variance.
\begin{theorem}
\label{thm:HT-LOO-Binary-Variance}
Under simple random assignment, LOORA-HT (\Cref{alg:ht-loo-reg-adj}) is exactly unbiased, with variance equal to 
\begin{align}
\label{eq:reg-adj-ht-var}
\frac{1}{n^2}  \sum_{i=1}^n \frac{(\widetilde{\bx}_{i}^\top \bbeta_\lambda-\mu_i)^2}{(1-\widetilde{h}_{\lambda ii})^2} + \frac{1}{n^2}\sum_{i=1}^{n-1} \sum_{j = i+1}^n (\widetilde{h}_{\lambda ij})^2 \left(\frac{ t_j}{r_j(1-\widetilde{h}_{\lambda ii})} + \frac{ t_i}{r_i(1-\widetilde{h}_{\lambda jj})}\right)^2.
\end{align}
\end{theorem}

The first term in \eqref{eq:reg-adj-ht-var} mirrors the variance of the infeasible regression-adjusted version HT estimator in \eqref{eq:fixed-reg-adj-ht-var}, and quantifies the variance reduction achieved by regression adjustment.  
The factor $(1 - \widetilde{h}_{\lambda ii})^{-2}$ arises from the removal of a single row in the leave-one-out regressions.  
More specifically, by \Cref{cor:sum-sqr-err-with-removal}, the first term of \eqref{eq:reg-adj-ht-var} can be written as
\[
\frac{1}{n^2}\sum_{i=1}^n (\widetilde{\bx}_i^\top \bbeta_\lambda^{(-i)} - \mu_i)^2,
\]
where
\[
\bbeta_{\lambda}^{(-i)} = \argmin_{\bb\in \mathbb R^k} \norm{\bmu_{-i}-\widetilde{\bX}_{-i} \bb}{2}^2 + \lambda\norm{\bb}{2}^2.
\]

When a leverage score $\widetilde{h}_{\lambda ii}$ approaches one, the variance becomes unbounded.  
In such high-leverage cases, it is essential to employ ridge regression with a sufficiently large regularization parameter~$\lambda$.  
The following lemma provides an upper bound on the ridge leverage scores.

\begin{lemma}
\label{lem:ridge-lev-upper-bound}
Let $\bX \in \R^{n \times k}$, $c \ge 0$, and $\lambda = c \norm{\bX}{2,\infty}^2$.  
Then, for all $i = 1, \ldots, n$,
\[
h_{\lambda ii} \le \frac{1}{1 + c}.
\]
\end{lemma}

High-leverage observations distort regression-adjusted estimates.
LOORA applies ridge regularization, which stabilizes the adjustment and delivers more reliable estimates.

The second term of \eqref{eq:reg-adj-ht-var} arises from using the random vector $\widetilde{\by}$, with $\E[\widetilde{\by}]=\bmu$, in place of $\bmu$ in the regression adjustment.  
In other words, it captures the error introduced by estimating $\bbeta_{\lambda}^{(-i)}$ with $\widehat{\bbeta}_{\lambda}^{(-i)}$.

We next provide a loose upper bound for this term to show that it scales as $k / n^2$.  
Let 
\[
\widetilde{\bX}_{\lambda} =
\begin{bmatrix}
\widetilde{\bX} \\[2pt]
\sqrt{\lambda}\, \bI
\end{bmatrix}.
\]
It follows that 
\begin{align*}
\sum_{i=1}^{n-1} \sum_{j=i+1}^n \widetilde{h}_{\lambda ij}^2
&\le \frac{1}{2} \norm{\widetilde{\bH}_\lambda}{F}^2
\le \frac{1}{2} \norm{\widetilde{\bX}_{\lambda} (\widetilde{\bX}_{\lambda}^\top \widetilde{\bX}_{\lambda})^{-1} \widetilde{\bX}_{\lambda}^\top}{F}^2
= \frac{k}{2},
\end{align*}
where the second inequality holds because $\widetilde{\bH}_\lambda$ is an $(n \times n)$ block of the projection matrix associated with $\widetilde{\bX}_{\lambda}$, and the final equality holds since the rank of $\widetilde{\bX}_{\lambda}$ is $k$, and the trace of a projection matrix equals its rank.

Therefore, the second term of \eqref{eq:reg-adj-ht-var} is bounded by
\begin{equation*}
\frac{2k}{n^2}
\norm{(1 - \widetilde{\bh}_\lambda)^{-1}}{\infty}^2
\norm{\bt / \br}{\infty}^2,
\end{equation*}
where $\widetilde{\bh}_\lambda$ denotes the vector of diagonal entries of $\widetilde{\bH}_{\lambda}$.  
That is, for a fixed number of covariates, the second term of \eqref{eq:reg-adj-ht-var} scales as $1 / n^2$ as $n$ increases, and is typically much smaller than the first term, which scales as $1 / n$.

In particular, for $\lambda = \norm{\widetilde{\bX}}{2,\infty}^2$, by \Cref{lem:ridge-lev-upper-bound}, the variance of LOORA-HT is bounded by
\[
\frac{4}{n^2}\sum_{i=1}^n (\widetilde{\bx}_{i}^\top \bbeta_\lambda - \mu_i)^2 + \frac{8k}{n^2} \norm{\bt/\br}{\infty}^2\,.
\]

We next study the asymptotic behavior of the LOORA-HT estimator.

\subsection{Asymptotic Normality of LOORA-HT}
\label{sec:asymptotic-ht}

This section derives a large-sample approximation to the distribution of LOORA-HT under the following assumptions.
\smallskip

\Crefname{enumi}{Assumption}{Assumptions}
\begin{enumerate}[label=\textbf{Assumption~\arabic*}, ref=\arabic*, leftmargin=*,
  align=left,
  labelwidth=!,
  labelsep=0.6em]
\item \label{ass:boundedness}
\textbf{(uniform boundedness).}
There exists a finite constant $L < \infty$ such that, for all $i \in \mathbb{N}$,
  $\norm{\bx_i}{2} \le L$, 
  $|y_i^{(1)}|\le L$, and $|y_i^{(0)}| \le L$.
\item \label{ass:positivity}
\textbf{(positivity).}
There exists a constant $m \in (0,0.5]$ such that for simple random assignment $m < p_i < 1-m$, for all $i \in \mathbb{N}$; and for complete random assignment $m < n_T/n < 1- m$ for all $n \in \mathbb{N}$.

\item \label{ass:conv-2nd-moment}
\textbf{(positive definiteness).}
The smallest eigenvalue of
$n^{-1}\bX^{\top}\bX$ is bounded away from zero uniformly in $n$.
\end{enumerate}

Assumptions \ref{ass:boundedness} and \ref{ass:conv-2nd-moment} provide bounds on potential outcomes and covariates, but allow the means of $(y_i^{(1)}, y_i^{(0)}, x_i)$ to fluctuate, without necessarily converging to a limit as $n$ increases. Assumption \ref{ass:positivity} requires that each experimental unit be assigned to treatment and control with strictly positive probabilities. It rules out degenerate experimental designs in which causal effects cannot be identified.

Let 
\[
\bbeta_{\lambda}
  = (\widetilde{\bX}^\top \widetilde{\bX} + \lambda \bI)^{-1}
    \widetilde{\bX}^\top \bmu
\]
be the true solution of the ridge regression problem on $\bmu$. 
The next theorem establishes the large-sample distribution of the
LOORA-HT estimator and identifies its asymptotic variance in closed form. 

\begin{theorem}\label{thm:asym_dist_HT_wo_lim}
Suppose \Cref{ass:boundedness,ass:positivity,ass:conv-2nd-moment} hold.
Let $\lambda \ge 0$ be fixed.
Then, if $\|\widetilde{\bX} \bbeta_{\lambda} - \bmu\|_{2} \to \infty$, under simple random assignment,
\[
\frac{\sqrt{n} (\hatt_{\textup{LHT}} - \tau)}{
\norm{\widetilde{\bX} \bbeta_{\lambda} - \bmu}{2} / \sqrt{n}}
\]
converges in distribution to a  standard normal random variable.
\end{theorem}

In the next subsection, we build on Theorem~\ref{thm:asym_dist_HT_wo_lim} to analyze the
asymptotic efficiency of LOORA-HT.

\subsection{Asymptotic Efficiency of LOORA-HT}
\label{sec:loora-ht-efficiency}

We now examine the asymptotic efficiency of the LOORA-HT estimator.  
Theorem~\ref{thm:asym_dist_HT_wo_lim} established the asymptotic normality of LOORA-HT.
In this subsection, we show that under simple random assignment, LOORA-HT is asymptotically efficient among the class of Horvitz-Thompson linearly adjusted estimators, defined as estimators of the form,
\[
\hatt_{\textup{HLA}} = \frac{1}{n} \sum_{i=1}^n \frac{z_i}{q_i} \bigl(y_i - \bx_i^\top \hat{\bgamma}_{n}^{(-i)}\bigr),
\]
where there exists a bounded deterministic sequence of vectors $\bgamma_1,\bgamma_2,\ldots$ and a sequence of vectors $\hat{\bgamma}_1,\hat{\bgamma}_2,\ldots$ such that
\begin{align}
\label{eq:HLT-assumption}
\max_{i \in [n]}\norm{\hat{\bgamma}_{n}^{(-i)} - \hat{\bgamma}_n}{2} = o_p(n^{-1/2})
\quad \text{and} \quad
\norm{\bgamma_n - \hat{\bgamma}_n}{2} \pto 0.
\end{align}
The class of HT linearly adjusted estimators is larger than the class of estimators usually considered for the purpose of determining efficiency (e.g., see \cite{cytrynbaum2024covariate}), because it allows for unit-dependent adjustment vectors, $\hat{\bgamma}_{n}^{(-i)}$.

LOORA-HT is contained in the class of HT linearly adjusted estimators since
\begin{align}
\label{eq:loora-ht-properties}
\max_{i \in [n]} \norm{\widehat{\bbeta}_{\lambda}^{(-i)} - \widehat{\bbeta}_{\lambda}}{2} = O(n^{-1})
\quad \text{and} \quad
\norm{\widehat{\bbeta}_\lambda - \bbeta_\lambda}{2} = O_p(n^{-1/2}),
\end{align}
where
\[
\widehat{\bbeta}_{\lambda}
  = (\widetilde{\bX}^\top \widetilde{\bX} + \lambda \bI)^{-1}
    \widetilde{\bX}^\top \widetilde{\by},
\]
as shown in the proof of \Cref{thm:asym_dist_HT_wo_lim}. The next result describes the asymptotic behavior of HT linearly adjusted estimators. 
\begin{lemma}
\label{lemma:var-of-HLA}
Suppose \Cref{ass:boundedness,ass:positivity,ass:conv-2nd-moment} hold and $\norm{\widetilde{\bX} \bgamma_{n} - \bmu}{2} = \Omega(\sqrt{n})$. Then under simple random assignment, $\hatt_{\textup{HLA}}$ is a consistent estimator and
\[
\frac{\sqrt{n}(\hatt_{\textup{HLA}} - \tau)}{\frac{1}{\sqrt{n}}\norm{\widetilde{\bX} \bgamma_{n} - \bmu}{2}} \dto \mathcal{N}(0,1).
\]
\end{lemma}
Using this result, we can establish the asymptotic efficiency of LOORA-HT in the case where $\norm{\widetilde{\bX} \bbeta_{0} - \bmu}{2} = \Omega(\sqrt{n})$.
\begin{theorem}
\label{thm:loora-ht-efficiency}
Under the assumptions of \Cref{lemma:var-of-HLA} and simple randomization, 
LOORA-HT with $\lambda = 0$ is an asymptotically efficient estimator in the class of HT linearly adjusted estimators.
\end{theorem}

For the case of $\lambda=0$ and provided that the regression adjustment includes an intercept, by \Cref{lemma:recenter}, the variance of LOORA-HT reduces to
\[
\min_{\bbeta \in \R^k} \frac{1}{n^2} \norm{(\widetilde{\bX} - \overline{\widetilde{\bX}})\bbeta - (\bmu - \overline{\bmu})}{2}^2,
\]
 This expression is similar to the variance of the estimator of \cite{lin2013agnostic}. 
  
In the next subsection, we develop valid and asymptotically
correct confidence intervals for LOORA-HT.

\subsection{Confidence Intervals for LOORA-HT}
\label{sec:loora-ht-conf-int}

The next theorem presents a variance estimator for $\widehat{\tau}_{\textup{LHT}}$. 

\begin{theorem}
\label{thm:loora-ht-est-var}
Suppose \Cref{ass:boundedness,ass:positivity,ass:conv-2nd-moment} hold.
Let $\lambda \ge 0$ be fixed,
\begin{align*}
\widehat{s}_i &
  = z_i \left(\frac{y_i - \bx_i^\top \widehat{\bbeta}_{\lambda}^{(-i)}}
     {q_i} \right) - \,\widehat{\tau}_{\textup{LHT}}, ~~\
     \widehat{V}_{\textup{LHT}} = \frac{1}{n} \sum_{i=1}^n \widehat{s}_i^2, ~~ \text{and} 
     \\ 
     V_n & = \frac{1}{n} \left(\norm{\widetilde{\bX} \bbeta_{\lambda} - \bmu}{2}^2 + \sum_{i=1}^n (\tau_i - \tau)^2  \right),
\end{align*}
where $\tau_i = y^{(1)}_i - y^{(0)}_i$. If $V_n$ is bounded away from zero,
then
\[
\frac{\widehat{V}_{\textup{LHT}}}{V_n} \stackrel{p}{\longrightarrow} 1.
\]  
\end{theorem}

As is generally the case in all finite-population inference \citep[see, e.g.,][]{abadie2020sampling}, $\widehat{V}_{\textup{LHT}}$ is consistent under homogeneous individual treatment effects and conservative when treatment effects are heterogeneous. We construct confidence intervals using the estimate of standard deviation as $\sqrt{\widehat{V}_{\textup{LHT}}/n}$ and the appropriate quantiles of the standard normal distribution.

\section{Regression Adjustment for Difference-in-Means Estimator}
\label{sec:reg-adj-dim}

This section considers the difference-in-means (DM) estimator under complete random
assignment, where exactly $n_T$ units are assigned to treatment and
$n_C = n - n_T$ to control. 

The DM estimator
\[
\widehat{\tau}_{\textup{DM}}
  \;=\; \frac{1}{n_T}\sum_{i=1}^n d_i y_i
       \;-\; \frac{1}{n_C}\sum_{i=1}^n (1-d_i) y_i
\]
is unbiased for the average treatment effect under complete random assignment.
Its finite-population variance is given by
\[
\Var(\widehat{\tau}_{\textup{DM}})
  \;=\;
  \frac{S_n^2(y^{(1)})}{n_T}
  \;+\;
  \frac{S_n^2(y^{(0)})}{n_C}
  \;-\;
  \frac{S_n^2(\tau)}{n},
\]
where $S_n^2(\cdot)$ denotes finite-population variances computed with an $n-1$ divisor (see, e.g., \citet[Chapter 6, Appendix A]{imbens2015causal}).
Equivalently—and most convenient for our regression-adjustment analysis—this
variance admits the representation
\begin{align*}
\Var\!\big(\widehat{\tau}_{\textup{DM}}\big)
\;=\;
\frac{1}{\,n\,(n-1)\,}\;
\big\|\, \widetilde{\bmu} - \overline{\widetilde{\bmu}} \,\big\|_2^2,
\end{align*}
where
\begin{equation}
\label{equation:mutilde}
\widetilde{\bmu} \;=\; \sqrt{\frac{n_C}{n_T}}\,\by^{(1)} \;+\; \sqrt{\frac{n_T}{n_C}}\,\by^{(0)}.
\end{equation}

For a fixed vector $\bb \in \mathbb{R}^k$, the covariate-adjusted DM estimator
\[
\widehat{\tau}_{\textup{DM}}(\bb) = 
\frac{1}{n_T}\sum_{i=1}^n d_i\,(y_i - \eta^{-1}\bx_i^\top\bb)
\;-\;
\frac{1}{n_C}\sum_{i=1}^n (1-d_i)\,(y_i - \eta^{-1}\bx_i^\top\bb),
\]
where $\eta=\sqrt{n_C/n_T} + \sqrt{n_T/n_C}$,
is unbiased for $\tau$, and its variance is equal to
\begin{equation}\label{eq:dm-adj-fixed-bb-final}
\Var\!\big(\widehat{\tau}_{\textup{DM}}(\bb)\big)
\;=\;
\frac{1}{\,n\,(n-1)\,}\;
\big\|\,
\big(\widetilde{\bmu}-\bX\bb\big)
-
\overline{\big(\widetilde{\bmu}-\bX\bb\big)}
\,\big\|_2^2,
\end{equation}
with
$\overline{(\widetilde{\bmu}-\bX\bb)} = (n^{-1}\sum_{i=1}^n(\widetilde{\mu}_i - \bx_i^\top\bb))\boldsymbol{1}$ .
Therefore, the minimum variance of 
$\widehat{\tau}_{\textup{DM}}(\bb)$
is equal to
\begin{equation}\label{eq:dm-projection-final}
\frac{1}{\,n\,(n-1)\,}\;
\min_{\bb\in\mathbb{R}^k}
\norm{\,
(\bX-\overline{\bX})\,\bb
-
(\widetilde{\bmu}-\overline{\widetilde{\bmu}})
\,}{2}^2.
\end{equation}

The vector $\widetilde{\bmu}$ is not observed, so the variance in \eqref{eq:dm-projection-final} cannot be attained in practice. The next section introduces an estimable surrogate
and uses it to define a leave-one-out regression-adjusted DM estimator (LOORA-DM).

\subsection{Leave-One-Out Regression-Adjusted Difference-in-Means}
\label{sec:loora-dm}

Similar to LOORA-HT, our proposed leave-one-out regression-adjusted difference-in-means (LOORA-DM) estimator (\Cref{alg:dim-loo-reg-adj}) computes a distinct coefficient vector $\bb$ for each unit $i$. A key difference, however, is the dependent variable used in the regression to estimate the adjustment coefficients. This regression must be modified to account for the dependence in treatment assignments induced by complete randomization.

\begin{algorithm}[t] 
\caption{LOORA-DM estimator}
\label{alg:dim-loo-reg-adj}
\begin{algorithmic}[1]
\State {\bfseries Input:} Covariates $\bX \in \R^{n \times k}$, outcome vector $\by \in \R^{n}$, treatment assignment vector $\bd \in \{0,1\}^{n}$ with $n_T = \sum_{i=1}^n d_i$ and $n_C = n-n_T$,
regularization factor $\lambda \geq 0$.

\State Calculate $\bv$ with $v_i = 1/n_T$ if $d_i=1$, and $v_i = 1/n_C$, otherwise.

\State Calculate $\bf_T \in \R^{n}$ with $f_{Tj} = \frac{1}{n_T(n_T - 1)} $ if $d_j=  1$, and $f_{Tj} = \frac{1}{n^2_C}$ otherwise.

\State Calculate $\bf_C \in \R^{n}$ with $f_{Cj} = \frac{1}{n_T^2} $ if $d_j= 1$, and $f_{Cj} = \frac{1}{n_C(n_C - 1)}$  otherwise.

\State Calculate $\bz = 2(\bd -1/2)\in \mathbb R^n$.

\State Set $S = 0$.

\For{$i \in [n]$}

\If{$d_i = 1$} 

\State Let $\bf^{(-i)} = \bf_T$.
\Else
\State Let $\bf^{(-i)} = \bf_C$.
\EndIf

\State Set $$\widetilde{\by}^{(-i)} = \frac{n_Tn_C(n-1)}{n}\bf^{(-i)} \by.$$

\State Set
\begin{align*}
\widehat{\bbeta}^{(-i)}_\lambda = \argmin_{\bb} \norm{\bX_{-i} \bb - \widetilde{\by}^{(-i)}_{-i}}{2}^2 + \lambda\norm{\bb}{2}^2\,.
\end{align*}

\State Set $$S = S + v_i z_i (y_i - \bx_{i}^\top \widehat{\bbeta}^{(-i)}_\lambda).$$
\EndFor\\

\Return $$\widehat{\tau}_{\textup{LDM}}= S.$$
\end{algorithmic}
\end{algorithm}

For each unit $i$, we define
\[
\widetilde{\by}^{(-i)} \;=\; \frac{n_T n_C (n-1)}{n}\,\bf^{(-i)} \by,
\]
where the components of $\bf^{(-i)}$ are given by
\[
f^{(-i)}_j =
\begin{cases}
\dfrac{1}{n_T (n_T - 1)} & \text{if } d_j = 1, \\[1em]
\dfrac{1}{n_C^2} & \text{if } d_j = 0,
\end{cases}
\]
if $d_i=1$, and
\[
f^{(-i)}_j =
\begin{cases}
\dfrac{1}{n_T^2} & \text{if } d_j = 1, \\[1em]
\dfrac{1}{n_C (n_C - 1)} & \text{if } d_j = 0,
\end{cases}
\]
if $d_i=0$.

The coefficient vector $\widehat{\bbeta}^{(-i)}_{\lambda}$ is then obtained by
regressing $\widetilde{\by}^{(-i)}_{-i}$ on $\bX_{-i}$ using ridge regression.
Finally, the outcome of unit $i$ is adjusted by
$\bx_i^\top \widehat{\bbeta}^{(-i)}_{\lambda}$,
\[
\widehat{\tau}_{\textup{LDM}}=\frac{1}{n_T}\sum_{i=1}^n d_i\,(y_i - \bx_{i}^\top \widehat{\bbeta}^{(-i)}_\lambda)
\;-\;
\frac{1}{n_C}\sum_{i=1}^n (1-d_i)\,(y_i - \bx_{i}^\top \widehat{\bbeta}^{(-i)}_\lambda).
\]

\subsection{The Bias and Variance of LOORA-DM}
\label{sec:loora-dm-var}

This section establishes that LOORA-DM is exactly unbiased for the average
treatment effect and derives an exact expression for its finite-sample variance.
The expression of the variance is more intricate for LOORA-DM than for LOORA-HT because complete randomization induces dependence between treatment assignments. Unbiasedness is also
slightly more delicate to verify than in the LOORA-HT case, but it
ultimately follows from the fact that $\widetilde{\by}^{(-i)}$ is constructed so
that $\E[\widetilde{\by}^{(-i)}_{-i}]=(\sqrt{n_c n_T}/n)\widetilde{\bmu}_{-i}$ for each $i$.

\begin{theorem}\label{thm:LOORA_DiM_Var}
Under complete random assignment, LOORA-DM is unbiased with variance equal to 
\begin{align*}
& \frac{1}{n(n-1)} \sum_{i=1}^n \left( \frac{\widetilde{\mu}_i-\bx_i^\top\bbeta_{\lambda}}{1- h_{\lambda ii}} - \frac{1}{n}\sum_{j=1}^n \frac{\widetilde{\mu}_j-\bx_j^\top\bbeta_{\lambda}}{1- h_{\lambda jj}} \right)^2
\\ &
\quad - \frac{2}{n^2(n-1)} \sum_{\substack{i,j \in [n]:\\ i \neq j}} \left( \frac{(\widetilde{\mu}_i-\bx_i^\top\bbeta_{\lambda})\, h_{\lambda ij}\, \widetilde{\mu}_i}{(1- h_{\lambda ii})(1- h_{\lambda jj})} - \frac{1}{n-2} \sum_{\substack{k \in [n]: \\k\neq i,j}}  \frac{(\widetilde{\mu}_i-\bx_i^\top\bbeta_{\lambda})\, h_{\lambda jk}\,  \widetilde{\mu}_k}{(1- h_{\lambda ii})(1- h_{\lambda jj})}
\right)
\\ & \quad + \begin{bmatrix}
    \widetilde\bt^{(0)} \\ \widetilde\bt^{(1)}
\end{bmatrix}^{\!\top} \bQ
\begin{bmatrix}
    \widetilde\bt^{(0)} \\ \widetilde\bt^{(1)}
\end{bmatrix},
\end{align*}
where
\begin{align*}
\bbeta_{\lambda} &=(\bX^\top \bX + \lambda \bI)^{-1}\bX^\top \widetilde{\bmu}, \\
\widetilde\bt^{(1)} & = n_C^2\by^{(1)} - n_T(n_T-1)\by^{(0)}, \qquad
\widetilde\bt^{(0)} = n_C(n_C-1)\by^{(1)} - n_T^2\by^{(0)}.
\end{align*}
The entries of $\bQ$ depend only on $n_C$, $n_T$, and the entries of the hat matrix
$\bH_{\lambda}$; its explicit form is provided in the appendix.
\end{theorem}
The first line in the variance expression is the finite-population variance
of $(\widetilde{\mu}_i-\bx_i^\top\bbeta_{\lambda})/(1-h_{\lambda ii})$ rescaled by $1/n$.
The
 $1-h_{\lambda ii}$ in the denominator provides a leverage correction that 
generates from leave-one-out adjustment. 
The second line of the variance expression in \Cref{thm:LOORA_DiM_Var} collects cross-unit terms that include the off-diagonal leverages
$h_{\lambda ij}$. These cross-unit contributions did not appear in the
LOORA-HT analysis because simple random assignment makes assignments independent across units, whereas complete random assignment induces dependence across units. The matrix $\bQ$
bundles additional corrections that depend only on $(n_T,n_C)$ and the leverage
structure encoded in $\bH_{\lambda}$. For a fixed number of covariates, the
first term scales as $O(1/n)$, while both the cross-unit term and the $\bQ$-term
scale as $O(1/n^2)$.\footnote{To see why the cross-unit term in the result of Theorem \ref{thm:LOORA_DiM_Var} is $O(1/n^2)$, apply Cauchy-Schwarz inequality and note that $(\sum_{\substack{i,j \in [n]: i \neq j}} h_{\lambda ij}^2 )^{0.5} \leq \|\bH\|_{F} \leq \sqrt{k}$.}

Theorem~\ref{thm:LOORA_DiM_Var} shows that the variance of
LOORA-DM is dominated by the centered dispersion of the leverage-adjusted signal
and that dependence induced by complete random assignment contributes only
second-order terms. This decomposition is useful to establish
asymptotic distributional results and efficiency comparisons.

\subsection{Asymptotic Normality of LOORA-DM}
\label{sec:asymptotic-var-LDM}

This section analyzes  the large-sample behavior of LOORA-DM under
complete random assignment.  
It shows that, in large samples, LOORA-DM satisfies a central limit theorem.
Similar to \Cref{sec:asymptotic-ht}, we study the behavior of truncated sequences $(\bx_i, y_i^{(1)}, y_i^{(0)})_{i \in [n]}$ as $n \to \infty$, and assume that the ratio $n_T/n$ is bounded away from zero and one; see \Cref{ass:positivity}.

Recall that for each unit $i$, the adjusted outcome is constructed using
a coefficient vector estimated from all other units:
\begin{align*}
\widehat{\bbeta}^{(-i)}_\lambda
&=
\argmin_{\bb \in \R^k}
\big\|
\bX_{-i}\bb - \widetilde{\by}^{(-i)}_{-i}
\big\|_2^2
+
\lambda \|\bb\|_2^2.
\end{align*}

To study the limiting behavior of these regression adjustments, define the following
population version of 
$\widehat{\bbeta}^{(-i)}_\lambda$
\begin{align}
\label{eq:beta-star-def-dm}
\bbeta_{\lambda}
=
\argmin_{\bb \in \R^{k}}
\left\|
\bX\bb - \widetilde{\bmu}
\right\|_2^2
+
\lambda \|\bb\|_2^2,
\end{align}
$\bbeta_{\lambda}$ represents the deterministic coefficient vector that best
captures the relation between covariates and the signal component
$\widetilde{\bmu}$ in the population.  
\begin{theorem}
\label{thm:asymptotic-normality-LDM}
Suppose \Cref{ass:boundedness,ass:positivity,ass:conv-2nd-moment} hold.
Let $\lambda \ge 0$ be fixed.
If 
\[
\norm{(\bX - \overline{\bX}) \bbeta_{\lambda}
\allowbreak - \allowbreak
(\widetilde{\bmu} - \overline{\widetilde{\bmu}})}{2} \to \infty,
\]
then
\[
\frac{\sqrt{n} (\hatt_{\textup{LDM}} - \tau)}{
\norm{(\bX - \overline{\bX}) \bbeta_{\lambda}
-
(\widetilde{\bmu} - \overline{\widetilde{\bmu}})}{2} / \sqrt{n}} 
\]
converges in distribution to a  standard normal random variable.
\end{theorem}

The variance expression above parallels that of the LOORA-HT estimator
but differs in two key aspects.
First, the large sample variance of the difference-in-means estimator under complete random assignment automatically recenters $\bX$ and $\widetilde{\bmu}$. In contrast, LOORA-HT requires the inclusion of an intercept term to obtain a similar variance expression.
Second, the scaling of $\widetilde{\bmu}$ reflects the finite-sample proportions of
treated and control units (equation \eqref{equation:mutilde}).

Theorem~\ref{thm:asymptotic-normality-LDM} shows that LOORA-DM achieves
an asymptotically normal distribution with a variance determined by the
projection of the population signal onto the null space of the covariate matrix.
It implies that LOORA-DM is as efficient in large samples as the interacted adjustment estimator of \cite{lin2013agnostic}.

\subsection{Confidence Intervals for LOORA-DM}
\label{sec:ci-ldm}

The next theorem provides a consistency result for an estimator of the variance of LOORA-DM. 

\begin{theorem}
\label{thm:loora-dm-est-var}
Suppose \Cref{ass:boundedness,ass:positivity,ass:conv-2nd-moment} hold.
Let $\lambda \ge 0$ be fixed, define 
\[
\widehat{V}_{\textup{LDM}} = \frac{n-1}{n_C (n_C - 1)} \sum_{i \in [n]: d_i=0} (\widehat{s}_i^{(0)})^2 + \frac{n-1}{n_T(n_T-1)} \sum_{i \in [n]: d_i=1} (\widehat{s}_i^{(1)})^2,
\]
where
\begin{align*}
\widehat{s}_i^{(0)} & = y_i - \bx_i^\top\hat\bbeta_{\lambda}^{(-i)} - \frac{1}{n_C} \sum_{j \in [n]:d_j=0} y_j - \bx_j^\top\hat\bbeta_{\lambda}^{(-j)},
\\
\widehat{s}_i^{(1)} & = y_i - \bx_i^\top\hat\bbeta_{\lambda}^{(-i)} - \frac{1}{n_T} \sum_{j \in [n]:d_j=1} y_j - \bx_j^\top\hat\bbeta_{\lambda}^{(-j)}.
\end{align*}
Moreover let
\[
W_n = \frac{1}{n} \left(\norm{(\bX - \overline{\bX}) \bbeta_{\lambda}
-
(\widetilde{\bmu} - \overline{\widetilde{\bmu}})}{2}^2 + \sum_{i=1}^n (\tau_i - \tau)^2  \right).
\]
If $W_n$ is bounded away from zero, then
\[
\frac{\widehat{V}_{\textup{LDM}}}{W_n} \stackrel{p}{\longrightarrow} 1.
\]
\end{theorem}
As with LOORA-HT, and more generally with finite-population variance estimation for treatment-effect estimators, $\widehat{V}_{\textup{LDM}}$ is consistent under treatment-effect homogeneity and conservative otherwise. Large-sample confidence intervals for $\tau$ follow by combining this variance estimator with quantiles of the standard normal distribution.
\section{Discussion}
\label{sec:detailed-related-work}

LOORA targets three properties that existing estimators do not achieve simultaneously: finite-sample unbiasedness, asymptotic efficiency, and robustness to high-leverage observations. Classical estimators such as Horvitz–Thompson and the difference-in-means are exactly unbiased, but they leave efficiency gains on the table. Regression-adjusted estimators recover those gains asymptotically, yet they introduce finite-sample bias and can become unstable in the presence of high-leverage observations. LOORA delivers all three properties simultaneously: it preserves exact unbiasedness, attains the \citet{lin2013agnostic} efficiency bound in large samples, and remains stable under influential covariates through its leave-one-out construction.

Table~\ref{tabel:comapre_estiamtors} summarizes the properties for LOORA and competing estimators under simple and complete randomization.

\vspace{0.3em}
\noindent\textbf{Classical estimators.}
Under simple randomization, the Horvitz--Thompson (HT) estimator is
exactly unbiased but can exhibit high variance.  Regression adjustment aims to improve efficiency by controlling
for observed covariates.  The most commonly used adjusted estimator,
$\widehat\tau_{\text{ADJ}}$, regresses $\by$ on $\bX$ and the treatment
indicator $\bd$ with intercept, producing a consistent estimator only
when treatment effects are homogeneous.  The interactive version,
$\widehat\tau_{\text{INT}}$, adds treatment-by-covariate interactions
and is asymptotically efficient under both simple (with equal treatment probability assignment for all units) and complete
randomization \citep{lin2013agnostic}. 
However, both estimators, $\widehat\tau_{\text{ADJ}}$ and $\widehat\tau_{\text{INT}}$,  can
become unstable when the sample contains 
high-leverage observations, and neither is unbiased.

\vspace{0.3em}
\noindent\textbf{Leave-one-out regression adjustment.}
The LOORA framework modifies regression adjustment to preserve
finite-sample unbiasedness without sacrificing efficiency.  By using
leave-one-out fitted values, the LOORA-HT and LOORA-DM estimators remove
the bias induced by regression adjustment while maintaining
asymptotic optimality.  Under simple randomization, LOORA-HT achieves
the same variance bound as the \citet{lin2013agnostic} estimator but is
exactly unbiased and robust to leverage.  Under complete randomization,
LOORA-DM extends these results to dependent assignment structures by
appropriately rescaling the leave-one-out predictions.  Both estimators
combine the finite-sample exactness of HT and DM with the
asymptotic efficiency of regression adjustment.

\vspace{0.3em}
\noindent\textbf{Relation to recent work.}
Recent work has proposed innovative methods for covariate adjustment in randomized experiments.  
\citet{ghadiri2023finite} develops a cross-fitting unbiased regression-adjusted estimator for simple random assignment. This approach, however, does not attain the \cite{lin2013agnostic} variance in large samples.  
\citet{harshaw2022design}, in contrast, takes a design-based perspective: rather than modifying the estimator, they optimize the randomization mechanism itself to minimize the variance of the unadjusted Horvitz--Thompson estimator.  
Finally, \citet{spiess2025optimal} characterizes the necessary conditions under which design-based estimators such as Horvitz--Thompson or difference-in-means are unbiased.  
For the HT estimator, these conditions require a leave-one-out structure, which our LOORA-HT estimator satisfies by construction.  
For the DM estimator, unbiasedness requires a leave-two-out structure. We show below that LOORA-DM can be represented as a leave-two-out estimator, thereby satisfying Spiess’s condition.  
Hence, our estimators are not only consistent with the theoretical framework of \citet{spiess2025optimal}, but also provide constructive, closed-form realizations of unbiased estimators that satisfy these necessary design-based criteria. The remainder of this section provides a detailed comparison with prior work.

\begin{table}
\caption{Properties of estimators}\label{tabel:comapre_estiamtors}
\begin{tabular}{lcccc}\hline
&&unbiased&attains \cite{lin2013agnostic}&robust to\\
&&&asymptotic variance&high leverage\\
{\it Under simple randomization:}\\[.5ex]
\hspace{2ex} $\widehat\tau_{\text{HT}}$         &&\cmark&\xmark&--\\
\hspace{2ex} $\widehat\tau_{\text{LHT}}$   &&\cmark&\cmark&\cmark\\[1ex]
{\it Under complete randomization:}\\[.5ex]
\hspace{2ex} $\widehat\tau_{\text{DM}}$         &&\cmark&\xmark&--\\
\hspace{2ex} $\widehat\tau_{\text{ADJ}}$\quad(Linear regression adjustment)&        &\xmark&\xmark&\xmark\\
\hspace{2ex} $\widehat\tau_{\text{INT}}$ \quad(Interactive regression adjustment, \cite{lin2013agnostic})        &&\xmark&\cmark&\xmark\\
\hspace{2ex} $\widehat\tau_{\text{LDM}}$   &&\cmark&\cmark&\cmark\\[1ex]
 \hline
\end{tabular}
\end{table}

\subsection{Comparison with Cross-Fitted Regression-Adjusted HT}
\label{sec:compare-crossfit}

\citet{ghadiri2023finite} proposes a finite-sample unbiased regression-adjusted Horvitz–Thomp-\allowbreak son estimator built on cross-fitting. The procedure splits the sample into two groups, learns a regression vector within each group, and then uses each group’s fitted vector to adjust the outcomes in the other. \citet{ghadiri2023finite} considers only the case of constant $p_i=0.5$ for all $i$ and derives a high-probability upper bound on the variance of their estimator. Specifically, conditional on an event $\mathcal{E}$ that holds with probability at least $1-\delta$, the variance of the estimator in \citet{ghadiri2023finite} is bounded by
\begin{align}
\label{eq:cross-fit-reg-adj-ht-var-bound}
\frac{8(1+\epsilon)}{n^2} 
\min_{\bb \in \mathbb{R}^k} 
\Big(
\norm{ \bX \bb - \bmu }{2}^2 
+ 100 \log(n/\delta)\, \zeta_{\bX}^2 \norm{ \bb }{2}^2
\Big)
+ \frac{32k}{n^2} 
\norm{\by^{(1)} - \by^{(0)} }{\infty}^2,
\end{align}
where $\zeta_{\bX}^2 = \norm{\bX}{2, \infty}^2$ and $\varepsilon$ is an approximation error that depends on $\delta$.\footnote{$\mathcal{E}$ is the event under which the leverage-score sampling in \citet{ghadiri2023finite} delivers a $(1+\epsilon)$-approximate solution to the full regression problem. Concretely, the method randomly subsamples rows of $\bX$ and solves the regression on the subsample rather than on all $n$ observations. The resulting approximation error relative to the full-sample regression induces the $(1+\epsilon)$ factor in \eqref{eq:cross-fit-reg-adj-ht-var-bound}}. 
By \Cref{lem:ridge-lev-upper-bound}, \citet{ghadiri2023finite} choice of regularization parameter $\lambda = 100 \log(n/\delta)\, \zeta_{\bX}^2$, the ridge leverage scores satisfy 
$h_{\lambda i i}(\bX) \leq 1/(1 + 100 \log(n/\delta))$.  
Moreover, since all $p_i = 0.5$, it follows that
$
\max_{i \in [n]} \norm{\widetilde{\bx}_i}{2}
= 2 \max_{i \in [n]} \norm{\bx_i}{2}$.
Consider
\begin{align*}
\bbeta
&= \argmin_{\bb \in \mathbb{R}^k}
\Big(
\norm{ \bX \bb - \bmu }{2}^2
+ 100 \log(n/\delta)\, \zeta_{\bX}^2 \norm{ \bb }{2}^2
\Big).
\end{align*}
Note that for any $\bb\in\R^k$,
\[
\norm{ \bX \bb - \bmu }{2}^2
+ 100 \log(n/\delta)\, \zeta_{\bX}^2 \norm{ \bb }{2}^2 = \norm{ \widetilde{\bX} (\bb/2) - \bmu }{2}^2
+ 100 \log(n/\delta)\, \zeta_{\widetilde{\bX}}^2 \norm{ \bb / 2 }{2}^2.
\]
Therefore, 
\[
\argmin_{\bb \in \mathbb{R}^k}
\Big(
\norm{ \widetilde{\bX} \bb - \bmu }{2}^2
+ 100 \log(n/\delta)\, \zeta_{\widetilde{\bX}}^2 \norm{ \bb }{2}^2
\Big) = \bbeta/2.
\]

For  $\lambda = 100 \log(n/\delta)\, \zeta_{\widetilde{\bX}}^2$ with $\delta\leq \exp(-1)$, the variance of LOORA-HT from \Cref{thm:HT-LOO-Binary-Variance}, is bounded by
\begin{align*}
&
\left(1 - \frac{1}{101}\right)^{-2} \left(\frac{1}{n^2}
\sum_{i=1}^n 
(\widetilde{\bx}_i^\top (\bbeta/2) - \mu_i)^2
+ \frac{2k}{n^2} 
\norm{\bt / \br}{\infty}^2 \right) 
\\ &
= 
\left(\frac{101}{100}\right)^{2} \left(\frac{1}{n^2}
\sum_{i=1}^n 
(\bx_i^\top \bbeta - \mu_i)^2
+ \frac{2k}{n^2} 
\norm{\by^{(1)} - \by^{(0)}}{\infty}^2 \right).
\end{align*}
Relative to the bound of \citet{ghadiri2023finite} in \eqref{eq:cross-fit-reg-adj-ht-var-bound}, the bound above improves the first term by roughly a factor of eight. It also eliminates the additional ridge-regularization penalty term that appears in \eqref{eq:cross-fit-reg-adj-ht-var-bound} and improves the final term by about a factor of sixteen.

In addition, unlike the procedure in  \citet{ghadiri2023finite}, LOORA comes with consistent
variance estimation and valid confidence intervals.

\subsection{Comparison with Gram--Schmidt Random Walk Design}
\label{sec:compare-gsw}

\citet{harshaw2022design} proposes an ingenious design-based approach to variance reduction.  
Rather than modifying the estimator, they optimize the randomization mechanism itself through a
Gram-Schmidt Walk (GSW) design that ensures covariate balance while controlling the marginal
assignment probabilities. One potential limitation of GSW is that it requires covariate information on all sample units prior to treatment assignment. In contrast, simple or complete random assignment can be implemented online (sequentially) as additional experimental subjects are recruited for a study. Moreover, with simple or complete randomization, a researcher can freely add or remove covariates,
analyze different subgroups, or draw additional samples even after the experiment
has been conducted.  
This flexibility, together with the transparency and ease of implementation of standard randomization,
makes simple and complete randomization attractive in many applied settings. Motivated by these considerations, LOORA-HT and LOORA-DM are explicitly designed for simple and complete randomization, respectively.

\citet{harshaw2022design} shows that the GSW design achieves the following variance bound for any
choice of $\phi \in (0,1)$:
\begin{align}
\label{eq:gsw-design-var}
\frac{1}{n^2}\min_{\bbeta \in \mathbb{R}^k}
\Big[
\frac{1}{\phi}\norm{\bmu - \bX \bbeta}{2}^2
+ \frac{\zeta_{\bX}^2}{1 - \phi}\norm{\bbeta}{2}^2
\Big].
\end{align}
We now compare this bound with the LOORA-HT variance.
The second term in \eqref{eq:gsw-design-var} and the corresponding regularization term in
\eqref{eq:reg-adj-ht-var} are not directly comparable in general, although both scale on the order of
$k/n^2$.  
For the first term, with a constant probability of assignment,
\[
\min_{\bbeta \in \mathbb{R}^k}
\Big[
\frac{1}{\phi}\norm{\bmu - \bX \bbeta}{2}^2
+ \frac{\zeta_{\bX}^2}{1 - \phi}\norm{\bbeta}{2}^2\Big] = \min_{\bbeta \in \mathbb{R}^k}
\Big[
\frac{1}{\phi}\norm{\bmu - \widetilde{\bX} \bbeta}{2}^2
+ \frac{\zeta_{\widetilde{\bX}}^2}{1 - \phi}\norm{\bbeta}{2}^2\Big].
\]
Therefore,
setting $\lambda = \zeta_{\widetilde{\bX}}^2 / (1 - \phi)$, \Cref{lem:ridge-lev-upper-bound} implies that
\[
\widetilde h_{\lambda ii}
\;\leq\;
\frac{1 - \phi}{2 - \phi}.
\]
Consequently,
\[
\frac{1}{(1 - \widetilde h_{\lambda ii})^2}
\;\leq\;
(2 - \phi)^2.
\]
For $\phi \in [0, (3 - \sqrt{5})/2)$, we have $(2 - \phi)^2 \leq 1/\phi$.  
Hence, the scaling factor that appears in our variance bound is uniformly smaller than that of the
GSW design.  
Moreover, the first term in \eqref{eq:gsw-design-var} diverges as $\phi \to 0$, whereas the
corresponding constant in the variance of LOORA-HT remains bounded.

Finally, when the leverage scores of $\widetilde{\bX}$ are small, LOORA can set $\lambda=0$. In contrast, under the GSW design the smallest feasible regularization parameter (i.e., coefficient of $\norm{\bbeta}{2}^2$ in \eqref{eq:gsw-design-var}) is $\lambda=\zeta_{\bX}^2$. 

\subsection{Leave-One-Out and Leave-Two-Out Estimators}
\label{sec:connection-spiess}

\citet{aronow2013class} studies a general class of unbiased
Horvitz-Thompson estimators. These estimators adjust unit~$i$'s
outcome by a function $f(\bx_i,\btheta_i)$, where $\btheta_i$ is a
vector of parameters. The estimator is unbiased when $d_i$ is
uncorrelated with $f(\bx_i,\btheta_i)$.
A useful case arises when $\btheta_i$ depends only on outcomes $y_j$
whose treatment indicators $d_j$ are independent of $d_i$. Under
independent treatment assignment, this condition suggests 
leave-one-out constructions. LOORA-HT follows this logic by setting
$\btheta_i=\widehat{\bbeta}^{(-i)}_{\lambda}$, as defined in
\Cref{alg:ht-loo-reg-adj}.

\citet{wu2018loop} studies a special case of \citet{aronow2013class} in which the adjustment term for observation $i$ is constructed as a linear combination of two counterfactual outcomes for unit $i$ (obtained by any method of our choosing, such as linear regression or random forests)---a predicted outcome for unit $i$ under treatment and a predicted outcome under control---via a leave-one-out fit, which satisfies the sufficient condition for $f$ stated in \citet{aronow2013class}.

Building on these contributions, \citet{spiess2025optimal} provides a general characterization of unbiased regression-adjusted estimators for average treatment effects under arbitrary randomization schemes when the treatment-assignment probability is the same for all units.
He shows that for Horvitz--Thompson and difference-in-means type estimators, unbiasedness imposes specific structural constraints on how each unit’s outcome can depend on treatment assignments of other units.  
In particular, under simple random assignment with equal treatment assignment probability for all units, the Horwitz--Thompson estimator must have a leave-one-out structure—each adjusted outcome for unit $i$ may depend on the treatment assignments of all other units except $i$—whereas under complete random assignment, the estimator must satisfy a stronger leave-two-out condition, meaning that each pairwise contribution between units $i$ and $j$ can depend only on the treatment assignments of units other than $i$ and $j$.

We next show that LOORA-DM satisfies Spiess's leave-two-out requirement and therefore meets the necessary conditions for unbiasedness in \citet{spiess2025optimal}. Specifically, we show that the LOORA--DM estimator admits the representation
\[
\widehat{\tau}_{\textup{LDM}}
=
\frac{1}{n_T n_C}
\sum_{i<j}
(d_i - d_j)
\big(y_i - y_j - \phi_{ij}(\bz_{-ij})\big),
\]
where $\phi_{ij}$ is an adjustment function and $\bz_{-ij}$ collects the treatment assignments of all units other than $i$ and $j$.  
Note that
\[
\frac{1}{n_T n_C} \sum_{i<j} (d_i - d_j)(y_i - y_j)
=
\frac{1}{n_T}\sum_{d_i=1} y_i
-
\frac{1}{n_C}\sum_{d_i=0} y_i,
\]
so it suffices to characterize the adjustment term $\phi_{ij}(\bz_{-ij})$.

\paragraph{Pairwise representation.}
We define
\[
\widetilde{y}_{\ell}
=
\begin{cases}
\displaystyle\frac{n_C (n-1)}{(n_T - 1)n} y^{(1)}_{\ell} & \mbox{if } d_{\ell} = 1, \\[0.6ex]
\displaystyle\frac{n_T (n-1)}{(n_C - 1)n} y^{(0)}_{\ell} & \mbox{if }d_{\ell} = 0,
\end{cases}
\]
and set
\[
\phi_{ij}(\bz_{-ij})
=
\bx_i^\top
(\bX_{-i}^\top \bX_{-i} + \lambda \bI)^{-1}
\bX_{-ij}^\top
\widetilde{\by}_{-ij}
-
\bx_j^\top
(\bX_{-j}^\top \bX_{-j} + \lambda \bI)^{-1}
\bX_{-ij}^\top
\widetilde{\by}_{-ij}.
\]
This form makes explicit that $\phi_{ij}$ depends only on $\bz_{-ij}$—that is, on the treatment assignments of all units other than $i$ and $j$—and hence defines a leave-two-out structure.

Substituting this definition, we obtain
\begin{align*}
- \frac{1}{n_T n_C} \sum_{i<j} (d_i - d_j)\phi_{ij}(\bz_{-ij})
= {} &
-\frac{1}{n_T n_C}
\sum_{d_i=1}
\bx_i^\top
(\bX_{-i}^\top \bX_{-i} + \lambda \bI)^{-1}
\sum_{d_j=0}
\bX_{-ij}^\top \widetilde{\by}_{-ij}
\\
& +
\frac{1}{n_T n_C}
\sum_{d_i=0}
\bx_i^\top
(\bX_{-i}^\top \bX_{-i} + \lambda \bI)^{-1}
\sum_{d_j=1}
\bX_{-ij}^\top \widetilde{\by}_{-ij}.
\end{align*}
Some algebra shows that these expressions are exactly the adjustment values in \Cref{alg:dim-loo-reg-adj} for both treated and control units. Thus, $\widehat{\tau}_{\textup{LDM}}$ admits a leave-two-out representation. \Cref{sec:equivalence-loora-dm-leave-two-out} gives a detailed proof.

\section{Evidence from Within-Subject Experiments}
\label{sec:experiments}

This section presents empirical evidence on the finite-sample performance of LOORA estimators.
The analysis uses experimental data from two within-subject studies. In a within-subject experiment, each unit experiences both the treatment and control conditions, so both potential outcomes are observed for every unit.
We use within-subject experiment data to simulate randomized evaluations with known ground truth and with a realistic joint distribution of potential outcomes.

The datasets are drawn from \citet{allcott2015evaluating} and \citet{mcdonald2025evaluating}.  
A brief description of each is provided below.\smallskip 

\begin{enumerate}
    \item \textbf{Statehood dataset.}  
    \citet{mcdonald2025evaluating} studies the persuasiveness of policy arguments using a within-subject design.  
    Respondents first record their opinions on a policy issue, then read one or more arguments either supporting or opposing the policy, and finally re-state their opinions after exposure to each argument.  
    We use the portion of their data concerning opinions about granting statehood to the District of Columbia (DC).  
    The outcome under control is the respondent’s opinion before reading any arguments; the outcome under treatment is the opinion after reading (i) an argument against DC statehood emphasizing political corruption, and (ii) an argument in favor highlighting taxation without representation.  
    We restrict attention to male respondents aged 40-49.  
    Covariates include indicators for the different values of six categorical variables: party affiliation, race, voter registration, political attentiveness, education, and ideology.  
    The resulting dataset contains $36$ units and $32$ covariates. 

    \item \textbf{Lightbulb dataset.}  
    \citet{allcott2015evaluating} conducts an experiment on consumer preferences between two types of lightbulbs offered at varying relative prices.  
    Respondents first choose between the two options, after which they receive information about the energy costs of each type and make the choice again.  
    The initial choice forms the control outcome; the post-information choice forms the treatment outcome.  
    Covariates are binary indicators of categorical features from the original data, including employment status, renter status, U.S. region, living in a metropolitan area, marital status, income, housing type, gender, and ethnicity.  
    The final dataset contains $123$ units and $53$ covariates.
\end{enumerate}
\smallskip

\begin{table}[t]
\centering
\caption{Simulation results for Statehood dataset (covariate-dependent assignment, $\alpha=0.05$).}
\label{tab:sim_results_dep_simple_statehood}
\footnotesize
\begin{tabular}{lccccc}
\toprule
& \multicolumn{5}{c}{\textbf{ATE Estimators}} \\
\cmidrule(lr){2-6}
& \multicolumn{2}{c}{OLS} & \multicolumn{3}{c}{LOORA-HT} \\
\cmidrule(lr){2-3}\cmidrule(lr){4-6}
& Linear & Interacted & $\lambda{=}\zeta^2$ & $\lambda{=}2\zeta^2$ & $\lambda{=}3\zeta^2$ \\
\midrule

\multicolumn{6}{l}{} \\
Bias        & 0.045 & 0.047 & 0.000 & 0.000 & 0.000 \\
SD          & 0.519 & 0.639 & 0.314 & 0.296 & 0.290 \\
RMSE        & 0.521 & 0.640 & 0.314 & 0.296 & 0.290 \\
CI Coverage (HC0)& 0.727 & 0.167 & 0.948 & 0.950 & 0.951 \\
CI Coverage (HC2)& 0.929 & 0.723 & 0.948 & 0.950 & 0.951 \\
\bottomrule
\end{tabular}
\end{table}

\begin{table}[t]
\centering
\caption{Simulation results for Lightbulb dataset (covariate-dependent assignment, $\alpha=0.05$).}
\label{tab:sim_results_dep_simple_lightbulb}
\footnotesize
\begin{tabular}{lccccc}
\toprule
& \multicolumn{5}{c}{\textbf{ATE Estimators}} \\
\cmidrule(lr){2-6}
& \multicolumn{2}{c}{OLS} & \multicolumn{3}{c}{LOORA-HT} \\
\cmidrule(lr){2-3}\cmidrule(lr){4-6}
& Linear & Interacted & $\lambda{=}\zeta^2$ & $\lambda{=}2\zeta^2$ & $\lambda{=}3\zeta^2$ \\
\midrule

\multicolumn{6}{l}{} \\
Bias        & 0.002 & 0.001 & 0.000 & 0.000 & 0.000 \\
SD          & 0.106 & 0.152 & 0.089 & 0.086 & 0.084 \\
RMSE        & 0.106 & 0.152 & 0.089 & 0.086 & 0.084 \\
CI Coverage (HC0) & 0.911 & 0.715 & 0.956 & 0.958 & 0.958 \\
CI Coverage (HC2) & 0.960 & 0.942 & 0.955 & 0.957 & 0.958 \\
\bottomrule
\end{tabular}
\end{table}
 
We simulate randomized assignments by masking one outcome per unit according to a prescribed randomization rule.
This setup lets us assess each estimator’s finite-sample performance, bias, standard deviation, RMSE, and the average    coverage of confidence intervals, under realistic distributions for potential outcomes and covariates.
We consider three treatment-assignment mechanisms.  
In the first two designs, treatment assignments are independent across units.
In the first design, 
treatment assignment probabilities depend on covariate values.  
Specifically, we draw a random Gaussian vector in $\R^k$, compute for each unit $i$ the cosine similarity $c_i$ between this vector and its covariate vector, and set
\[
p_i = \max\!\left\{\min\!\left\{\frac{1+c_i}{2},\, 0.8\right\},\, 0.2\right\}.
\]
The second design is simple random assignment, with all treatment assignment probabilities equal to $0.5$.
The third design is complete random assignment, with half the units assigned to treatment.

Simulation results are based on $100{,}000$ repetitions. 
\Cref{tab:sim_results_dep_simple_statehood,tab:sim_results_dep_simple_lightbulb} report results for covariate-dependent assignment;
\Cref{tab:sim_results_simple_statehood,tab:sim_results_simple_lightbulb} for simple random assignment with equal probabilities; and
\Cref{tab:sim_results_complete_statehood,tab:sim_results_complete_lightbulb} for complete random assignment.
For simple random assignment, we compare our estimator, LOORA--HT, with two benchmarks: (i) linear regression-adjustment, and (ii) the regression adjustment of \cite{lin2013agnostic} with additional interacted terms. For complete random assignment, in addition to these, we also compare our estimator, LOORA--DM, with the exact bias correction approach of \cite{chang2024exact}, which we refer to as ``Debiased'' in \Cref{tab:sim_results_complete_statehood,tab:sim_results_complete_lightbulb}.

We report the coverage probabilities of confidence intervals in \Cref{tab:sim_results_dep_simple_statehood,tab:sim_results_dep_simple_lightbulb,tab:sim_results_simple_statehood,tab:sim_results_simple_lightbulb,tab:sim_results_complete_statehood,tab:sim_results_complete_lightbulb} using heteroskedasticity-consistent (HC) variance estimators. HC0 denotes White's original HC estimator \citep{white1980heteroskedasticity}, while HC2 denotes the leverage-adjusted HC estimator of \citet{mackinnon1985some}. The variance estimator in \Cref{thm:loora-ht-est-var} corresponds to HC0,
whereas the variance estimator in \Cref{thm:loora-dm-est-var} corresponds to
HC2.

All simulations use recentered covariates and an intercept. That is, regression adjustment uses the matrix of regressors 
\[
\bM = \begin{bmatrix}
\bX - \overline{\bX} & ~~ \boldsymbol{1}
\end{bmatrix}.
\]
The value of $\zeta$ in the tables is equal to $\norm{\bM}{2,\infty}$.

Throughout Tables \ref{tab:sim_results_dep_simple_statehood}--\ref{tab:sim_results_complete_lightbulb}, LOORA estimators consistently achieve lower bias, higher precision, and lower RMSE than competing methods. 
Moreover, the confidence intervals for LOORA estimators based on heteroskedasticity-robust standard errors exhibit coverage rates remarkably close to the nominal level, whereas the alternative methods tend to undercover substantially.
These results demonstrate the potential role of LOORA and ridge regularization to stabilize regression adjustment without sacrificing unbiasedness or inferential validity.

\begin{table}[t]
\centering
\caption{Simulation results on Statehood dataset (simple randomization, $\alpha=0.05$).}
\label{tab:sim_results_simple_statehood}
\footnotesize
\begin{tabular}{lccccc}
\toprule
& \multicolumn{5}{c}{\textbf{ATE Estimators}} \\
\cmidrule(lr){2-6}
& \multicolumn{2}{c}{OLS} & 
\multicolumn{3}{c}{LOORA-HT} \\
\cmidrule(lr){2-3}
\cmidrule(lr){4-6}
& Linear & Interacted &
$\lambda{=}\zeta^2$ & $\lambda{=}2\zeta^2$ & $\lambda{=}3\zeta^2$ \\
\midrule

\multicolumn{6}{l}{} \\
Bias        & 0.047 & 0.048 & 
0.000 & 0.000 & 0.000 \\
SD          & 0.512 & 0.628 & 
0.299 & 0.285 & 0.279 \\
RMSE        & 0.514 & 0.630 & 
0.299 & 0.285 & 0.279 \\
CI Coverage (HC0) & 0.724 & 0.149 & 
0.952 & 0.954 & 0.955 \\
CI Coverage (HC2) & 0.927 & 0.720 & 
0.953 & 0.953 & 0.955 \\
\bottomrule
\end{tabular}
\end{table}

\begin{table}[t]
\centering
\caption{Simulation results on Lightbulb dataset (simple randomization, $\alpha=0.05$).}
\label{tab:sim_results_simple_lightbulb}
\footnotesize
\begin{tabular}{lccccc}
\toprule
& \multicolumn{5}{c}{\textbf{ATE Estimators}} \\
\cmidrule(lr){2-6}
& \multicolumn{2}{c}{OLS} & 
\multicolumn{3}{c}{LOORA-HT} \\
\cmidrule(lr){2-3}
\cmidrule(lr){4-6}
& Linear & Interacted & 
$\lambda{=}\zeta^2$ & $\lambda{=}2\zeta^2$ & $\lambda{=}3\zeta^2$ \\
\midrule

\multicolumn{6}{l}{} \\
Bias        & 0.002 & 0.001 & 
0.000 & 0.000 & 0.001 \\
SD          & 0.105 & 0.134 & 
0.084 & 0.082 & 0.081 \\
RMSE        & 0.105 & 0.134 & 
0.084 & 0.082 & 0.081 \\
CI Coverage (HC0)& 0.909 & 0.733 & 
0.961 & 0.962 & 0.963 \\
CI Coverage (HC2) & 0.961 & 0.947 & 
0.961 & 0.961 & 0.962 \\
\bottomrule
\end{tabular}
\end{table}

\begin{table}[t]
\centering
\caption{Simulation results on Statehood dataset (complete randomization, $\alpha=0.05$).}
\label{tab:sim_results_complete_statehood}
\footnotesize
\begin{tabular}{lccccccc}
\toprule
& \multicolumn{7}{c}{\textbf{ATE Estimators}} \\
\cmidrule(lr){2-8}
& \multicolumn{2}{c}{OLS} & \multicolumn{2}{c}{Debiased} & \multicolumn{3}{c}{LOORA-DM} \\
\cmidrule(lr){2-3}\cmidrule(lr){4-5}\cmidrule(lr){6-8}
& Linear & Interacted & Linear & Interacted & $\lambda{=}\zeta^2$ & $\lambda{=}2\zeta^2$ & $\lambda{=}3\zeta^2$ \\
\midrule

\multicolumn{8}{l}{} \\
Bias        & 0.046 & 0.043 & 0.001 & 0.039 & 0.001 & 0.001 & 0.001 \\
SD          & 0.501 & 0.615 & 0.311 & 0.410 & 0.282 & 0.274 & 0.272 \\
RMSE        & 0.503 & 0.616 & 0.311 & 0.412 & 0.282 & 0.274 & 0.272 \\
CI Coverage (HC0)& 0.726 & 0.104 & 0.901 & 0.147 & 0.955 & 0.958 & 0.959 \\
CI Coverage (HC2)& 0.928 & 0.701 & 0.985 & 0.723 & 0.961 & 0.963 & 0.964 \\

\bottomrule
\end{tabular}
\end{table}

\begin{table}[t]
\centering
\caption{Simulation results on Lightbulb dataset (complete randomization, $\alpha=0.05$).}
\label{tab:sim_results_complete_lightbulb}
\footnotesize
\begin{tabular}{lccccccc}
\toprule
& \multicolumn{7}{c}{\textbf{ATE Estimators}} \\
\cmidrule(lr){2-8}
& \multicolumn{2}{c}{OLS} & \multicolumn{2}{c}{Debiased} & \multicolumn{3}{c}{LOORA-DM} \\
\cmidrule(lr){2-3}\cmidrule(lr){4-5}\cmidrule(lr){6-8}
& Linear & Interacted & Linear & Interacted & $\lambda{=}\zeta^2$ & $\lambda{=}2\zeta^2$ & $\lambda{=}3\zeta^2$ \\
\midrule

\multicolumn{8}{l}{} \\
Bias        & 0.002 & 0.003 & 0.000 & 0.001 & 0.000 & 0.000 & 0.000 \\
SD          & 0.105 & 0.128 & 0.085 & 0.089 & 0.082 & 0.080 & 0.080 \\
RMSE        & 0.105 & 0.128 & 0.085 & 0.089 & 0.082 & 0.080 & 0.080 \\
CI Coverage (HC0)& 0.910 & 0.748 & 0.964 & 0.886 & 0.963 & 0.962 & 0.963 \\
CI Coverage (HC2)& 0.962 & 0.951 & 0.989 & 0.985 & 0.964 & 0.965 & 0.965 \\
\bottomrule
\end{tabular}
\end{table}

\section{Conclusion}

This article develops leave-one-out regression-adjusted estimators of average treatment effects under simple and complete random assignment.  
We show that these estimators are exactly unbiased in finite populations, achieve the asymptotic efficiency of the interacted regression of \citet{lin2013agnostic}, and admit closed-form expressions for their variance.  
Regularization through leave-one-out ridge adjustment ensures robustness to high-leverage observations and stabilizes inference even in small samples.  
Our theoretical results establish a design-based foundation for regression adjustment that unifies classical unbiasedness, modern efficiency results, and practical inference procedures.  
Empirical evaluations confirm the estimators’ excellent finite-sample performance across different designs.  

Natural next steps include extending the framework to stratified or covariate-adaptive designs and to clustered or networked experiments, and incorporating high-dimensional or nonparametric adjustments. Taken together, these directions define a broad agenda for design-based regression adjustment in modern experimental settings.

\bibliographystyle{ecta-fullname}
\bibliography{main}

@article{abadie2020sampling,
author = {Abadie, Alberto and Athey, Susan and Imbens, Guido W. and Wooldridge, Jeffrey M.},
title = {Sampling-based versus design-based uncertainty in regression analysis},
journal = {Econometrica},
volume = {88},
number = {1},
pages = {265-296},
year = {2020}
}

@misc{spiess2025optimal,
    author = {Jann Spiess},
    year = {2025},
    title = {Optimal estimation when researcher and social preferences are misaligned},
    note = {Forthcoming in {\it Econometrica}}
}

@article{harshaw2022design,
  title={Balancing covariates in randomized experiments with the {G}ram--{S}chmidt walk design},
  author={Harshaw, Christopher and S{\"a}vje, Fredrik and Spielman, Daniel A and Zhang, Peng},
  journal={Journal of the American Statistical Association},
  volume={119},
  number={548},
  pages={2934--2946},
  year={2024},
  publisher={Taylor \& Francis}
}

@book{billingsley2012probability,
  title={Probability and Measure},
  author={Billingsley, P.},
  isbn={9781118341919},
  series={Wiley Series in Probability and Statistics},
  url={https://books.google.com/books?id=a3gavZbxyJcC},
  year={2012},
  publisher={Wiley}
}

@article{hoeffdingCLT,
 ISSN = {00034851},
 URL = {http://www.jstor.org/stable/2236924},
 author = {Wassily Hoeffding},
 journal = {The Annals of Mathematical Statistics},
 number = {4},
 pages = {558--566},
 publisher = {Institute of Mathematical Statistics},
 title = {A combinatorial central limit theorem},
 urldate = {2026-01-06},
 volume = {22},
 year = {1951}
}

@article{mcdonald2025evaluating,
  title={Evaluating methods for examining the relative persuasiveness of policy arguments},
  author={McDonald, Jared and Hanmer, Michael J},
  journal={Political Science Research and Methods},
  volume={13},
  number={1},
  pages={229--236},
  year={2025},
  publisher={Cambridge University Press}
}

@article{allcott2015evaluating,
  title={Evaluating behaviorally motivated policy: Experimental evidence from the lightbulb market},
  author={Allcott, Hunt and Taubinsky, Dmitry},
  journal={American Economic Review},
  volume={105},
  number={8},
  pages={2501--2538},
  year={2015},
  publisher={American Economic Association 2014 Broadway, Suite 305, Nashville, TN 37203}
}

@article{lei2021regression,
  title={Regression adjustment in completely randomized experiments with a diverging number of covariates},
  author={Lei, Lihua and Ding, Peng},
  journal={Biometrika},
  volume={108},
  number={4},
  pages={815--828},
  year={2021},
  publisher={Oxford University Press}
}

@article{freedman2008regression,
  title={On regression adjustments to experimental data},
  author={Freedman, David A},
  journal={Advances in Applied Mathematics},
  volume={40},
  number={2},
  pages={180--193},
  year={2008},
  publisher={Elsevier}
}

@article{freedman_2008_b,
 ISSN = {19326157},
 URL = {http://www.jstor.org/stable/30244182},
 author = {David A. Freedman},
 journal = {The Annals of Applied Statistics},
 number = {1},
 pages = {176--196},
 publisher = {Institute of Mathematical Statistics},
 title = {On regression adjustments in experiments with several treatments},
 urldate = {2025-04-25},
 volume = {2},
 year = {2008}
}

@article{aronow2013class,
  title={A class of unbiased estimators of the average treatment effect in randomized experiments},
  author={Aronow, Peter M and Middleton, Joel A},
  journal={Journal of Causal Inference},
  volume={1},
  number={1},
  pages={135--154},
  year={2013},
  publisher={De Gruyter}
}

@inproceedings{ghadiri2023finite,
  title={Finite population regression adjustment and non-asymptotic guarantees for treatment effect estimation},
  author={Ghadiri, Mehrdad and Arbour, David and Mai, Tung and Musco, Cameron N and Rao, Anup},
  booktitle={Thirty-seventh Conference on Neural Information Processing Systems},
  year={2023}
}

@article{young2019channeling,
  title={Channeling {F}isher: Randomization tests and the statistical insignificance of seemingly significant experimental results},
  author={Young, Alwyn},
  journal={The Quarterly Journal of Economics},
  volume={134},
  number={2},
  pages={557--598},
  year={2019},
  publisher={Oxford University Press}
}

@article{bai2025new,
  title={A new design-based variance estimator for finely stratified experiments},
  author={Bai, Yuehao and Huang, Xun and Romano, Joseph P and Shaikh, Azeem M and Tabord-Meehan, Max},
  journal={arXiv preprint arXiv:2503.10851},
  year={2025}
}

@article{bai2025inference,
  title={Inference in experiments with matched pairs and imperfect compliance},
  author={Bai, Yuehao and Guo, Hongchang and Shaikh, Azeem M and Tabord-Meehan, Max},
  journal={Journal of Business \& Economic Statistics},
  volume={43},
  number={3},
  pages={627--642},
  year={2025},
  publisher={Taylor \& Francis}
}

@article{cytrynbaum2024covariate,
  title={Covariate adjustment in stratified experiments},
  author={Cytrynbaum, Max},
  journal={Quantitative Economics},
  volume={15},
  number={4},
  pages={971--998},
  year={2024},
  publisher={Wiley Online Library}
}

@article{chang2024exact,
  title={Exact bias correction for linear adjustment of randomized controlled trials},
  author={Chang, Haoge and Middleton, Joel A and Aronow, PM},
  journal={Econometrica},
  volume={92},
  number={5},
  pages={1503--1519},
  year={2024},
  publisher={Wiley Online Library}
}

@article{white1980heteroskedasticity,
  author  = {White, Halbert},
  title   = {A Heteroskedasticity-Consistent Covariance Matrix Estimator and a Direct Test for Heteroskedasticity},
  journal = {Econometrica},
  year    = {1980},
  volume  = {48},
  number  = {4},
  pages   = {817--838},
  doi     = {10.2307/1912934}
}

@article{mackinnon1985some,
  title={Some heteroskedasticity-consistent covariance matrix estimators with improved finite sample properties},
  author={MacKinnon, James G and White, Halbert},
  journal={Journal of econometrics},
  volume={29},
  number={3},
  pages={305--325},
  year={1985},
  publisher={Elsevier}
}

@article{alaoui2015fast,
  title={Fast randomized kernel ridge regression with statistical guarantees},
  author={Alaoui, Ahmed and Mahoney, Michael W},
  journal={Advances in neural information processing systems},
  volume={28},
  year={2015}
}

@article{imbens2015causal,
  title={Causal inference for statistics, social, and biomedical sciences: An introduction.},
  author={Imbens, Guido W and Rubin, Donald B},
  journal={Cambridge University Press},
  year={2015},
  publisher={ERIC}
}

@article{wu2018loop,
  title={The LOOP estimator: Adjusting for covariates in randomized experiments},
  author={Wu, Edward and Gagnon-Bartsch, Johann A},
  journal={Evaluation review},
  volume={42},
  number={4},
  pages={458--488},
  year={2018},
  publisher={SAGE Publications Sage CA: Los Angeles, CA}
}

@article{gu2025assumption,
  title={Assumption-lean covariate adjustment under covariate adaptive randomization when $ p= o (n) $},
  author={Gu, Yujia and Liu, Lin and Ma, Wei},
  journal={arXiv preprint arXiv:2512.20046},
  year={2025}
}

@article{zhao2024covariate,
  title={Covariate Adjustment in Randomized Experiments Motivated by Higher-Order Influence Functions},
  author={Zhao, Sihui and Wang, Xinbo and Liu, Lin and Zhang, Xin},
  journal={arXiv preprint arXiv:2411.08491},
  year={2024}
}

@article{neyman1923application,
  title={On the application of probability theory to agricultural experiments: essay on principles, Section 9},
  author={Neyman, J},
  journal={Statistical Science},
  volume={5},
  pages={465--480},
  year={1923}
}

@article{lin2013agnostic,
  title={Agnostic notes on regression adjustments to experimental data: reexamining {F}reedman's critique},
  author={Lin, Winston},
  journal={The Annals of Applied Statistics},
  pages={295--318},
  year={2013},
  publisher={JSTOR}
}

@article{horvitz1952generalization,
  title={A generalization of sampling without replacement from a finite universe},
  author={Horvitz, Daniel G and Thompson, Donovan J},
  journal={Journal of the American statistical Association},
  volume={47},
  number={260},
  pages={663--685},
  year={1952},
  publisher={Taylor \& Francis}
}

@article{rubin1974estimating,
  title={Estimating causal effects of treatments in randomized and nonrandomized studies.},
  author={Rubin, Donald B},
  journal={Journal of educational Psychology},
  volume={66},
  number={5},
  pages={688},
  year={1974},
  publisher={American Psychological Association}
}

@article{miller1974unbalanced,
  title={An unbalanced jackknife},
  author={Miller, Rupert G Jr},
  journal={The Annals of Statistics},
  pages={880--891},
  year={1974},
  publisher={JSTOR}
}

@book{fisher1971design,
  title={The design of experiments},
  author={Fisher, Ronald A.},
  edition = {9th},
  year={1971},
  publisher={Springer}
}

@article{fahrbach2022subquadratic,
  title={Subquadratic {K}ronecker regression with applications to tensor decomposition},
  author={Fahrbach, Matthew and Fu, Gang and Ghadiri, Mehrdad},
  journal={Advances in Neural Information Processing Systems},
  volume={35},
  pages={28776--28789},
  year={2022}
}

@article{bai2024primer,
  title={A primer on the analysis of randomized experiments and a survey of some recent advances},
  author={Bai, Yuehao and Shaikh, Azeem M and Tabord-Meehan, Max},
  journal={arXiv preprint arXiv:2405.03910},
  year={2024}
}

\begin{appendix}

\section{Linear Algebra Tools}
\label{sec:bounds-lev-loo}

\begin{lemma}[(Lemma 3.2 of \cite{miller1974unbalanced})]
\label{lemma:sol-change-in-row-removal}
Let 
\begin{align*}
\widehat{\bbeta} = \argmin_{\bb\in \mathbb R^k} \norm{\by - \bX \bb}{2}^2  ~~ \text{and} ~~ \widehat{\bbeta}^{(-i)} = \argmin_{\bb\in \mathbb R^k} \norm{\by_{-i} - \bX_{-i} \bb}{2}^2.
\end{align*}
Then
\begin{align*}
\widehat\bbeta - \widehat{\bbeta}^{(-i)} = \frac{(\bX^\top \bX)^{-1} \bx_i (y_i - \bx_i^\top \widehat\bbeta)}{1 - h_{ii}},
\end{align*}
where $h_{ii} = \bx_i^\top (\bX^\top \bX)^{-1} \bx_i$ is the leverage score for observation $i$.
\end{lemma}

\begin{lemma}[\cite{alaoui2015fast,fahrbach2022subquadratic}]
\label{lem:ridge-lev-formula}
Let $\bU\bSigma\bV$ be the compact SVD of $\bX$. Let $r$ be the rank of $\bX$ and $\sigma_1,\ldots,\sigma_r$ be its singular values. Then
\[
h_{\lambda ii} = \sum_{j=1}^r \frac{\sigma_j^2 u_{ij}^2}{\sigma_j^2 + \lambda}\,.
\]
\end{lemma}

\begin{proof}[Proof of \Cref{lem:ridge-lev-upper-bound}]
Let $\bU\bSigma\bV$ be the compact SVD of $\bX$ and $r$ be its rank. 
Then by \Cref{lem:ridge-lev-formula}, 
\begin{align}
\label{eq:ridge-lev-formula}
h_i(\bX,\lambda) = \sum_{j=1}^r \frac{\sigma_j^2 u_{ij}^2}{\sigma_j^2 + \lambda}
\end{align}
Since $\bU$ is an orthogonal matrix, for all $i\in[n]$,
\[
\sum_{j=1}^r u_{ij}^2 \leq 1.
\]
By definition of $\zeta$, for all $i\in [n]$,
\[
\ell_i=\norm{\bx_i}{2}^2 = \sum_{j=1}^r \sigma_j^2 u_{ij}^2 \leq \zeta^{2}.
\]
Therefore by \eqref{eq:ridge-lev-formula},
\begin{align*}
h_i(\bX,c \cdot \zeta^{2}) \leq \sum_{j=1}^r \frac{\sigma_j^2 u_{ij}^2}{\sigma_j^2 + c\ell_i}
\end{align*}
Let $Y$ be a random variable that is equal to $\sigma_j^2$ with probability $u_{ij}^2$, and take value 0 with probability ($1-\sum_{j=1}^ru^2_{ij}$), and $\phi:\R_{\geq 0} \to \R_{\geq 0}$ be a function with $\phi(g) = \frac{g}{g+c\ell_i}$. Therefore
\[
\sum_{j=1}^r \frac{\sigma_j^2 u_{ij}^2}{\sigma_j^2 + c\ell_i} = \E[\phi(Y)].
\]
Now by Jensen's ienquality, since $\phi$ is a concave function,  $\E[\phi(Y)] \leq \phi(\E[Y])$. Therefore recalling that $\ell_i = \sum_{j=1}^r \sigma_j^2 u_{ij}^2$, we have
\[
\sum_{j=1}^r \frac{\sigma_j^2 u_{ij}^2}{\sigma_j^2 + c\ell_i} \leq \phi(\E[Y]) = \frac{\ell_i}{\ell_i + c\ell_i} = \frac{1}{1+c}\,.
\]
\end{proof}

\begin{lemma}
\label{lemma:recenter}
Let $\bM \in \R^{n \times k}$ and $\bmu \in \R^{n}$. Then
\[
\min_{\bb \in \R^{k}} \norm{(\bM - \overline{\bM}) \bb - (\bmu - \overline{\bmu})}{2}^2
= \min_{\bb \in \R^{k+1}} \norm{\begin{bmatrix}
\bM  & ~~ \boldsymbol{1}
\end{bmatrix} \bb - \bmu}{2}^2.
\]
\end{lemma}
\begin{proof}
Let $\bN = \begin{bmatrix}
(\bM - \overline{\bM}) & ~~\boldsymbol{1}
\end{bmatrix}$ and let $\bH^{(\bN)}$ denote the projection matrix onto the column space of $\bN$, i.e.,
$\bH^{(\bN)} = \bN(\bN^\top \bN)^{-1} \bN^\top$.
Then
\[
\min_{\bb \in \R^{k+1}} \norm{\bN \bb - \bmu}{2}^2
= \norm{\bH^{(\bN)} \bmu - \bmu}{2}^2.
\]
Since $\boldsymbol{1}$ is orthogonal to the columns of $\bM - \overline{\bM}$, we have
\begin{align}
\label{eq:recenter-1}
\bH^{(\bN)} \bmu
= \bH^{(\bM - \overline{\bM})} \bmu + \bH^{(\boldsymbol{1})} \bmu
= \bH^{(\bM - \overline{\bM})} (\bmu - \overline{\bmu}) + \overline{\bmu},
\end{align}
where the second equality uses $\overline{\bmu} = \bH^{(\boldsymbol{1})} \bmu$ and
$\bH^{(\bM - \overline{\bM})} \overline{\bmu} = 0$.
Therefore,
\begin{align*}
\norm{\bH^{(\bN)} \bmu - \bmu}{2}^2
&= \norm{\bH^{(\bM - \overline{\bM})} (\bmu - \overline{\bmu}) - (\bmu - \overline{\bmu})}{2}^2 \\
&= \min_{\bb \in \R^{k}} \norm{(\bM - \overline{\bM}) \bb - (\bmu - \overline{\bmu})}{2}^2.
\end{align*}
Finally note that the span of columns of $\begin{bmatrix}
\bM & ~~ \boldsymbol{1}
\end{bmatrix}$ and $\bN$ are the same. Thus
\[
\min_{\bb \in \R^{k+1}} \norm{\bN \bb - \bmu}{2}^2 =  \min_{\bb \in \R^{k+1}} \norm{\begin{bmatrix}
\bM & ~~ \boldsymbol{1}
\end{bmatrix} \bb - \bmu}{2}^2,
\]
which concludes the result.
\end{proof}

\section{Main Proofs}

LOORA-HT depends on a set of $n$ leave-one-out regression coefficients $\widehat{\bbeta}_\lambda^{(-i)}$.
Next lemma characterizes the difference between the full-sample regularized least-squares coefficient vector and its leave-one-out counterparts 
(see \Cref{lemma:sol-change-in-row-removal} 
for a simpler version for least-squares).
\begin{lemma}
\label{cor:sum-sqr-err-with-removal}
Let $\lambda\geq 0$,
\begin{align*}
\widehat{\bbeta}_{\lambda} = \argmin_{\bb\in \mathbb R^k} \norm{\by - \bX \bb}{2}^2+ \lambda \norm{\bb}{2}^2,  ~~ \text{\ and\ } ~~ \widehat{\bbeta}_{\lambda}^{(-i)} = \argmin_{\bb\in \mathbb R^k} \norm{\by_{-i} - \bX_{-i} \bb}{2}^2+ \lambda \norm{\bb}{2}^2.
\end{align*}
Then, 
$$\widehat\bbeta_\lambda - \widehat{\bbeta}_\lambda^{(-i)} = \frac{(\bX^\top \bX + \lambda \bI)^{-1} \bx_i (y_i - \bx_i^\top \widehat\bbeta_\lambda)}{1 - h_{\lambda ii}},$$
and
$$y_i - \bx_i^\top \widehat\bbeta_{\lambda} = (1- h_{\lambda ii})(y_i - \bx_i^\top \widehat{\bbeta}_{\lambda}^{(-i)}),$$
where $h_{\lambda ii} = \bx_i^\top (\bX^\top \bX + \lambda \bI)^{-1} \bx_i$. Moreover, 
\[
\bx_i^\top\widehat{\bbeta}_{\lambda}^{(-i)} = 
\frac{\bx_i^\top(\bX^\top \bX + \lambda \bI)^{-1} \bX_{-i}^\top \by_{-i}}{1-h_{\lambda ii}}.
\]
\end{lemma}

\begin{proof}
Let
\begin{align*}
\bX_{\lambda} = \begin{bmatrix}
\bX \\ \sqrt{\lambda} \bI
\end{bmatrix} \in \R^{(n+k) \times k} ~, ~ \text{and} ~~
\widecheck{\by} = \begin{bmatrix}
\by \\ \boldsymbol{0}
\end{bmatrix} \in \R^{n+k} \,.
\end{align*}
One can observe that for any $\bb$,
\begin{align*}
\norm{\bX_{\lambda} \bb - \widecheck{\by}}{2}^2 = \norm{\bX \bb - \by}{2}^2 + \lambda \norm{\bb}{2}^2\,.
\end{align*}
Therefore it immediately follows by \Cref{lemma:sol-change-in-row-removal} that
\begin{align}
\label{eq:sol-dif-loo}
\widehat\bbeta_\lambda - \widehat{\bbeta}_\lambda^{(-i)} = \frac{(\bX^\top \bX + \lambda \bI)^{-1} \bx_i (y_i - \bx_i^\top \widehat\bbeta_\lambda)}{1 - h_{\lambda ii}},
\end{align}
and therefore
\begin{align*}
& \bx_i^\top \widehat\bbeta_\lambda - \bx_i^\top \widehat{\bbeta}_\lambda^{(-i)} = \frac{\bx_i^\top(\bX^\top \bX + \lambda \bI)^{-1} \bx_i (y_i - \bx_i^\top \widehat\bbeta_\lambda)}{1 - h_i(\bX,\lambda)} = \frac{h_i(\bX,\lambda)(y_i - \bx_i^\top \widehat\bbeta_\lambda)}{1 - h_i(\bX, \lambda)}
\\
\implies &
\bx_i^\top \widehat\bbeta_\lambda - y_i = (1- h_i(\bX, \lambda)) (\bx_i^\top \widehat{\bbeta}_\lambda^{(-i)} - y_i)\,.
\end{align*}
Thus,
\begin{align*}
\bx_i^\top\widehat{\bbeta}_{\lambda}^{(-i)} & = \bx_i^\top \widehat\bbeta_\lambda -\frac{h_{\lambda ii} (y_i - \bx_i^\top \widehat\bbeta_\lambda)}{1 - h_{\lambda ii}} = \frac{\bx_i^\top \widehat\bbeta_\lambda - h_{\lambda ii} y_i}{1 - h_{\lambda ii}}
\\ & =
\frac{\bx_i^\top(\bX^\top \bX + \lambda \bI)^{-1} \bX^\top \by - \bx_i^\top(\bX^\top \bX + \lambda \bI)^{-1} \bx_{i} y_i}{1 - h_{\lambda ii}}
\\ & =
\frac{
\bx_i^\top(\bX^\top \bX + \lambda \bI)^{-1} \bX_{-i}^\top \by_{-i}}{1 - h_{\lambda ii}}.
\end{align*}
\end{proof}

\Cref{cor:sum-sqr-err-with-removal} characterizes the deviation of the solution, estimator, and residual when we solve a leave-one-out regression instead of the full regression. The key observation is that all deviations depend on the leverage score of the removed row, which in turn helps characterize the robustness of our leave-one-out estimators to observation removal at inference time. This issue is particularly important in small-sample settings, where a single observation may have a large (close to $1$) leverage score and thus exert a disproportionate influence on quantities such as variance estimates and confidence intervals, which is undesirable. Hence, this result provides one of the first motivations for favoring regularized regression adjustment methods: as seen in \Cref{lem:ridge-lev-upper-bound}, regularization enables us to uniformly reduce leverage scores across all observations.

\begin{lemma}
\label{lem:HT-LOO-Binary-estimate-diff} 
Let 
\begin{align*}
g_i & = \frac{(r_i\mu_i-\bx_{i}^\top \bbeta_\lambda) - \bx_i^\top (\widetilde{\bX}^\top \widetilde{\bX} + \lambda \bI)^{-1} \widetilde{\bX}_{-i}^{\top} (\bz\bt/\bq)_{-i}}{q_i(1-\widetilde{h}_{\lambda ii})}\,.
\end{align*}
Under simple random assignment, the LOORA-HT estimator in \Cref{alg:ht-loo-reg-adj} satisfies 
\begin{align*}
\widehat{\tau}_{\textup{LHT}} - \tau = \frac{1}{n} \bz^\top \bg.
\end{align*}
\end{lemma}

\begin{proof} First, notice that
\begin{align*}
\widehat{\tau}_{\textup{LHT}} - \tau = \frac{1}{n}\sum_{i=1}^n\left( \frac{z_i}{q_i} (y_i - \bx_{i}^\top \widehat{\bbeta}_\lambda^{(-i)}) - (y^{(1)}_i - y^{(0)}_i)\right).
\end{align*}
Because $r_i\mu_i = (1-p_i)y_i^{(1)}+p_i y_i^{(0)}$, it follows that 
\[
\frac{z_i}{q_i} (y_i - \bx_{i}^\top \widehat{\bbeta}_\lambda^{(-i)}) - (y^{(1)}_i - y^{(0)}_i) = \frac{z_i}{q_i}(r_i \mu_i -\bx_{i}^\top \widehat{\bbeta}_\lambda^{(-i)}).
\]
Algebraic manipulations yield $
\widetilde \by -\bmu = \bz\bt/\bq$ and therefore, 
\begin{align*}
\widehat{\bbeta}_\lambda^{(-i)} = (\widetilde{\bX}_{-i}^\top \widetilde{\bX}_{-i} + \lambda \bI)^{-1} \widetilde{\bX}_{-i}^{\top} (\bmu + \bz\bt/\bq)_{-i}\,.
\end{align*}
By \Cref{cor:sum-sqr-err-with-removal},
\begin{align*}
\mu_i-\widetilde\bx_{i}^\top \widehat{\bbeta}^{(-i)}_{\lambda} &=  \mu_i-\widetilde\bx_{i}^\top (\widetilde{\bX}_{-i}^\top \widetilde{\bX}_{-i} + \lambda \bI)^{-1} \widetilde{\bX}_{-i}^{\top} (\bmu + \bz\bt/\bq)_{-i}\\
&=\mu_i - \widetilde\bx_{i}^\top (\widetilde{\bX}_{-i}^\top \widetilde{\bX}_{-i} + \lambda \bI)^{-1} \widetilde{\bX}_{-i}^{\top} \bmu_{-i}-\widetilde\bx_{i}^\top (\widetilde{\bX}_{-i}^\top \widetilde{\bX}_{-i} + \lambda \bI)^{-1} \widetilde{\bX}_{-i}^\top (\bz\bt/\bq)_{-i}
\\&= \frac{\mu_i - \widetilde\bx_{i}^\top \bbeta_{\lambda} - \widetilde\bx_i^\top (\widetilde{\bX}^\top \widetilde{\bX} + \lambda \bI)^{-1} \widetilde{\bX}_{-i}^{\top} (\bz\bt/\bq)_{-i}}{1-\widetilde h_{\lambda ii}},
\end{align*}
where $\bbeta_\lambda$ is defined as in \Cref{eq:real_lambda_ridge}. 
Therefore, 
\begin{align*}
\frac{z_i}{q_i}(r_i\mu_i - \bx_{i}^\top \widehat{\bbeta}_\lambda^{(-i)}) &=\frac{z_i}{q_i}r_i(\mu_i - \widetilde\bx_{i}^\top \widehat{\bbeta}_\lambda^{(-i)})\\
&=z_i\frac{(r_i\mu_i - \bx_{i}^\top \bbeta_{\lambda}) - \bx_i^\top (\widetilde{\bX}^\top \widetilde{\bX} + \lambda \bI)^{-1} \widetilde{\bX}_{-i}^{\top} (\bz\bt/\bq)_{-i}}{q_i(1-\widetilde h_{\lambda ii})}.
\end{align*}
\end{proof}

\subsection{Omitted Proofs of \Cref{sec:reg-adj-ht}}

\begin{proof}[Proof of \Cref{thm:HT-LOO-Binary-Variance}]
Because $z_i$ is independent of  $\widehat{\bbeta}^{(-i)}$,
\begin{align*}
\E\left[\frac{z_i}{q_i} (y_i - \bx_{i}^\top \widehat{\bbeta}^{(-i)}) \right] 
& = 
\E [y^{(1)}_i - \bx_i^\top \widehat{\bbeta}^{(-i)} | z_i=1] - \E[ y^{(0)}_i - \bx_i^\top \widehat{\bbeta}^{(-i)} | z_i = -1]
\\ & =
y^{(1)}_i - y^{(0)}_i\,.
\end{align*}
As a result, $\widehat{\tau}_{\textup{LHT}}$ is unbiased. By \Cref{lem:HT-LOO-Binary-estimate-diff}
\begin{align*}
\E[(\widehat{\tau}_{\textup{LHT}} - \tau)^2] = \frac{1}{n^2} \E[\bz^\top \bg \bg^\top \bz].
\end{align*}
By \Cref{lem:HT-LOO-Binary-estimate-diff}, we have $\widehat{\tau} -\tau = \frac{1}{n} \bz^\top \bg$.
Therefore
\begin{align*}
&\E[(\widehat{\tau} - \tau)^2]
= 
\frac{1}{n^2} \E[\bz^\top \bg \bg^\top \bz].
\end{align*}
To calculate $\E[\bz^\top \bg \bg^\top \bz]$, first notice that $z_i^2=1$ and $\E[1/q_i^2]=1/r_i^2$. Therefore, 
\begin{align*}
\E\left[\left(\frac{z_i( r_i \mu_i-\bx_{i}^\top \bbeta_\lambda)}{q_i(1-\widetilde{h}_{\lambda ii})}\right)^2 \right] 
& = 
\frac{(r_i\mu_i-\bx_{i}^\top \bbeta_\lambda )^2}{r_{i}^2(1-\widetilde{h}_{\lambda ii})^2}
=
\frac{(\mu_i-\widetilde{\bx}_{i}^\top \bbeta_\lambda )^2}{(1-\widetilde{h}_{\lambda ii})^2}.
\end{align*}
Because $z_i/q_i$ and $z_j/q_j$ are independent for $i\neq j$, and $\E[z_i/q_i] = 0$, then
\begin{align*}
\E\Bigg[\frac{z_i(r_i \mu_i - \bx_{i}^\top \bbeta_\lambda)}{q_i(1-\widetilde{h}_{\lambda ii})} & \frac{ z_i\bx_i^\top (\widetilde{\bX}^\top \widetilde{\bX} + \lambda \bI)^{-1} \widetilde{\bX}_{-i}^{\top} (\bz\bt/\bq)_{-i}}{q_i(1-\widetilde{h}_{\lambda ii})} \Bigg]\\ &=
\E\Bigg[\frac{(r_i \mu_i - \bx_{i}^\top \bbeta_\lambda)}{q_i^2}\Bigg]  \frac{ \bx_i^\top (\widetilde{\bX}^\top \widetilde{\bX} + \lambda \bI)^{-1} \widetilde{\bX}_{-i}^{\top} }{(1-\widetilde{h}_{\lambda ii})^2}\,\E[(\bz\bt/\bq)_{-i}]
=0.
\end{align*}
Moreover for $i\neq j$,
\begin{align*}
\E\Bigg[\frac{z_i(r_i \mu_i - \bx_{i}^\top \bbeta_\lambda)}{q_i(1-\widetilde{h}_{\lambda ii})} & \frac{ z_j\bx_j^\top (\widetilde{\bX}^\top \widetilde{\bX} + \lambda \bI)^{-1} \widetilde{\bX}_{-j}^{\top} (\bz\bt/\bq)_{-j}}{q_j(1-\widetilde{h}_{\lambda jj})} \Bigg]\\ &=
\E\Bigg[\frac{z_j}{q_j}\Bigg]
\E\Bigg[\frac{z_i(r_i \mu_i - \bx_{i}^\top \bbeta_\lambda)}{q_i(1-\widetilde{h}_{\lambda ii})}\frac{\bx_j^\top (\widetilde{\bX}^\top \widetilde{\bX} + \lambda \bI)^{-1} \widetilde{\bX}_{-j}^{\top} (\bz\bt/\bq)_{-j}}{(1-\widetilde{h}_{\lambda jj})} \Bigg]
=0.
\end{align*}
Using $\E[1/q_i^2]=1/r_i^2$ and independent between $z_i/q_i$ and $z_j/q_j$ for $i\neq j$, we obtain
\begin{align*}
\E\left[\left(\frac{z_i \bx_i^\top (\widetilde{\bX}^\top \widetilde{\bX} + \lambda \bI)^{-1} \widetilde{\bX}_{-i}^{\top} (\bz\bt/\bq)_{-i}}{q_i(1-\widetilde{h}_{\lambda ii})} \right)^2 \right]
&=\frac{1}{(1-\widetilde h_{\lambda ii})^2}\E\Bigg[\Bigg(\sum_{j\neq i} \widetilde h_{\lambda ij} t_j z_j/q_j\Bigg)^2\Bigg]\\
& =
\sum_{j \neq i} \frac{\widetilde h_{\lambda ij}^{2} t_j^2}{r_j^2 (1-\widetilde{h}_{\lambda ii})^2}.
\end{align*}
Because  $z_i/q_i$ and $z_j/q_j$ are independent for $i\neq j$ and $\E[z_i/q_i]=0$, 
\begin{align*}
\E\left[\frac{z_i(r_i\mu_i-\bx_{i}^\top \bbeta_{\lambda})}{q_i(1-\widetilde{h}_{\lambda ii})} \frac{z_j( r_j\mu_j-\bx_{j}^\top \bbeta_{\lambda} )}{q_j(1-\widetilde{h}_{\lambda jj})} \right] = 0.
\end{align*}
Finally,
\begin{align*}
\E&\left[\frac{z_i \bx_i^\top (\widetilde{\bX}^\top \widetilde{\bX} + \lambda \bI)^{-1} \widetilde{\bX}_{-i}^{\top} (\bz\bt/\bq)_{-i}}{q_i(1-\widetilde{h}_{\lambda ii})}  \frac{z_j \bx_j^\top (\widetilde{\bX}^\top \widetilde{\bX} + \lambda \bI)^{-1} \widetilde{\bX}_{-j}^{\top} (\bz\bt/\bq)_{-j}}{q_j(1-\widetilde{h}_{\lambda jj})} \right]
\\ 
&=\E\left[\frac{z_i r_i \widetilde\bx_i^\top (\widetilde{\bX}^\top \widetilde{\bX} + \lambda \bI)^{-1} \widetilde{\bX}_{-i}^{\top} (\bz\bt/\bq)_{-i}}{q_i(1-\widetilde{h}_{\lambda ii})}  \frac{z_j r_j\widetilde\bx_j^\top (\widetilde{\bX}^\top \widetilde{\bX} + \lambda \bI)^{-1} \widetilde{\bX}_{-j}^{\top} (\bz\bt/\bq)_{-j}}{q_j(1-\widetilde{h}_{\lambda jj})} \right]\\
&=\E\left[\frac{r_ir_j \widetilde h_{\lambda ij}^2 t_it_j}{q_i^2q_j^2(1-\widetilde{h}_{\lambda ii})(1-\widetilde{h}_{\lambda jj})}\right]
\\& =
\frac{\widetilde h_{\lambda ij}^2 t_it_j}{r_ir_j(1-\widetilde{h}_{\lambda ii})(1-\widetilde{h}_{\lambda jj})}.
\end{align*}
Combining all the above, we have
\begin{align*}
\frac{1}{n^2}& \left(\sum_{i=1}^n \frac{(\mu_i-\widetilde\bx_{i}^\top \bbeta_\lambda)^2}{(1-\widetilde{h}_{\lambda ii})^2} + \sum_{i=1}^{n} \sum_{j\neq i} \frac{(\widetilde h_{\lambda ij} t_j)^2}{r_j^2 (1-\widetilde{h}_{\lambda ii})^2} + \frac{\widetilde h_{\lambda ij}^2 ~t_i t_j}{r_i r_j(1-\widetilde{h}_{\lambda ii})(1-\widetilde{h}_{\lambda jj})} \right)
\\ & = 
\frac{1}{n^2} \left( \sum_{i=1}^n \frac{(\mu_i-\widetilde\bx_{i}^\top \bbeta_\lambda)^2}{(1-\widetilde{h}_{\lambda ii})^2} + \sum_{i=1}^{n-1} \sum_{j = i+1}^n \widetilde h_{\lambda ij}^2 \left(\frac{ t_j}{r_j(1-\widetilde{h}_{\lambda ii})} + \frac{ t_i}{r_i(1-\widetilde{h}_{\lambda jj})}\right)^2 \right).
\end{align*}
\end{proof}

\begin{proof}[Proof of \Cref{thm:asym_dist_HT_wo_lim}]
We first consider the case where $\norm{\widetilde{\bX}\bbeta_{\lambda} - \bmu}{2} \to \infty$.
We write
\begin{align*}
\widehat{\tau}_{\textup{LHT}}^*
&= \frac{1}{n} \sum_{i=1}^n \frac{z_i}{q_i}
\bigl( y_i - \bx_i^\top \bbeta_{\lambda} \bigr), ~~
\widehat{\tau}_{\textup{LHT}}
= \frac{1}{n} \sum_{i=1}^n \frac{z_i}{q_i}
\bigl( y_i - \bx_i^\top \widehat{\bbeta}^{(-i)}_{\lambda} \bigr),
\end{align*}
and decompose
\[
\frac{\sqrt{n}(\widehat{\tau}_{\textup{LHT}} - \tau)}
{\norm{\widetilde{\bX}\bbeta_{\lambda} - \bmu}{2} / \sqrt{n}}
=
A_n + B_n,
\]
where
\begin{align*}
A_n
&= 
\frac{\sqrt{n}(\hat \tau_{\textup{LHT}}^* - \tau)}{\norm{\widetilde{\bX}\bbeta_{\lambda} - \bmu}{2} / \sqrt{n}}, ~~
B_n= 
\frac{\sqrt{n}(\hat \tau_{\textup{LHT}} - \widehat{\tau}_{\textup{LHT}}^*)}
{\norm{\widetilde{\bX}\bbeta_{\lambda} - \bmu}{2} / \sqrt{n}}.
\end{align*}
  
\paragraph{Step 1: Oracle CLT.}
Under independent treatment assignment with probabilities bounded away from zero and one, bounded outcomes, and uniformly bounded covariates, the random variables $\frac{z_i}{q_i}(y_i - \bx_i^\top \bbeta_{\lambda})$ are bounded. Therefore
\[
\Var\!\left(\sqrt{n}(\widehat{\tau}_{\textup{LHT}}^* - \tau)\right)
=
\frac{1}{n} \sum_{i=1}^n
\Var(\frac{z_i}{q_i}(y_i - \bx_i^\top \bbeta_{\lambda}))
=
\frac{1}{n}
\norm{\widetilde{\bX}\bbeta_{\lambda} - \bmu}{2}^2,
\]
by definition of $\widetilde{\bX}\bbeta_{\lambda} - \bmu$. Moreover, the assumptions imply that the random variables $\frac{z_i}{q_i}(y_i - \bx_i^\top \bbeta_{\lambda}) - \tau_i$ are bounded. Let $C_1$ be a constant such that for all $i \in \mathbb{N}$, $\abs{\frac{z_i}{q_i}(y_i - \bx_i^\top \bbeta_{\lambda}) - \tau_i} \leq C_1$. Then for any constant $\delta > 0$,
\begin{align*}
&
\frac{1}{\Var\!\left(n(\widehat{\tau}_{\textup{LHT}}^* - \tau)\right)^{1+\delta/2}} \sum_{i=1}^n \E \left[\abs{\frac{z_i}{q_i}(y_i - \bx_i^\top \bbeta_{\lambda}) - \tau_i}^{2+\delta} \right]
\\ & \leq
\frac{C_1^{\delta}}{\Var\!\left(n(\widehat{\tau}_{\textup{LHT}}^* - \tau)\right)^{1+\delta/2}} \sum_{i=1}^n \E \left[\abs{\frac{z_i}{q_i}(y_i - \bx_i^\top \bbeta_{\lambda}) - \tau_i}^{2} \right]
\\ &
= \frac{C_1^{\delta}}{\Var\!\left(n(\widehat{\tau}_{\textup{LHT}}^* - \tau)\right)^{\delta/2}} = \frac{C_1^{\delta}}{\norm{\widetilde{\bX}\bbeta_{\lambda} - \bmu}{2}^{\delta}}.
\end{align*}
Thus since $\norm{\widetilde{\bX}\bbeta_{\lambda} - \bmu}{2} \to \infty$ and $C_1$ is a constant, we have
\[
\lim_{n \to \infty} \frac{1}{\Var\!\left(n(\widehat{\tau}_{\textup{LHT}}^* - \tau)\right)^{1+\delta/2}} \sum_{i=1}^n \E \left[\abs{\frac{z_i}{q_i}(y_i - \bx_i^\top \bbeta_{\lambda}) - \tau_i}^{2+\delta} \right] = 0.
\]
Therefore, by the Lyapunov central limit theorem (\cite{billingsley2012probability}),
\[
A_n
=
\frac{\sqrt{n}(\widehat{\tau}_{\textup{LHT}}^* - \tau)}
{\norm{\widetilde{\bX}\bbeta_{\lambda} - \bmu}{2} / \sqrt{n}}
\dto
\mathcal{N}(0,1).
\]

\paragraph{Step 2: Negligibility of the plug-in error.}
Then
\[
\widehat{\tau}_{\textup{LHT}}^* - \widehat{\tau}_{\textup{LHT}} 
=
\frac{1}{n}\sum_{i=1}^n
\frac{z_i}{q_i}\,\bx_i^\top (\widehat{\bbeta}^{(-i)}_{\lambda} - \bbeta_{\lambda}).
\]
Therefore by triangle inequality and Cauchy–Schwarz inequality, we have
\begin{align*}
\abs{\widehat{\tau}_{\textup{LHT}}^* - \widehat{\tau}_{\textup{LHT}}} & 
\leq \frac{1}{n} \abs{\sum_{i=1}^n
\frac{z_i}{q_i}\,\bx_i^\top (\widehat{\bbeta}_{\lambda} - \bbeta_{\lambda})} + \frac{1}{n} \abs{\sum_{i=1}^n
\frac{z_i}{q_i}\,\bx_i^\top (\widehat{\bbeta}^{(-i)}_{\lambda} - \widehat{\bbeta}_{\lambda})} 
\\ &
= \frac{1}{n} \abs{(\bX^\top(\bz/\bq))^{\top} (\widehat{\bbeta}_{\lambda} - \bbeta_{\lambda})} + \frac{1}{n} \abs{\sum_{i=1}^n
\frac{z_i}{q_i}\,\bx_i^\top (\widehat{\bbeta}^{(-i)}_{\lambda} - \widehat{\bbeta}_{\lambda})} 
\\ &
\leq \frac{1}{n} \norm{(\bX^\top(\bz/\bq))}{2} \norm{ \widehat{\bbeta}_{\lambda} - \bbeta_{\lambda}}{2} + \frac{1}{n} \abs{\sum_{i=1}^n
\frac{z_i}{q_i}\,\bx_i^\top (\widehat{\bbeta}^{(-i)}_{\lambda} - \widehat{\bbeta}_{\lambda})}.
\end{align*}
We have
\begin{align}
\nonumber
\norm{ \widehat{\bbeta}_{\lambda} - \bbeta_{\lambda}}{2}
& = 
\norm{(\widetilde{\bX}^\top \widetilde{\bX} + \lambda \bI)^{-1}
  \widetilde{\bX}^\top (\bz \bt/\bq)}{2} = O(n^{-1}) \norm{
  \widetilde{\bX}^\top (\bz \bt/\bq)}{2}
  \\ & = \label{eq:beta-hat-beta-diff-loora-ht}
  O_p(k^{1/2} n^{-1/2}),
\end{align}
where the last line holds using \Cref{ass:boundedness,ass:positivity} and \Cref{lemma:concentrate_vec_sum}. Moreover, invoking \Cref{lemma:concentrate_vec_sum} also gives $\norm{(\bX^\top(\bz/\bq))}{2} = O_{p}(\sqrt{n})$, concluding $\frac{1}{n} \norm{(\bX^\top(\bz/\bq))}{2} \norm{ \widehat{\bbeta}_{\lambda} - \bbeta_{\lambda}}{2} = O_{p}(1/n)$.
Moreover,
\begin{align*}
\hat{\bbeta}_{\lambda} - \hat{\bbeta}^{(-i)}_{\lambda}
  =
  \frac{
    (\widetilde{\bX}^\top \widetilde{\bX} + \lambda \bI)^{-1}
    \widetilde{\bx}_i
    (\widetilde{y}_i - \widetilde{\bx}_i^\top \hat{\bbeta}_{\lambda})
  }{
    1 - \widetilde{h}_{\lambda ii}
  } = O(n^{-1}),
\end{align*}
Therefore
\begin{align*}
\frac{1}{n} \abs{\sum_{i=1}^n
\frac{z_i}{q_i}\,\bx_i^\top (\widehat{\bbeta}^{(-i)}_{\lambda} - \widehat{\bbeta}_{\lambda})} = O(n^{-1})
\end{align*}

Combining the above, we have
\[
\abs{\widehat{\tau}_{\textup{LHT}}^* - \widehat{\tau}_{\textup{LHT}}} = O_p(n^{-1}).
\]

\paragraph{Step 3: Conclusion.}
Since $A_n \dto \mathcal{N}(0,1)$ and $B_n \pto 0$, Slutsky's Theorem implies
\[
\frac{\sqrt{n}(\widehat{\tau}_{\textup{LHT}} - \tau)}
{\norm{\widetilde{\bX}\bbeta_{\lambda} - \bmu}{2} / \sqrt{n}}
\dto
\mathcal{N}(0,1).
\]
\end{proof}

\begin{proof}[Proof of \Cref{lemma:var-of-HLA}]
To observe consistency, note that
\[
\frac{1}{n} \sum_{i=1}^n \frac{z_i}{q_i} \bigl(y_i - \bx_i^\top \hat{\bgamma}_{n}^{(-i)}\bigr)
= \frac{1}{n} \sum_{i=1}^n \frac{z_i}{q_i} \bigl(y_i - \bx_i^\top \bgamma_{n}\bigr)
+ \frac{1}{n} \sum_{i=1}^n \frac{z_i}{q_i} \bx_i^\top \bigl( \bgamma_{n} - \hat{\bgamma}_{n}^{(-i)}\bigr).
\]
Since $\bgamma_n$ is deterministic, the first term on the right-hand side above is an unbiased estimate of $\tau$. Moreover, for the second term, for some constant $C$, we have
\begin{align}
\nonumber
\abs{\frac{1}{n} \sum_{i=1}^n \frac{z_i}{q_i} \bx_i^\top \bigl( \bgamma_{n} - \hat{\bgamma}_{n}^{(-i)}\bigr)}
& \leq 
\frac{1}{n} \sum_{i=1}^n \abs{\frac{z_i}{q_i}} \norm{\bx_i}{2} \norm{\bgamma_{n} - \hat{\bgamma}_{n}^{(-i)}}{2}
\\ & \leq
\label{eq:err-term-go-to-zero}
C \max_{i \in [n]}\norm{\hat{\bgamma}_{n}^{(-i)} - \bgamma_n}{2},
\end{align}
where the second inequality follows from \Cref{ass:boundedness}. Therefore, by \eqref{eq:HLT-assumption},
\begin{align*}
\frac{1}{n} \sum_{i=1}^n \frac{z_i}{q_i} \bx_i^\top \bigl( \bgamma_{n} - \hat{\bgamma}_{n}^{(-i)}\bigr) \pto 0,
\end{align*}
and $\hatt_{\textup{HLA}}$ is consistent. For the variance, note that 
\[
\hatt_{\textup{HLA}} - \tau = \frac{1}{n} \sum_{i=1}^n \frac{z_i}{q_i} \bigl(r_i\mu_i - \bx_i^\top \hat{\bgamma}_{n}^{(-i)}\bigr).
\]
Therefore
\begin{align}
\label{eq:HLA-dif-decomp}
\sqrt{n}(\hatt_{\textup{HLA}} - \tau) & = \frac{1}{\sqrt{n}} \sum_{i=1}^n \frac{z_i}{q_i} \left(r_i\mu_i - \bx_i^\top \bgamma_{n}\right) + \frac{1}{\sqrt{n}} \sum_{i=1}^n \frac{z_i}{q_i} \left(\bx_i^\top \bgamma_{n}- \bx_i^\top \hat{\bgamma}_{n}\right)
\\ & + \nonumber
\frac{1}{\sqrt{n}} \sum_{i=1}^n \frac{z_i}{q_i} \left(\bx_i^\top \hat{\bgamma}_{n} - \bx_i^\top \hat{\bgamma}_{n}^{(-i)}\right).
\end{align}
By Cauchy–Schwarz inequality,
\begin{align*}
\frac{1}{\sqrt{n}} \sum_{i=1}^n \frac{z_i}{q_i} \bigl(\bx_i^\top \bgamma_{n} - \bx_i^\top \hat{\bgamma}_{n}\bigr) & = \frac{1}{\sqrt{n}}\left( \sum_{i=1}^n \frac{z_i}{q_i} \bx_i \right)^\top (\bgamma_{n} - \hat{\bgamma}_{n})
\\ & \leq
\frac{1}{\sqrt{n}} \norm{\sum_{i=1}^n \frac{z_i}{q_i} \bx_i}{2} \norm{\bgamma_{n} - \hat{\bgamma}_{n}}{2}.
\end{align*}
Note that by \Cref{ass:boundedness,ass:positivity} and \Cref{lemma:concentrate_vec_sum}, $\norm{\sum_{i=1}^n \frac{z_i}{q_i} \bx_i}{2} = O_p(\sqrt{n})$.
Therefore since $\norm{\widetilde{\bX} \bgamma_{n} - \bmu}{2} = \Omega(\sqrt{n})$ by \eqref{eq:HLT-assumption},
\begin{align}
\label{eq:HLA-dif-decomp-1}
\frac{\frac{1}{\sqrt{n}} \sum_{i=1}^n \frac{z_i}{q_i} \bigl(\bx_i^\top \bgamma_{n} - \bx_i^\top \hat{\bgamma}_{n}\bigr)}{\norm{\widetilde{\bX} \bgamma_{n} - \bmu}{2} / \sqrt{n}} \pto 0.
\end{align}
Moreover by triangle inequality and Cauchy–Schwarz inequality,
\begin{align*}
\abs{\frac{1}{\sqrt{n}} \sum_{i=1}^n \frac{z_i}{q_i} \bigl(\bx_i^\top \hat{\bgamma}_{n} - \bx_i^\top \hat{\bgamma}_{n}^{(-i)}\bigr)} & \leq
\frac{1}{\sqrt{n}} \sum_{i=1}^n \norm{\frac{z_i}{q_i} \bx_i}{2} \norm{\hat{\bgamma}_{n} - \hat{\bgamma}_{n}^{(-i)}}{2} 
\\ & \leq
\frac{1}{\sqrt{n}} \left( \sum_{i=1}^n \norm{\frac{z_i}{q_i} \bx_i}{2} \right) \max_{i \in [n]} \norm{\hat{\bgamma}_{n} - \hat{\bgamma}_{n}^{(-i)}}{2}
\\ & = o_p(1),
\end{align*}
where the last equality follows from \Cref{ass:boundedness,ass:positivity} and \eqref{eq:HLT-assumption}. Therefore since $\norm{\widetilde{\bX} \bgamma_{n} - \bmu}{2} = \Omega(\sqrt{n})$,
\begin{align}
\label{eq:HLA-dif-decomp-2}
\frac{\frac{1}{\sqrt{n}} \sum_{i=1}^n \frac{z_i}{q_i} \bigl(\bx_i^\top \hat{\bgamma}_{n} - \bx_i^\top \hat{\bgamma}_{n}^{(-i)}\bigr)}{\norm{\widetilde{\bX} \bgamma_{n} - \bmu}{2} /\sqrt{n}} \pto 0.
\end{align}
Now similar to the proof of \Cref{thm:asym_dist_HT_wo_lim}, the Lyapunov central limit theorem (\cite{billingsley2012probability}) implies that 
\begin{align}
\label{eq:HLA-dif-decomp-3}
\frac{\frac{1}{\sqrt{n}} \sum_{i=1}^n \frac{z_i}{q_i} \bigl(r_i\mu_i - \bx_i^\top \bgamma_{n}\bigr)}{\norm{\widetilde{\bX} \bgamma_{n} - \bmu}{2} /\sqrt{n}} \dto \mathcal{N}(0,1).
\end{align}
Therefore by \eqref{eq:HLA-dif-decomp}, \eqref{eq:HLA-dif-decomp-1}, \eqref{eq:HLA-dif-decomp-2}, and \eqref{eq:HLA-dif-decomp-3}, and Slutsky's theorem,
\[
\frac{\sqrt{n}(\hatt_{\textup{HLA}} - \tau)}{\norm{\widetilde{\bX} \bgamma_{n} - \bmu}{2} /\sqrt{n}} \dto \mathcal{N}(0,1).
\]
\end{proof}

\begin{proof}[Proof of \Cref{thm:loora-ht-efficiency}]
Note that $\bbeta_{0}=\argmin_{\bgamma} \norm{\widetilde{\bX} \bgamma - \bmu}{2}$ where $\bbeta_{0}$ is $\bbeta_{\lambda}$ for $\lambda=0$. Therefore the sequence $\norm{\widetilde{\bX} \bbeta_0 - \bmu}{2}$ is pointwise smaller than or equal to the sequence $\norm{\widetilde{\bX} \bgamma_n - \bmu}{2}$ for any choice of $\bgamma_1,\bgamma_2,\ldots$. Therefore the results is implied by \Cref{lemma:var-of-HLA,thm:asym_dist_HT_wo_lim}.
\end{proof}

\begin{proof}[Proof of \Cref{thm:loora-ht-est-var}]

Let $s_i
  = z_i \left(\frac{y_i - \bx_i^\top \bbeta_{\lambda}}
     {q_i} \right) - \,\tau$,
     \[
     \widehat{V}_{\textup{LHT}} = \frac{1}{n} \sum_{i=1}^n \widehat{s}_i^2, ~~ \text{and} ~~ \widetilde{V}_{\textup{LHT}} = \frac{1}{n} \sum_{i=1}^n s_i^2.
     \]
Let $\tau_i = y^{(1)}_i - y^{(0)}_i$. Then
\begin{align*}
\E\left[ \widetilde{V}_{\textup{LHT}} \right] & = \frac{1}{n} \sum_{i=1}^n \E\left[\left(z_i \frac{y_i - \bx_i^\top \bbeta_{\lambda}}{q_i} - \tau \right)^2 \right]
\\ & =
\frac{1}{n} \sum_{i=1}^n \E \left[\left(z_i \frac{y_i - \bx_i^\top \bbeta_{\lambda}}{q_i} - \tau_i \right)^2 \right] 
\\ & + 
2 \E \left[z_i \frac{y_i - \bx_i^\top \bbeta_{\lambda}}{q_i} - \tau_i \right] (\tau_i - \tau) + (\tau_i - \tau)^2.
\end{align*}
Therefore since
\begin{align*}
\E \left[\left(z_i \frac{y_i - \bx_i^\top \bbeta_{\lambda}}{q_i} - \tau_i \right)^2 \right] = (\mu_i - \widetilde{\bx}_i^\top \bbeta_{\lambda})^2, ~~ \text{and} ~~ \E \left[z_i \frac{y_i - \bx_i^\top \bbeta_{\lambda}}{q_i} - \tau_i \right] = 0,
\end{align*}
we have
\begin{align*}
\E\left[ \widetilde{V}_{\textup{LHT}} \right] = \frac{1}{n} \norm{\widetilde{\bX} \bbeta_{\lambda} - \bmu}{2}^2 + \frac{1}{n} \sum_{i=1}^n (\tau_i - \tau)^2 \geq \frac{1}{n} \norm{\widetilde{\bX} \bbeta_{\lambda} - \bmu}{2}^2 = n \E\left[ (\widehat{\tau}_{\textup{LHT}}^* - \tau)^2 \right].
\end{align*}
Therefore $\widetilde{V}_{\textup{LHT}}$ is a conservative estimate for the variance of $\sqrt{n}\widehat{\tau}_{\textup{LHT}}^*$. Also note that we have $\Var(\tilde{V}_{\textup{LHT}}) = \frac{1}{n^2}\sum_{i=1}^n\Var(s_i^2) = O(n^{-1})$. So Using Chebyshev we get
\begin{equation*}
    \tilde{V}_{\textup{LHT}} -\E[\tilde{V}_{\textup{LHT}}] \pto 0.
\end{equation*}
For $\hat{V}_{\textup{LHT}}$ we can write:
\begin{align*}
    \hat{V}_{\textup{LHT}} - \tilde{V}_{\textup{LHT}} = \frac{1}{n}\sum_{i=1}^n\Delta_i(\hat{s}_i+s_i),
\end{align*}
where $\Delta_i =  \hat{s}_i - s_i = \frac{z_i}{q_i}\bx_i^\top(\bbeta_\lambda - \hat{\bbeta}_{\lambda}^{(-i)}) - (\hatt_{\textup{LHT}} - \tau)$. Using the proof of \Cref{thm:asym_dist_HT_wo_lim}, it follows that $\norm{\bbeta_\lambda - \hat{\bbeta}_\lambda^{(-i)}}{2} = O_p(n^{-1/2})$, and $\hatt_{\textup{LHT}} - \tau = O_{p}(n^{-1/2})$. Also we have $s_i = O(1)$.
These two results imply that $\hat{s}_i = O_{p}(1)$. These in conclusoin give $\hat{V}_{\textup{LHT}}  - \tilde{V}_{\textup{LHT}} \pto 0$, and therefore we get $\hat{V}_{\textup{LHT}} - \E[\tilde{V}_{\textup{LHT}}] \pto 0$. Finally, since $V_n=\E[\tilde{V}_{\textup{LHT}}]$ is bounded away from zero for all $n \in \mathbb{N}$, $\hat{V}_{\textup{LHT}} - \E[\tilde{V}_{\textup{LHT}}] \pto 0$ implies
\[
\frac{\hat{V}_{\textup{LHT}}}{V_n} \pto 1.
\]
\end{proof}

\subsection{Omitted Proofs of \Cref{sec:reg-adj-dim}}
\begin{proof}[Proof of \Cref{thm:LOORA_DiM_Var}]
We first show that the estimator is unbiased. Let $\lam{i} = (\bX_{-i}^\top \bX_{-i} + \lambda \bI)^{-1}$. By law of total expectation, we have
\begin{align*}
& \E\left[v_i z_i (y_i - \bx_{i}^\top \widehat{\bbeta}_{\lambda}^{(-i)}) \right] 
\\ & = 
\frac{n_T}{n}\E \left[\frac{1}{n_T}(y^{(1)}_i - \bx_i^\top \widehat{\bbeta}_{\lambda}^{(-i)}) | z_i = +1 \right] - \frac{n_C}{n} \E \left[ \frac{1}{n_C} (y^{(0)}_i - \bx_i^\top \widehat{\bbeta}_{\lambda}^{(-i)}) | z_i = -1 \right]
\\ & =
\frac{1}{n}(y^{(1)}_i - y^{(0)}_i) - \frac{n_Cn_T(n-1)}{n^2} \E \left[ \bx_i^\top\lam{i} \bX_{-i}^{\top} (\bf^{(T)}  \by)_{-i} | z_i = +1 \right]
\\ & + 
\frac{n_Cn_T(n-1)}{n^2} \E \left[ \bx_i^\top \lam{i} \bX_{-i}^{\top} (\bf^{(C)}  \by)_{-i} | z_i = -1 \right]\,.
\end{align*}
Now by linearity of expectation,
\begin{align*}
& \E \left[ \bx_i^\top \lam{i} \bX_{-i}^{\top} (\bf^{(T)}  \by)_{-i} | z_i = +1 \right] = 
(\bx_i^\top \lam{i}\bX_{-i}^{\top}) \E \left[(\bf^{(T)}  \by)_{-i} | z_i = +1\right]\,.
\end{align*}
Moreover
\begin{align*}
\E \left[(\bf^{(T)}  \by)_{-i} | z_i = +1 \right]
& = \frac{n_T-1}{n-1} \cdot \frac{\by^{(1)}_{-i}}{n_T(n_T - 1)} + \frac{n_C}{n-1} \cdot \frac{\by^{(0)}_{-i}}{n_C^2}
= \frac{(\by^{(1)}/n_T + \by^{(0)}/n_C)_{-i}}{n-1}\,.
\end{align*}
Similarly
\[
\E \left[(\bf^{(C)}  \by)_{-i} | z_i = -1 \right] = \frac{(\by^{(1)}/n_T + \by^{(0)}/n_C)_{-i}}{n-1}\,.
\]
Therefore
\[
\E\left[v_i z_i (y_i - \bx_{i}^\top \widehat{\bbeta}_{\lambda}^{(-i)}) \right] = \frac{1}{n} (y^{(1)}_i - y^{(0)}_i),
\]
and LOORA-DM is an unbiased estimator. 
We define $\tilde{\bt}^{(-i)}$ and $\bv^{(-i)}$ as the following.
\begin{align*}
\tilde{\bt}^{(-i)} = \begin{cases}
    \tilde{\bt}^{(1)} & \text{if } z_i = 1, \\
    \tilde{\bt}^{(0)} & \text{if } z_i = -1.
\end{cases}
\end{align*}
\begin{align*}
v^{(-i)}_j = 
\begin{cases} 
    \dfrac{1}{n_T - 1} & \text{if } i \in T \text{ and } j \in T, \\[1.5ex]
    \dfrac{1}{n_C} & \text{if } i \in T \text{ and } j \in C, \\[1.5ex]
    \dfrac{1}{n_T} & \text{if } i \in C \text{ and } j \in T, \\[1.5ex]
    \dfrac{1}{n_C - 1} & \text{if } i \in C \text{ and } j \in C.
\end{cases} 
\end{align*}
Algebraic manipulations yield
\begin{align}
    \widetilde{\by}^{(-i)} = \frac{\sqrt{n_C n_T}}{n} \tilde{\bmu} + \frac{\bv^{(-i)}  \bz  \widetilde{\bt}^{(-i)}}{n}. \label{eq:simple_value}
\end{align}
Moreover, irresepective of the assignment of unit $i$,
\begin{align*}
v_iz_i (y_i - \bx_{i}^\top \widehat{\bbeta}_{\lambda}^{(-i)}) - \frac{1}{n}(y^{(1)}_i - y^{(0)}_i)
= v_iz_i(\frac{\sqrt{n_C n_T}}{n} \tilde{\mu}_i -\bx_{i}^\top \widehat{\bbeta}_{\lambda}^{(-i)}) \,.
\end{align*}
Therefore
\begin{equation}
\Var(\widehat\tau)
\;=\;
\E\bigl[(\widehat\tau - \tau)^2\bigr]
\;=\;
\E\Biggl[\Bigl(\sum_{i=1}^n v_i\,z_i\,
\Bigl(\frac{\sqrt{n_C n_T}}{n} \tilde{\mu}_i-\bx_i^\top\widehat\bbeta^{(-i)}\Bigr)\Bigr)^2\Biggr].
\label{eq:mse:def}
\end{equation}
By \eqref{eq:simple_value},
\[
\widehat\bbeta_{\lambda}^{(-i)}
\;=\;
\lam{i}\,
\bX_{-i}^\top\,
\left(\frac{\sqrt{n_C n_T}}{n} \tilde{\bmu}+\frac{\bv^{(-i)}\bz\widetilde\bt^{(-i)}}{n}\right)_{-i}.
\]
We denote $\bLambda_{\lambda} = (\bX^\top \bX + \lambda \bI)^{-1}$.
Substituting these into \eqref{eq:mse:def} and a short calculation yields
\begin{align}
\Var(\widehat\tau)
&=
\E\Biggl[\Biggl(\sum_{i=1}^n v_i\,z_i\,
\Bigl\{
\frac{\sqrt{n_C n_T}}{n} \tilde{\bmu}
-\bx_i^\top\,\lam{i}\,
\bX_{-i}^\top\,
\left(\frac{\sqrt{n_C n_T}}{n} \tilde{\bmu}+\frac{\bv^{(-i)}\bz\widetilde\bt^{(-i)}}{n}\right)_{-i}
\Bigr\}\Biggr)^2\Biggr]
\nonumber
\\ & =
\E\Biggl[\Biggl(\sum_{i=1}^n v_i\,z_i\,
\Bigl\{
\frac{\sqrt{n_C n_T}}{n}\frac{\widetilde\mu_i-\bx_i^\top\bbeta_{\lambda}}{(1-h_{\lambda ii})}
-\frac{\bx_i^\top\,\bLambda_{\lambda}\,
\bX_{-i}^\top}{(1-h_{\lambda ii})}\,
\left(\frac{\bv^{(-i)}\bz\widetilde\bt^{(-i)}}{n}\right)_{-i}
\Bigr\}\Biggr)^2\Biggr],
\label{eq:mse:variance}
\end{align}
where the second equality follows from \Cref{lemma:sol-change-in-row-removal,cor:sum-sqr-err-with-removal}.
For a more compact notation, we denote $\bar{h}_i = (1-h_{\lambda ii})$ for the rest of this proof. We decompose \eqref{eq:mse:variance} into three components and compute the corresponding expectations separately.
\begin{enumerate}
\item $T_1
=
\E\left[\sum_{i,j \in [n]}\frac{(n_C n_T) v_i\,z_i\,v_j\,z_j\,}{n^2}
\frac{\widetilde\mu_i-\bx_i^\top\bbeta_{\lambda}}
     {\bar{h}_i}
\,
\frac{\widetilde\mu_j-\bx_j^\top\bbeta_{\lambda}}
     {\bar{h}_j}
\right]$.

\item $T_2
=
-2\,
\E\left[\sum_{i,j \in [n]}\frac{\sqrt{n_C n_T}v_i\,z_i\,v_j\,z_j\,}{n^2}
\frac{\widetilde\mu_i-\bx_i^\top\bbeta_{\lambda}}
     {\bar{h}_i \bar{h}_j}
\;
\left(\bx_j^\top\,\bLambda_{\lambda}\,
     \bX_{-j}^\top\,
     (\bv^{(-j)}\bz\widetilde\bt^{(-j)})_{-j}\right)
\right]$.

\item $T_3
=
\E\left[\sum\limits_{i,j \in [n]} \!\! \frac{v_i\,z_i\,v_j\,z_j\,}{n^2\bar{h}_i\bar{h}_j} \!\!
\left(\bx_i^\top\,\bLambda_{\lambda}\,
      \bX_{-i}^\top\,
      (\bv^{(-i)}\bz\widetilde\bt^{(-i)})_{-i}\right) 
    \! \!
\left(\bx_j^\top\,\bLambda_{\lambda}\,
      \bX_{-j}^\top\,
      (\bv^{(-j)}\bz\widetilde\bt^{(-j)})_{-j}\right)
\right]$.
\end{enumerate}
One can easily observe that
\[
\Var(\widehat\tau) \;=\; T_1 + T_2 + T_3.
\]
We start by calculating $T_1$.
Note that
$$
T_1 = \E\Biggl[\Biggl(\sum_{i=1}^n \frac{\sqrt{n_C n_T}v_i\,z_i\,}{n} \frac{\wt{\mu}_i-\bx_i^\top\bbeta_{\lambda}}{\bar{h}_i}\Biggr)^2\Biggr].
$$
We denote $R_i = \frac{\wt{\mu}_i-\bx_i^\top \bbeta_{\lambda}}{\bar{h}_i}$. Since $\wt{\bmu}$ and $\bbeta_{\lambda}$ are fixed vectors and do not depend on the treatment assignment vector $\bz$, by linearity of expectation,
$$
T_1 = \frac{n_C n_T}{n^2} \E\Biggl[\Biggl(\sum_{i=1}^n v_i z_i R_i\Biggr)^2\Biggr] = \frac{n_C n_T}{n^2} \sum_{i=1}^n \sum_{j=1}^n R_i R_j \E[v_i z_i v_j z_j].
$$
For $i=j$, $\E[(v_i z_i)^2] = \E[v_i^2]$. Therefore by definition,
$$
\E[v_i z_i v_j z_j] = \E[v_i^2] = \frac{n_T}{n} \left(\frac{1}{n_T}\right)^2 + \frac{n_C}{n} \left(\frac{1}{n_C}\right)^2 = \frac{1}{n n_T} + \frac{1}{n n_C} = \frac{n_C+n_T}{n n_T n_C} = \frac{1}{n_T n_C}.
$$
For $i\neq j$, we can write $\E[v_i z_i v_j z_j]$ by considering the joint assignment of units $i$ and $j$.
\begin{align*}
\E[v_i z_i v_j z_j] &= \P(d_i=1,d_j=1) \left(\frac{1}{n_T}\cdot\frac{1}{n_T}\right) + \P(d_i=1,d_j=0) \left(\frac{1}{n_T}\cdot\frac{-1}{n_C}\right) \\
&+ \P(d_i=0,d_j=1) \left(\frac{-1}{n_C}\cdot\frac{1}{n_T}\right) + \P(d_i=0,d_j=0) \left(\frac{-1}{n_C}\cdot\frac{-1}{n_C}\right) \\
&= \frac{n_T(n_T-1)}{n(n-1)}\frac{1}{n_T^2} - \frac{n_T n_C}{n(n-1)}\frac{1}{n_T n_C} - \frac{n_C n_T}{n(n-1)}\frac{1}{n_C n_T} + \frac{n_C(n_C-1)}{n(n-1)}\frac{1}{n_C^2}
\\ & = -\frac{1}{(n-1)n_T n_C}.
\end{align*}
Substituting these into the expression for $T_1$, we have
\begin{align*}
T_1 &= \frac{n_T n_C}{n^2} \left( \sum_{i=1}^n R_i^2 \E[v_i^2] + \sum_{i \neq j} R_i R_j \E[v_i z_i v_j z_j] \right) \\
&= \frac{1}{n^2} \left( \sum_{i=1}^n R_i^2 - \frac{1}{n-1} \sum_{i,j \in [n] : i \neq j} R_i R_j \right).
\end{align*}
Using the identity $\sum_{i,j \in [n] : i \neq j} R_i R_j = \left(\sum_{i \in [n]} R_i\right)^2 - \sum_{i\in[n]} R_i^2$, simple calculations yield
\begin{align*}
T_1 
&= \frac{1}{n(n-1)} \sum_{i=1}^n \left( R_i - \frac{1}{n}\sum_{j=1}^n R_j \right)^2
\\ & =
\frac{1}{n(n-1)} \sum_{i=1}^n \left( \frac{\wt\mu_i-\bx_i^\top\bbeta_{\lambda}}{(1- h_{\lambda ii})} - \frac{1}{n}\sum_{j=1}^n \frac{\wt\mu_j-\bx_j^\top\bbeta_{\lambda}}{(1- h_{\lambda jj})} \right)^2.
\end{align*}
Recall that $\bH_{\lambda} = \bX^\top \bLambda_{\lambda}\bX$. 
We have
\begin{align*}
T_2 & = -2\,\E\left[\Biggl(\sum_{i=1}^n \frac{\sqrt{n_T n_C} v_i z_i R_i}{n }\Biggr) \Biggl(\sum_{j=1}^n \frac{v_j z_j}{n \bar{h}_j} \bx_j^\top \blambda_{\lambda} \bX_{-j}^\top (\bv^{(-j)}\bz\wt{\bt}^{(-j)})_{-j}\Biggr)\right]
\\ & =
\frac{-2\sqrt{n_T n_C}}{n^2} \sum_{i=1}^n \sum_{j=1}^n \frac{R_i}{\bar{h}_j} \sum_{k \in [n]: k \neq j} h_{\lambda jk} \, \E\left[ v_i z_i v_j z_j v_k^{(-j)} z_k \wt{t}_k^{(-j)} \right]
\end{align*}
We denote $E_{ijk} := \E\left[ v_i z_i v_j z_j v_k^{(-j)} z_k \wt{t}_k^{(-j)} \right]$. One can easily check that $E_{iik} =0$ for $i\neq k$. Therefore we consider two cases: 1) $i=k$, and 2) $i\neq k$, $i\neq j$---note that the above expression guarantees $j\neq k$. We have
\begin{align}
\label{eq:e-iji}
E_{iji} = \E\left[ v_i v^{(-j)}_i v_j z_j \wt{t}_i^{(-j)} \right] = \frac{\wt\mu_i}{(n-1)\sqrt{n_T n_C}}.
\end{align}
For $i\neq k$, $i\neq j$, considering the joint assignments of $i,j,k$ yields,
\[
E_{ijk} = \E\left[ v_i z_i v_j z_j v_k^{(-j)} z_k \wt{t}_k^{(-j)} \right] = \frac{-\wt\mu_k}{(n-1)(n-2)\sqrt{n_T n_C}}.
\]
We first consider the terms in $T_2$ with $i=k$. Note that in this case, since $j\neq k$, also $i\neq j$. Therefore by \eqref{eq:e-iji}, and some calculations, the sum of those terms are equal to 
\[
T_{2,A} = \frac{-2}{n^2(n-1)} \sum_{i,j \in [n]: i \neq j} \frac{R_i h_{\lambda ij} \wt\mu_i}{\bar{h}_j}.
\]
Similar calculations reveal that the terms corresponding to the case $i\neq k$, $i\neq j$ add up to
\[
T_{2,B} = \frac{2}{n^2 (n-1)(n-2)} \sum_{i,j \in [n]: i \neq j} \sum_{k\in [n]: k \neq i,j}  \frac{R_i h_{\lambda jk} \wt\mu_k}{\bar{h}_j}.
\]
Therefore
\[
T_2 = T_{2,A} + T_{2,B} = \frac{-2}{n^2(n-1)}  \sum_{i,j \in [n]: i \neq j} \left( \frac{R_i h_{\lambda ij} \widetilde{\mu}_i}{\overline{h}_{j}} - \frac{1}{n-2} \sum_{k \in [n]: k\neq i,j}  \frac{R_i h_{\lambda jk}  \widetilde{\mu}_k}{\overline{h}_{j}}
\right).
\]
Finally for $T_3$, following the same proof strategy as $T_1, T_2$, and calculating the expectation over the joint assignment of four distinct units $(i, j, k, l)$ for $16$ combinations, and after simplifying the algebra we have the following. We denote 
\begin{align*}
    F & = n^3 n_T n_C (n_T-1)(n_C-1), \\
    a_T &= \frac{n_Tn_C-2n_C+n_T^2-2n_T+1}{n_T-1} \\
    a_C &= \frac{n_Tn_C-2n_T+n_C^2-2n_C+1}{n_C-1}\\
    \bar{a}_T & = n_Cn_T - 3n_C + n_T^2 - 2n_T + 1, \\
    \bar{a}_C & = n_Cn_T - 3n_T + n_C^2 - 2n_C + 1.
\end{align*}
Then
\begin{align*}
    T_3 
    &= \frac{1}{F}\sum_{i \in [n]}\sum_{k \in [n]: k \neq i} \biggl[ h_{\lambda ik}^2\frac{(n_C-1)(\widetilde\bt^{(1)}_k)^2 + (n_T-1)(\widetilde\bt^{(0)}_k)^2 }{\bar{h}_i^2} \\ 
    &- \frac{1}{n-2}\sum_{l\in[n]: l\neq i, k}\frac{h_{\lambda ik}h_{\lambda il}}{\bar{h}_i^2}((n_C-1)\widetilde\bt^{(1)}_k\widetilde\bt^{(1)}_l + (n_T-1)\widetilde\bt^{(0)}_k\widetilde\bt^{(0)}_l) \biggr] \\
    &+ \frac{1}{F}\sum_{i,j \in [n]: i \neq j} 
    \sum_{k \in [n]: k \neq i, j} \frac{h_{\lambda ik}h_{\lambda jk}}{\bar{h}_i\bar{h}_j}\Bigl[a_T(\widetilde\bt^{(1)}_k)^2 - 2n\widetilde\bt^{(1)}_k\widetilde\bt^{(0)}_k + a_C(\widetilde\bt^{(0)}_k)^2 \Bigr] 
    \\
    &+ \frac{1}{n^3(n-1)}\sum_{i,j\in [n]: i \neq j} \biggl[\frac{h_{\lambda ij}^2}{\bar{h}_{i}\bar{h}_j}\Bigl[\frac{\widetilde\bt^{(1)}_i\widetilde\bt^{(1)}_j}{n_T(n_T-1)} + \frac{\widetilde\bt^{(0)}_i\widetilde\bt^{(0)}_j}{n_C(n_C-1)} + \frac{\widetilde\bt^{(1)}_i\widetilde\bt^{(0)}_j + \widetilde\bt^{(0)}_i\widetilde\bt^{(1)}_j}{n_Cn_T} \Bigr] \\
    &- \frac{1}{n-2}\sum_{k \in [n]: k \neq i, j}\frac{h_{\lambda ji}}{\bar{h}_{i}\bar{h}_j}\Bigl[\frac{(h_{\lambda ik} \widetilde\bt^{(1)}_i + h_{\lambda jk} \widetilde\bt^{(1)}_j)\widetilde\bt^{(1)}_k}{n_T(n_T-1)} + \frac{(h_{\lambda ik}\widetilde\bt^{(0)}_i + h_{\lambda jk}\widetilde\bt^{(0)}_j)\widetilde\bt^{(0)}_k}{n_C(n_C-1)}
    \\ & + \frac{(h_{\lambda ik}\widetilde\bt^{(1)}_i + h_{\lambda jk}\widetilde\bt^{(1)}_j)\widetilde\bt^{(0)}_k + (h_{\lambda ik}\widetilde\bt^{(0)}_i + h_{\lambda jk}\widetilde\bt^{(0)}_j)\widetilde\bt^{(1)}_k}{n_Cn_T}\Bigr] \\
    &- \frac{1/(n-2)}{n_Cn_T(n-3)}\!\!\sum_{\substack{k,l \in [n]: k \neq l, \\ k,l \notin \{i,j\}}} \!\!\!\!\!\frac{h_{\lambda ik}h_{\lambda jl}}{\bar{h}_i\bar{h}_j}\Bigl[\frac{\bar{a}_T\widetilde\bt^{(1)}_k\widetilde\bt^{(1)}_l}{n_T-1} + \frac{\bar{a}_C\widetilde\bt^{(0)}_k\widetilde\bt^{(0)}_l}{n_C-1} - (n+1)(\widetilde\bt^{(1)}_k\widetilde\bt^{(0)}_l + \widetilde\bt^{(0)}_k\widetilde\bt^{(1)}_l)\Bigr]\biggr].
\end{align*}

To further simplify the above expression, note that it be written as the following quadratic form
\[
\begin{bmatrix}
    \widetilde\bt^{(0)} \\ \widetilde\bt^{(1)}
\end{bmatrix}^\top \bQ
\begin{bmatrix}
    \widetilde\bt^{(0)} \\ \widetilde\bt^{(1)}
\end{bmatrix},
\]
where $\bQ$ is a block matrix of the following form
\[
\bQ = \begin{bmatrix}
\bQ^{00} & \bQ^{01} \\
\bQ^{10} & \bQ^{11}
\end{bmatrix},
\]
and
\begin{align*}
Q^{11}_{kk}
& = \frac{n_C-1}{F} \sum_{i \in [n]: i \neq k} \frac{h_{\lambda i k}^2}{\bar{h}_i^2}
+ \frac{a_T}{F} \sum_{\substack{i,j \in [n]: i \neq j \\ k \neq i,j}}
\frac{h_{\lambda i k} h_{\lambda j k}}{\bar{h}_i \bar{h}_j},
\\
Q^{00}_{kk}
& = \frac{n_T-1}{F} \sum_{i \in [n]: i \neq k} \frac{h_{\lambda i k}^2}{\bar{h}_i^2}
+ \frac{a_C}{F} \sum_{\substack{i,j \in [n]: i \neq j \\ k \neq i,j}}
\frac{h_{\lambda i k} h_{\lambda j k}}{\bar{h}_i \bar{h}_j},
\\
Q^{01}_{kk} = Q^{10}_{kk}
& = - \frac{n}{F}
\sum_{\substack{i,j \in [n]: i \neq j \\ k \neq i,j}}
\frac{h_{\lambda i k} h_{\lambda j k}}{\bar{h}_i \bar{h}_j},
\\
Q^{11}_{k\ell}
& = - \frac{n_C-1}{F(n-2)} \left(\sum_{i \in [n]: i \neq k, \ell}
\frac{h_{\lambda i k} h_{\lambda i \ell}}{\bar{h}_i^2} \right)
+ \frac{1}{n^3(n-1)} \biggl[
\frac{1}{n_T(n_T-1)} \frac{h_{\lambda k \ell}^2}{\bar{h}_k \bar{h}_\ell}
\\ &
- \frac{1}{n-2} \frac{1}{n_T(n_T-1)} 
\biggl(
\sum_{i \in [n]: i \neq k, \ell} \frac{h_{\lambda k i} h_{\lambda \ell k}}{\bar{h}_k \bar{h}_i}
+ \frac{h_{\lambda i k} h_{\lambda k \ell}}{\bar{h}_i \bar{h}_\ell}
\biggr)
\\ &
- \frac{1}{(n-2) n_C n_T (n-3)} \frac{\bar{a}_T}{n_T-1}
\sum_{\substack{i,j \in [n]: i \neq j \\ i,j \notin \{k,\ell\}}}
\frac{h_{\lambda i k} h_{\lambda j \ell}}{\bar{h}_i \bar{h}_j}
\biggr], \\
Q^{00}_{k\ell}
& = - \frac{n_T-1}{F(n-2)} \left(\sum_{i \in [n]: i \neq k, \ell}
\frac{h_{\lambda i k} h_{\lambda i \ell}}{\bar{h}_i^2} \right)
+ \frac{1}{n^3(n-1)} \biggl[
\frac{1}{n_C(n_C-1)} \frac{h_{\lambda k \ell}^2}{\bar{h}_k \bar{h}_\ell}
\\ &
- \frac{1}{n-2} \frac{1}{n_C(n_C-1)} 
\biggl(
\sum_{i \in [n]: i \neq k, \ell} \frac{h_{\lambda k i} h_{\lambda \ell k}}{\bar{h}_k \bar{h}_i}
+  \frac{h_{\lambda i k} h_{\lambda k \ell}}{\bar{h}_i \bar{h}_\ell}
\biggr)
\\ &
- \frac{1}{(n-2) n_C n_T (n-3)} \frac{\bar{a}_C}{n_C-1}
\sum_{\substack{i,j \in [n]: i \neq j \\ i,j \notin \{k,\ell\}}}
\frac{h_{\lambda i k} h_{\lambda j \ell}}{\bar{h}_i \bar{h}_j}
\biggr],
\\
Q^{01}_{k\ell} = Q^{10}_{\ell k}
& = \frac{1}{n^3(n-1)} \biggl[
\frac{1}{n_C n_T} \frac{h_{\lambda k \ell}^2}{\bar{h}_k \bar{h}_\ell}
- \frac{1}{n-2} \frac{1}{n_C n_T} 
\biggl(
\sum_{i \in [n]: i \neq k, \ell} \frac{h_{\lambda k i} h_{\lambda \ell k}}{\bar{h}_k \bar{h}_i}
+ \frac{h_{\lambda i k} h_{\lambda k \ell}}{\bar{h}_i \bar{h}_\ell}
\biggr)
\\ &
+ \frac{n+1}{(n-2) n_C n_T (n-3)}
\sum_{\substack{i,j \in [n]: i \neq j \\ i,j \notin \{k,\ell\}}}
\frac{h_{\lambda i k} h_{\lambda j \ell}}{\bar{h}_i \bar{h}_j}
\biggr].
\end{align*}

\end{proof}

\begin{proof}[Proof of \Cref{thm:asymptotic-normality-LDM}]
We define:
\begin{align*}
    \bbeta_\lambda & = \argmin_{\bb \in \mathbb{R}^k}\norm{\bX\bb - \tilde{\bmu}}{2}^2 + \lambda\norm{\bb}{2}^2.
\\
    \hat{\bbeta}_\lambda^{(i)} & = \argmin_{\bb \in \mathbb{R}^k}\norm{\bX\bb - \tilde{\by}^{(-i)}}{2}^2 + \lambda\norm{\bb}{2}^2.
\\
    \hat{\bbeta}_\lambda^{(-i)} &= \argmin_{\bb \in \mathbb{R}^k}\norm{\bX_{-i}\bb - \tilde{\by}_{-i}^{(-i)}}{2}^2 + \lambda\norm{\bb}{2}^2.
\end{align*}
Define the oracle difference-in-means estimator as
\[
\widehat{\tau}_{\textup{LDM}}^{*}
=\sum_{i=1}^n v_i z_i\bigl(y_i-\eta^{-1}\bx_i^\top\bbeta_\lambda\bigr).
\]
Note that LOORA-DM estimator is 
\[
\widehat{\tau}_{\textup{LDM}}
=\sum_{i=1}^n v_i z_i\bigl(y_i-\bx_i^\top\widehat\bbeta^{(-i)}_\lambda\bigr).
\]
We define
\begin{align*}
\frac{\sqrt{n}\bigl(\widehat\tau_{\textup{LDM}}-\tau\bigr)}{\sqrt{\Var(\sqrt{n}(\widehat\tau_{\textup{LDM}}^{*}-\tau))}}
=
A_n+B_n,
~
A_n=\frac{\sqrt{n}\bigl(\widehat\tau_{\textup{LDM}}^{*}-\tau\bigr)}{\sqrt{\Var(\sqrt{n}(\widehat\tau_{\textup{LDM}}^{*}-\tau))}},
~ \\
B_n=\frac{\sqrt{n}\bigl(\widehat\tau_{\textup{LDM}}-\widehat\tau_{\textup{LDM}}^{*}\bigr)}{\sqrt{\Var(\sqrt{n}(\widehat\tau_{\textup{LDM}}^{*}-\tau))}}.
\end{align*}
Note that
\[
\sqrt{\Var(\sqrt{n}(\widehat\tau_{\textup{LDM}}^{*}-\tau))} = \frac{\norm{(\bX - \overline{\bX}) \bbeta_{\lambda}
\allowbreak - \allowbreak
(\widetilde{\bmu} - \overline{\widetilde{\bmu}})}{2}}{\sqrt{n-1}}.
\]

\paragraph{Step 1: Negligibility of the plug-in error ($B_n=o_p(1)$).}
We have
\[
\widehat\tau_{\textup{LDM}}-\widehat\tau_{\textup{LDM}}^{*}
=
\sum_{i=1}^n v_i z_i\,\bx_i^\top(\eta^{-1}\bbeta_{\lambda}-\widehat\bbeta_{\lambda}^{(-i)}).
\]
Add and subtract $\widehat\bbeta_\lambda$:
\[
\abs{\widehat\tau_{\textup{LDM}}-\widehat\tau_{\textup{LDM}}^{*}}
\le
\underbrace{\abs{\sum_{i=1}^n v_i z_i\,\bx_i^\top(\eta^{-1}\bbeta_{\lambda}-\widehat\bbeta_{\lambda}^{(i)})}}_{T_{1}}
+
\underbrace{\abs{\sum_{i=1}^n v_i z_i\,\bx_i^\top(\widehat\bbeta_{\lambda}^{(i)}-\widehat\bbeta_{\lambda}^{(-i)})}}_{T_{2}}.
\]
Note that by \eqref{eq:simple_value},
\[
\eta ~ \widetilde{\by}^{(-i)} = \tilde{\bmu} + \eta \frac{\bv^{(-i)}  \bz  \widetilde{\bt}^{(-i)}}{n}.
\]
Therefore
\[
\eta^{-1}\bbeta_{\lambda}-\widehat\bbeta_{\lambda}^{(i)} = - (\bX^\top \bX + \lambda \bI)^{-1} \bX^{\top} \frac{\bv^{(-i)}  \bz  \widetilde{\bt}^{(-i)}}{n}.
\]
Let $\widetilde{\bbeta}$ be a vector such that
\[
\eta^{-1}\bbeta_{\lambda}-\widetilde\bbeta_{\lambda} = - (\bX^\top \bX + \lambda \bI)^{-1} \bX^{\top} \frac{\bv  \bz  \widetilde{\bt}}{n},
\]
where 
\[
\widetilde{\bt} = n_C(n_C - 1) \by^{(1)} - n_T(n_T - 1) \by^{(0)}.
\]
Then one can easily show that
\[
\norm{\widehat{\bbeta}_{\lambda}^{(i)} - \widetilde{\bbeta}_{\lambda}}{2} = O(n^{-1}).
\]
Now by triangle inequality and Cauchy–Schwarz inequality,
\begin{align*}
T_{1} & \leq \abs{\sum_{i=1}^n v_i z_i\,\bx_i^\top(\eta^{-1}\bbeta_{\lambda}-\widetilde\bbeta_{\lambda})} + \abs{\sum_{i=1}^n v_i z_i\,\bx_i^\top(\widetilde\bbeta_{\lambda}-\widehat\bbeta_{\lambda}^{(i)})}
\\ & = 
\abs{\sum_{i=1}^n v_i z_i\,\bx_i^\top(\eta^{-1}\bbeta_{\lambda}-\widetilde\bbeta_{\lambda})} + O(n^{-1})
\\ & =
\abs{(\bX^\top (\bv \bz))^\top(\eta^{-1}\bbeta_{\lambda}-\widetilde\bbeta_{\lambda})} + O(n^{-1})
\\ & \leq
\norm{\bX^\top (\bv \bz)}{2} \norm{\eta^{-1}\bbeta_{\lambda}-\widetilde\bbeta_{\lambda}}{2} + O(n^{-1})
\\ & \leq
\norm{\bX^\top (\bv \bz)}{2} \norm{(\bX^\top \bX + \lambda \bI)^{-1}}{2} \norm{ \bX^{\top} \frac{\bv  \bz  \widetilde{\bt}}{n}}{2} + O(n^{-1}) = O_p(n^{-1}),
\end{align*}
where the first equality holds because $v_i=O(n^{-1})$ for all $i \in [n]$ and the last equality holds due to \Cref{lemma:concentrate_vec_sum_dm} and noting that
\[
v_i z_i=\frac{d_i}{n_T}-\frac{1-d_i}{n_C}
=\Bigl(\frac{1}{n_T}+\frac{1}{n_C}\Bigr)(d_i-p_T).
\]
Now note that by \Cref{cor:sum-sqr-err-with-removal} and because of \Cref{ass:boundedness,ass:positivity,ass:conv-2nd-moment} 
\begin{align}
\label{eq:beta_i_close_to_beta_-i}
\norm{\widehat\bbeta_{\lambda}^{(i)} -\widehat\bbeta_{\lambda}^{(-i)}}{2} = O(n^{-1}).
\end{align}
Therefore by triangle inequality and Cauchy–Schwarz inequality,
\[
T_{2} \leq  \sum_{i=1}^n \norm{v_i z_i\,\bx_i^\top}{2} \norm{\widehat\bbeta_{\lambda}^{(i)} -\widehat\bbeta_{\lambda}^{(-i)}}{2} = O(n^{-1}),
\]
where the equality holds because $v_i=O(n^{-1})$ for all $i \in [n]$.
Therefore $\abs{\widehat\tau_{\textup{LDM}}-\widehat\tau_{\textup{LDM}}^{*}} = O_p(n^{-1})$ and due to the assumption that $\norm{\tilde\bmu-\bar{\tilde{\bmu}}-(\bX-\bar{\bX})\bbeta}{2} \to \infty$, $B_n = o_p(1)$.

\paragraph{Step 2: Oracle CLT via Hoeffding's combinatorial CLT.} 
Let 
\[
u_i^{(1)}=y_i^{(1)}- \eta^{-1}\bx_i^\top\bbeta_\lambda, ~~ \text{and} ~~ u_i^{(0)}=y_i^{(0)}- \eta^{-1} \bx_i^\top\bbeta_\lambda
\]
Therefore we have
\begin{align*}
\widehat\tau_{\textup{LDM}}^{*}
& =
\frac{1}{n_T}\sum_{i \in [n]: d_i = 1} u_i^{(1)}-\frac{1}{n_C}\sum_{i \in [n]: d_i = 0} u_i^{(0)}
\\ & =
-\frac{1}{n_C}\sum_{i=1}^n u_i^{(0)}
+
\sum_{i \in [n]: d_i = 1}\left(\frac{1}{n_T}u_i^{(1)}+\frac{1}{n_C}u_i^{(0)}\right)
\\ & =
-\frac{1}{n_C}\sum_{i=1}^n u_i^{(0)}
+
\frac{1}{\sqrt{n_T n_C}}\sum_{i \in [n]: d_i = 1}\left(\widetilde{\mu}_i -\bx_i^\top \bbeta_{\lambda}\right).
\end{align*}
The first term on the right-hand-side of the above is determinstic and the only sourse of randomness is the second term.

Suppose for each natural number $n$, we have $2n$ numbers $a_{ni},c_{ni}$. For a permutation $\pi(1),\ldots,\pi(n)$ of $[n]$, define the following random variable.
\[
S_n = \sum_{i=1}^n a_{ni} c_{n \pi(i)}.
\]
\cite{hoeffdingCLT} shows that
\[
\frac{S_n-\E(S_n)}{\sqrt{\Var(S_n)}}\;\dto\; \mathcal N(0,1)
\]
if 
\begin{align*}
\lim_{n \to \infty} n \cdot \frac{\max_{i\in [n]} (a_{ni} - \overline{a}_n)^2}{ \sum_{i=1}^n (a_{ni} - \overline{a}_n)^2} \cdot \frac{\max_{i\in[n]} (c_{ni} - \overline{c}_n)^2}{ \sum_{i=1}^n (c_{ni} - \overline{c}_n)^2} = 0,
\end{align*}
where
\[
\overline{a}_n = \frac{1}{n} \sum_{i=1}^n a_{ni}, ~~ \text{and} ~~ \overline{c}_n = \frac{1}{n} \sum_{i=1}^n c_{ni}.
\]
We set $a_{ni} = (\widetilde{\mu}_i -\bx_i^\top \bbeta_{\lambda})$ and $c_{ni}=\frac{1}{\sqrt{n_T n_C}}\mathbbm{1}\{i\le n_T\}$, where $\mathbbm{1}\{i\le n_T\}$ is an indicator function for $i\le n_T$. \Cref{ass:boundedness,ass:positivity,ass:conv-2nd-moment} imply that each $a_{ni}$ is bounded. Therefore each $a_{ni} - \overline{a}_n$ is bounded. Let $C$ be a constant such that $|a_{ni} - \overline{a}_n| \leq C$ for all $n$ and $i$. Then
\begin{align*}
\frac{\max_{i\in [n]} (a_{ni} - \overline{a}_n)^2}{ \sum_{i=1}^n (a_{ni} - \overline{a}_n)^2}
& \leq 
\frac{C^2}{\norm{(\bX - \overline{\bX}) \bbeta_{\lambda}
\allowbreak - \allowbreak
(\widetilde{\bmu} - \overline{\widetilde{\bmu}})}{2}}.
\end{align*}
Now note that $\overline{c}_n = \frac{1}{\sqrt{n_T n_C}} \frac{n_T}{n}$ and therefore each $c_{ni} - \overline{c}_n$ is either equal to $\frac{1}{\sqrt{n_T n_C}} \frac{n_C}{n}$ or $-\frac{1}{\sqrt{n_T n_C}} \frac{n_T}{n}$. Therefore we have
\begin{align*}
\frac{\max_{i\in[n]} (c_{ni} - \overline{c}_n)^2}{ \sum_{i=1}^n (c_{ni} - \overline{c}_n)^2} & = \frac{\max\{\left(\frac{n_C}{n}\right)^2,\left(\frac{n_T}{n}\right)^2\}}{n_T \left(\frac{n_C}{n}\right)^2 + n_C \left(\frac{n_T}{n}\right)^2} = \frac{\max\{n_C^2,n_T^2\}}{n_C n_T n} = \frac{1}{n} \max\{\frac{n_C}{n_T}, \frac{n_T}{n_C}\} 
\\ & \leq \frac{1}{n} \frac{1-m}{m}.
\end{align*}
where the inequality follows from \Cref{ass:positivity}. Thus
\begin{align*}
n \cdot \frac{\max_{i\in [n]} (a_{ni} - \overline{a}_n)^2}{ \sum_{i=1}^n (a_{ni} - \overline{a}_n)^2} \cdot \frac{\max_{i\in[n]} (c_{ni} - \overline{c}_n)^2}{ \sum_{i=1}^n (c_{ni} - \overline{c}_n)^2} \leq \frac{C^2}{\norm{(\bX - \overline{\bX}) \bbeta_{\lambda}
\allowbreak - \allowbreak
(\widetilde{\bmu} - \overline{\widetilde{\bmu}})}{2}} \frac{1-m}{m}.
\end{align*}
Therefore because of the assumption that $\norm{(\bX - \overline{\bX}) \bbeta_{\lambda}
\allowbreak - \allowbreak
(\widetilde{\bmu} - \overline{\widetilde{\bmu}})}{2} \to \infty$, the limit of the above is zero. Therefore by \cite{hoeffdingCLT},
\[
\frac{\sqrt{n}\bigl(\widehat\tau_{\textup{LDM}}^{*}-\tau\bigr)}{\sqrt{\Var(\sqrt{n}(\widehat\tau_{\textup{LDM}}^{*}-\tau))}}
\;\dto\; \mathcal N(0,1).
\]
Moreover note that
\[
\sqrt{\Var(\sqrt{n}(\widehat\tau_{\textup{LDM}}^{*}-\tau))} = \frac{\norm{(\bX - \overline{\bX}) \bbeta_{\lambda}
\allowbreak - \allowbreak
(\widetilde{\bmu} - \overline{\widetilde{\bmu}})}{2}}{\sqrt{n-1}}.
\]

\paragraph{Step 3: Conclusion.}
Since $A_n\dto \mathcal N(0,1)$ and $B_n=o_p(1)$, Slutsky's theorem gives
\[
\frac{\sqrt{n}\bigl(\widehat\tau_{\textup{LDM}}-\tau\bigr)}{\norm{(\bX - \overline{\bX}) \bbeta_{\lambda}
\allowbreak - \allowbreak
(\widetilde{\bmu} - \overline{\widetilde{\bmu}})}{2} / \sqrt{n}}
\;\dto\; \mathcal N(0,1),
\]
where we also use the fact that $\lim_{n\to \infty} \frac{n}{n-1} = 1$.
\end{proof}

\begin{proof}[Proof of \Cref{thm:loora-dm-est-var}]
We define
\begin{align*}
    \bbeta_\lambda & = \argmin_{\bb \in \mathbb{R}^k}\norm{\bX\bb - \tilde{\bmu}}{2}^2 + \lambda\norm{\bb}{2}^2, \\
    \hat{\bbeta}_\lambda^{(i)} & = \argmin_{\bb \in \mathbb{R}^k}\norm{\bX\bb - \tilde{\by}^{(-i)}}{2}^2 + \lambda\norm{\bb}{2}^2, \\
    \hat{\bbeta}_\lambda^{(-i)} & = \argmin_{\bb \in \mathbb{R}^k}\norm{\bX_{-i}\bb - \tilde{\by}_{-i}^{(-i)}}{2}^2 + \lambda\norm{\bb}{2}^2.
\end{align*}
and
\[
\widetilde{V}_{\textup{LDM}} = \frac{n-1}{n_C (n_C - 1)} \sum_{i \in [n]: d_i=0} (\widetilde{s}_i^{(0)})^2 + \frac{n-1}{n_T(n_T-1)} \sum_{i \in [n]: d_i=1} (\widetilde{s}_i^{(1)})^2,
\]
where
\begin{align*}
\widetilde{s}_i^{(0)} & = y_i^{(0)} - \eta^{-1}\bx_i^\top\bbeta_{\lambda} - \frac{1}{n_C} \sum_{j \in [n]:d_j=0} y_j^{(0)} - \eta^{-1}\bx_j^\top\bbeta_{\lambda}, 
\\
\widetilde{s}_i^{(1)} & = y_i^{(1)} - \eta^{-1}\bx_i^\top\bbeta_{\lambda} - \frac{1}{n_T} \sum_{j \in [n]:d_j=1} y_j^{(1)} - \eta^{-1}\bx_j^\top\bbeta_{\lambda}.
\end{align*}
We show
\[
\hat V_{\textup{LDM}} - 
\left( \frac{1}{n}\norm{(\bX-\overline{\bX}) \bbeta_{\lambda}
-
(\widetilde{\bmu} - \overline{\widetilde{\bmu}})}{2}^2 + \frac{1}{n} \sum_{i=1}^n(\tau_i-\tau)^2 \right) \pto 0
\]
by separating the left-hand-side to two terms: $\hat V_{\textup{LDM}} - \widetilde{V}_{\textup{LDM}}$ and
\[
\widetilde{V}_{\textup{LDM}} -\left( \frac{1}{n}\norm{(\bX-\overline{\bX}) \bbeta_{\lambda}
-
(\widetilde{\bmu} - \overline{\widetilde{\bmu}})}{2}^2 + \frac{1}{n} \sum_{i=1}^n(\tau_i-\tau)^2 \right). 
\]

\paragraph{Step 1: Concentration of the oracle variance $\tilde V_{\textup{LDM}}$.}
Let
\begin{align*}
s_i^{(0)} & =  y_i^{(0)} - \eta^{-1}\bx_i^\top\bbeta_{\lambda} - \frac{1}{n} \sum_{j \in [n]} y_j^{(0)} - \eta^{-1}\bx_j^\top\bbeta_{\lambda},
\\
s_i^{(1)} & =  y_i^{(1)} - \eta^{-1}\bx_i^\top\bbeta_{\lambda} - \frac{1}{n} \sum_{j \in [n]} y_j^{(1)} - \eta^{-1}\bx_j^\top\bbeta_{\lambda}.
\end{align*}
Using standard formulas regarding the connection between sample variance and population variance, one can easily confirm that
\begin{align*}
\E\left[ \frac{1}{n_C-1} \sum_{i \in [n]: d_i=0} (\widetilde{s}_i^{(0)})^2 \right] & = \frac{1}{n-1} \sum_{i \in [n]} ( s_i^{(0)} )^2, \\
\E\left[ \frac{1}{n_T-1} \sum_{i \in [n]: d_i=1} (\widetilde{s}_i^{(1)})^2 \right] & = \frac{1}{n-1} \sum_{i \in [n]} ( s_i^{(1)} )^2.
\end{align*}
Therefore
\[
\E[\widetilde{V}_{\textup{LDM}}] = \frac{1}{n_C} \sum_{i \in [n]} ( s_i^{(0)} )^2 + \frac{1}{n_T} \sum_{i \in [n]} ( s_i^{(1)} )^2.
\]
Since $\frac{1}{n_C} \sum_{i \in [n]: d_i=0} (\widetilde{s}_i^{(0)})^2$ is a sample mean of $(\widetilde{s}_i^{(0)})^2$'s, its variance is $O(n_C^{-1})$. Similarly the variance of $\frac{1}{n_T} \sum_{i \in [n]: d_i=1} (\widetilde{s}_i^{(1)})^2$ is $O(n_T^{-1})$. Therefore
\begin{align*}
\Var[\widetilde{V}_{\textup{LDM}}] & \leq  \frac{2(n-1)^2}{(n_C - 1)^2} \Var \left[ \frac{1}{n_C} \sum_{i \in [n]: d_i=0} (\widetilde{s}_i^{(0)})^2 \right] + \frac{2(n-1)^2}{(n_T-1)^2} \Var \left[ \frac{1}{n_T} \sum_{i \in [n]: d_i=1} (\widetilde{s}_i^{(1)})^2 \right]
\\ & = O(\frac{n^2}{n_C^3} + \frac{n^2}{n_T^3}) = o(1).
\end{align*}
Therefore since $\Var[\widetilde{V}_{\textup{LDM}}] \to 0$,
\[
\widetilde{V}_{\textup{LDM}} - \left(\frac{1}{n_C} \sum_{i \in [n]} ( s_i^{(0)} )^2 + \frac{1}{n_T} \sum_{i \in [n]} ( s_i^{(1)} )^2 \right) \pto 0.
\]
We now show
\begin{align}
\frac{1}{n_C} \sum_{i \in [n]} ( s_i^{(0)} )^2 + \frac{1}{n_T} \sum_{i \in [n]} ( s_i^{(1)} )^2 = \frac{1}{n} \norm{(\bX-\overline{\bX}) \bbeta_{\lambda}
-
(\widetilde{\bmu} - \overline{\widetilde{\bmu}})}{2}^2 + \frac{1}{n} \sum_{i=1}^n(\tau_i-\tau)^2.
\end{align}
First, by definition of $s_i^{(a)}$,
\[
y_i^{(a)}-\frac{1}{n}\sum_{j=1}^n y_j^{(a)}
=
s_i^{(a)} + \eta^{-1} \left(\bx_i-\frac{1}{n}\sum_{j=1}^n \bx_j\right)^\top\bbeta_{\lambda},
\qquad a\in\{0,1\}.
\]
Using $\widetilde{\bmu}=\sqrt{\frac{n_C}{n_T}}\by^{(1)}+\sqrt{\frac{n_T}{n_C}}\by^{(0)}$ and
$\eta=\sqrt{\frac{n_C}{n_T}}+\sqrt{\frac{n_T}{n_C}}$, we obtain for each $i$ that
\begin{align*}
\widetilde{\mu}_i-\overline{\widetilde{\mu}}
&=
\sqrt{\frac{n_C}{n_T}}
\Bigl(y_i^{(1)}-\frac{1}{n}\sum_{j=1}^n y_j^{(1)}\Bigr)
+
\sqrt{\frac{n_T}{n_C}}
\Bigl(y_i^{(0)}-\frac{1}{n}\sum_{j=1}^n y_j^{(0)}\Bigr)
\\
&=
\sqrt{\frac{n_C}{n_T}}\, s_i^{(1)}
+\sqrt{\frac{n_T}{n_C}}\, s_i^{(0)}
+\left(\sqrt{\frac{n_C}{n_T}}+\sqrt{\frac{n_T}{n_C}}\right) \eta^{-1}
\left(\bx_i-\overline{\bx}\right)^\top\bbeta_{\lambda}
\\
&=
\sqrt{\frac{n_C}{n_T}}\, s_i^{(1)}
+\sqrt{\frac{n_T}{n_C}}\, s_i^{(0)}
+(\bx_i-\overline{\bx})^\top\bbeta_{\lambda}.
\end{align*}
Rearranging yields
\[
(\bx_i-\overline{\bx})^\top\bbeta_{\lambda}-(\widetilde{\mu}_i-\overline{\widetilde{\mu}})
=
-\left(
\sqrt{\frac{n_C}{n_T}}\, s_i^{(1)}
+\sqrt{\frac{n_T}{n_C}}\, s_i^{(0)}
\right),
\]
and hence
\begin{equation}
\label{eq:norm-term}
\norm{(\bX-\overline{\bX}) \bbeta_{\lambda}
-
(\widetilde{\bmu} - \overline{\widetilde{\bmu}})}{2}^2
=
\sum_{i=1}^n
\left(
\sqrt{\frac{n_C}{n_T}}\, s_i^{(1)}
+\sqrt{\frac{n_T}{n_C}}\, s_i^{(0)}
\right)^2.
\end{equation}
Next, since $\tau_i=y_i^{(1)}-y_i^{(0)}$ and $\tau=\frac{1}{n}\sum_{j=1}^n \tau_j$, we have
\[
\tau_i-\tau
=
\Bigl(y_i^{(1)}-\frac{1}{n}\sum_{j=1}^n y_j^{(1)}\Bigr)
-
\Bigl(y_i^{(0)}-\frac{1}{n}\sum_{j=1}^n y_j^{(0)}\Bigr)
=
s_i^{(1)}-s_i^{(0)}.
\]
Therefore
\begin{equation}
\label{eq:tau-term}
\sum_{i=1}^n(\tau_i-\tau)^2=\sum_{i=1}^n (s_i^{(1)}-s_i^{(0)})^2.
\end{equation}
By combining \eqref{eq:norm-term} and \eqref{eq:tau-term}, and expanding:
\begin{align*}
&\frac{1}{n}\norm{(\bX-\overline{\bX}) \bbeta_{\lambda}
-
(\widetilde{\bmu} - \overline{\widetilde{\bmu}})}{2}^2 + \frac{1}{n} \sum_{i=1}^n(\tau_i-\tau)^2 \\
&\qquad= 
\frac{1}{n} \sum_{i=1}^n\left(
\sqrt{\frac{n_C}{n_T}}\, s_i^{(1)}
+\sqrt{\frac{n_T}{n_C}}\, s_i^{(0)}
\right)^2
+
(s_i^{(1)}-s_i^{(0)})^2
\\
&\qquad=
\frac{1}{n} \sum_{i=1}^n
\left(\frac{n_C}{n_T}+1\right)(s_i^{(1)})^2
+
\left(\frac{n_T}{n_C}+1\right)(s_i^{(0)})^2
\\
&\qquad=
\sum_{i=1}^n
\frac{1}{n_T}(s_i^{(1)})^2+\frac{1}{n_C}(s_i^{(0)})^2,
\end{align*}
where we used $n=n_C+n_T$. 

\paragraph{Step 2: Showing $\hat V_{\textup{LDM}}-\tilde V_{\textup{LDM}}=o_p(1)$.} We have
\begin{align*}
\hat V_{\textup{LDM}}-\tilde V_{\textup{LDM}} & = \frac{n-1}{n_C (n_C - 1)} \sum_{i \in [n]: d_i=0} \!\! \left[(\widehat{s}_i^{(0)})^2 - (\widetilde{s}_i^{(0)})^2 \right]
\\ & + \frac{n-1}{n_T(n_T-1)} \sum_{i \in [n]: d_i=1} \!\! \left[(\widehat{s}_i^{(1)})^2 - (\widetilde{s}_i^{(1)})^2 \right].
\end{align*}
It is easy to verify that by \Cref{ass:boundedness,ass:positivity,ass:conv-2nd-moment}, $\widehat{s}_i^{(0)}, \widehat{s}_i^{(1)}, \widetilde{s}_i^{(0)}, \widetilde{s}_i^{(1)}$ are all bounded for all $i$. Therefore to prove $\hat V_{\textup{LDM}}-\tilde V_{\textup{LDM}}$ converges to zero in probability, we only need to show that 
\begin{align*}
\widehat{s}_i^{(0)} - \widetilde{s}_i^{(0)} \pto 0, ~~ \forall i \in [n]: d_i=0, \\
\widehat{s}_i^{(1)} - \widetilde{s}_i^{(1)} \pto 0, ~~ \forall i \in [n]: d_i=1.
\end{align*}
We only prove this for the former and the latter follows similarly. Moreover,  by \eqref{eq:beta_i_close_to_beta_-i}, 
\[
\norm{\widehat\bbeta_{\lambda}^{(i)} -\widehat\bbeta_{\lambda}^{(-i)}}{2} = O(n^{-1}).\]
Hence since
\begin{align*}
\widetilde{s}_i^{(0)} & = y_i^{(0)} - \eta^{-1}\bx_i^\top\bbeta_{\lambda} - \frac{1}{n_C} \sum_{j \in [n]:d_j=0} y_j^{(0)} - \eta^{-1}\bx_j^\top\bbeta_{\lambda}, ~~ \text{and}
\\ \widehat{s}_i^{(0)} & = y_i^{(0)} - \bx_i^\top \widehat{\bbeta}_{\lambda}^{(-i)} - \frac{1}{n_C} \sum_{j \in [n]:d_j=0} y_j^{(0)} - \bx_j^\top \widehat{\bbeta}_{\lambda}^{(-j)},
\end{align*}
we only need to prove that $\bx_i^\top (\eta^{-1}\bbeta_{\lambda} - \widehat{\bbeta}_{\lambda}^{(i)}) \pto 0$ for all $i\in [n]$ with $d_i=0$.
By \eqref{eq:simple_value},
\[
\eta ~ \widetilde{\by}^{(-i)} = \tilde{\bmu} + \eta \frac{\bv^{(-i)}  \bz  \widetilde{\bt}^{(-i)}}{n}.
\]
Therefore 
\begin{align*}
\norm{\widehat{\bbeta}_{\lambda}^{(i)} - \eta^{-1}\bbeta_{\lambda}}{2} 
& = 
\norm{(\bX^\top \bX + \lambda \bI)^{-1} \bX^\top \frac{\bv^{(-i)}  \bz  \widetilde{\bt}^{(-i)}}{n}}{2}
\\ & \leq
\norm{(\bX^\top \bX + \lambda \bI)^{-1} \bX^\top \frac{\bv  \bz  \widetilde{\bt}^{(-i)}}{n}}{2}
\\ & +
\norm{(\bX^\top \bX + \lambda \bI)^{-1} \bX^\top \frac{(\bv^{(-i)} - \bv ) \bz  \widetilde{\bt}^{(-i)}}{n}}{2}.
\end{align*}
Now we have
\begin{align*}
\norm{(\bX^\top \bX + \lambda \bI)^{-1} \bX^\top \frac{\bv  \bz  \widetilde{\bt}^{(-i)}}{n}}{2} & \leq \norm{(\bX^\top \bX + \lambda \bI)^{-1}}{2} \norm{\bX^\top \frac{\bv  \bz  \widetilde{\bt}^{(-i)}}{n}}{2}
\\ & =
O_p(n^{-1/2}),
\end{align*}
where the equality follows from \Cref{ass:boundedness,ass:positivity,ass:conv-2nd-moment} and \Cref{lemma:concentrate_vec_sum_dm}. Moreover
\begin{align*}
\norm{(\bX^\top \bX + \lambda \bI)^{-1} \bX^\top \frac{(\bv^{(-i)} - \bv ) \bz  \widetilde{\bt}^{(-i)}}{n}}{2} & \leq 
\norm{(\bX^\top \bX + \lambda \bI)^{-1}}{2} \norm{\bX^\top \frac{(\bv^{(-i)} - \bv ) \bz  \widetilde{\bt}^{(-i)}}{n}}{2}
\\ & =
O(n^{-1}),
\end{align*}
where the equality follows from \Cref{ass:boundedness,ass:positivity,ass:conv-2nd-moment} and the fact that $\abs{v^{(-i)}_j - v_j} = O(n^{-2})$ for all $i$ and $j$. Therefore since $\bx_i$'s are bounded, $\bx_i^\top (\eta^{-1}\bbeta_{\lambda} - \widehat{\bbeta}_{\lambda}^{(-i)}) \pto 0$ for all $i\in [n]$ with $d_i=0$.

\paragraph{Conclusion.}
Combining Step 1 and Step 2,
\[
\hat V_{\textup{LDM}} - 
\left( \frac{1}{n}\norm{(\bX-\overline{\bX}) \bbeta_{\lambda}
-
(\widetilde{\bmu} - \overline{\widetilde{\bmu}})}{2}^2 + \frac{1}{n} \sum_{i=1}^n(\tau_i-\tau)^2 \right) \pto 0
\]
which completes the proof since by assumption for all $n$,
\[
\frac{1}{n}\norm{(\bX-\overline{\bX}) \bbeta_{\lambda}
-
(\widetilde{\bmu} - \overline{\widetilde{\bmu}})}{2}^2 + \frac{1}{n} \sum_{i=1}^n(\tau_i-\tau)^2
\]
is bounded away from zero.
\end{proof}

\subsection{Proof of equivalence of LOORA-DM to leave-two-estimaotrs}
\label{sec:equivalence-loora-dm-leave-two-out}
For a treated unit $e$ (i.e., $d_e=1$), the adjustment in the LOORA-DM estimator is
\begin{align}
\label{eq:dk1-adj-value}
- \frac{1}{n_T}
\bx_e^\top
(\bX_{-e}^\top \bX_{-e} + \lambda \bI)^{-1}
\bX_{-e}^\top
\widetilde{\by}^{(T)}_{-e},
\end{align}
where
\[
\widetilde{y}^{(T)}_{\ell}
=
\begin{cases}
\frac{n_C (n-1)}{(n_T - 1)n} y^{(1)}_{\ell}, & d_{\ell} = 1, \\[0.6ex]
\frac{n_T (n-1)}{n_C n} y^{(0)}_{\ell}, & d_{\ell} = 0.
\end{cases}
\]
For a control unit $e$ ($d_e=0$), the corresponding adjustment is
\begin{align}
\label{eq:dk0-adj-value}
+ \frac{1}{n_C}
\bx_e^\top
(\bX_{-e}^\top \bX_{-e} + \lambda \bI)^{-1}
\bX_{-e}^\top
\widetilde{\by}^{(C)}_{-e},
\end{align}
with
\[
\widetilde{y}^{(C)}_{\ell}
=
\begin{cases}
\frac{n_C (n-1)}{n_T n} y^{(1)}_{\ell}, & d_{\ell} = 1, \\[0.6ex]
\frac{n_T (n-1)}{(n_C - 1)n} y^{(0)}_{\ell}, & d_{\ell} = 0.
\end{cases}
\]

As discussed in \Cref{sec:connection-spiess}
\begin{align*}
- \frac{1}{n_T n_C} \sum_{i<j} (d_i - d_j) \phi_{ij}(\bz_{-ij}) 
= & 
-\frac{1}{n_C n_T} \sum_{d_i=1} \bx_i^\top (\bX_{-i}^\top \bX_{-i} + \lambda \bI)^{-1} \sum_{d_j=0} \bX_{-ij}^\top \widetilde{\by}_{-ij}
\\ &
+\frac{1}{n_C n_T} \sum_{d_i=0} \bx_i^\top (\bX_{-i}^\top \bX_{-i} + \lambda \bI)^{-1} \sum_{d_j=1} \bX_{-ij}^\top \widetilde{\by}_{-ij}\,.
\end{align*}
We have,
\begin{align*}
\sum_{d_j=0} \bX_{-ij}^\top \widetilde{\by}_{-ij} 
 & = 
\sum_{d_j=0} \sum_{\ell \neq i,j} \bx_\ell \widetilde{y}_{\ell}
\\ & = 
\sum_{d_j=0} \left(\sum_{d_\ell=1, \ell \neq i} \bx_\ell \widetilde{y}_{\ell} + \sum_{d_\ell=0, \ell\neq j} \bx_\ell \widetilde{y}_{\ell} \right)
\\ & =
n_C \sum_{d_\ell=1, \ell \neq i} \bx_\ell \widetilde{y}_{\ell} + (n_C-1) \sum_{d_\ell=0} \bx_\ell \widetilde{y}_{\ell}\,.
\end{align*}
Now note that if $d_{\ell}=1$, then $\widetilde{y}_{\ell} = \widetilde{y}^{(T)}_{\ell}$ and if $d_{\ell}=0$, then $\widetilde{y}_{\ell} = \frac{n_C}{n_C-1}\widetilde{y}^{(T)}_{\ell}$.
Therefore for $i$ with $d_i=1$, we have
\begin{align*}
&
-\frac{1}{n_C n_T} \bx_i^\top (\bX_{-i}^\top \bX_{-i} + \lambda \bI)^{-1} \sum_{d_j=0} \bX_{-ij}^\top \widetilde{\by}_{-ij}
\\ & = 
-\frac{1}{n_C n_T} \bx_i^\top (\bX_{-i}^\top \bX_{-i} + \lambda \bI)^{-1} \left( n_C \sum_{d_\ell=1, \ell \neq i} \bx_\ell \widetilde{y}_{\ell} + (n_C-1) \sum_{d_\ell=0} \bx_\ell \widetilde{y}_{\ell} \right)
\\ & = 
-\frac{1}{n_C n_T} \bx_i^\top (\bX_{-i}^\top \bX_{-i} + \lambda \bI)^{-1} \left( n_C \sum_{d_\ell=1, \ell \neq i} \bx_\ell \widetilde{y}^{(T)}_{\ell} + n_C \sum_{d_\ell=0} \bx_\ell \widetilde{y}^{(T)}_{\ell} \right)
\\ & =
- \frac{1}{n_T} \bx_i^\top (\bX_{-i}^\top \bX_{-i} + \lambda \bI)^{-1} \bX_{-i}^\top \widetilde{\by}^{(T)}_{-i}\,,
\end{align*}
which is equal to \eqref{eq:dk1-adj-value}.
With a similar argument, one can show that for $i$ with $d_i=0$, 
\begin{align*}
&
+\frac{1}{n_C n_T} \bx_i^\top (\bX_{-i}^\top \bX_{-i} + \lambda \bI)^{-1} \sum_{d_j=1} \bX_{-ij}^\top \widetilde{\by}_{-ij}
\\ = &
+ \frac{1}{n_C} \bx_i^\top (\bX_{-i}^\top \bX_{-i} + \lambda \bI)^{-1} \bX_{-i}^\top \widetilde{\by}^{(C)}_{-i}\,,
\end{align*}
which is equal to \eqref{eq:dk0-adj-value}. Therefore our estimator can be written as a leave-two-out estimator.

\section{Auxilary Lemmas}
\begin{lemma}[Euclidean norm of a zero-mean sum is $O_p(\sqrt n)$]\label{lemma:concentrate_vec_sum}
Let $\bx_1,\dots,\bx_n \in \mathbb{R}^k$ be random vectors such that
\[
\mathbb{E}[\bx_i]= \boldsymbol{0},
\qquad
\mathbb{E}\left[\norm{\bx_i}{2}^2\right] \le C \ \ \text{for all } i
\]
for some constant $C<\infty$ independent of $n$.
Assume in addition that the cross second moments vanish:
\[
\mathbb{E}[\bx_i^\top \bx_j]=0 \qquad \text{for all } i\neq j,
\]
(e.g., this holds if the $\bx_i$ are independent
).
Define $S_n = \sum_{i=1}^n \bx_i$. Then
\[
\norm{S_n}{2} = O_p(\sqrt n).
\]
\end{lemma}

\begin{proof}
We first bound the second moment of $\|S_n\|_2$:
\begin{align*}
\mathbb{E}\left[\|S_n\|_2^2 \right]
&= \mathbb{E}\left[\Big(\sum_{i=1}^n \bx_i\Big)^\top\Big(\sum_{j=1}^n \bx_j\Big)\right] \\
&= \sum_{i=1}^n \mathbb{E}\left[\norm{\bx_i}{2}^2\right] \;+\; \sum_{i\neq j}\mathbb{E}\left[\bx_i^\top \bx_j \right] \\
&\le \sum_{i=1}^n C \;+\; 0
= Cn.
\end{align*}
By Markov's inequality,
for any $M>0$,
\begin{align*}
\mathbb{P}\big(\|S_n\|_2 \ge M\sqrt n\big)
= \mathbb{P}\big(\|S_n\|_2^2 \ge M^2 n\big) 
\le \frac{\mathbb{E}\|S_n\|_2^2}{M^2 n}
\le \frac{Cn}{M^2 n}
= \frac{C}{M^2}.
\end{align*}
Given any $\varepsilon>0$, choose $M(\varepsilon)=\sqrt{C/\varepsilon}$. Then for all $n$,
\[
\mathbb{P}\!\left(\frac{\|S_n\|_2}{\sqrt n} > M(\varepsilon)\right)
\le \varepsilon,
\]
which shows that $\|S_n\|_2/\sqrt n$ is bounded in probability. Hence $\|S_n\|_2=O_p(\sqrt n)$.
\end{proof}

\begin{lemma}[Complete randomization with DM weights]\label{lemma:concentrate_vec_sum_dm}
Let $\bx_1,\dots,\bx_n\in\mathbb{R}^k$ with $\norm{\bx}{2}^2 \leq C$, for all $i \in [n]$, and integers $n_T,n_C\ge 1$ with $n_T+n_C=n$.
Draw a treatment set $T\subset[n]$ uniformly among all subsets of size $|T|=n_T$.
Define, for each $i\in[n]$,
\[
z_i =
\begin{cases}
\;\;1, & i\in T,\\
-1, & i\notin T,
\end{cases}
\qquad
v_i =
\begin{cases}
\;\;1/n_T, & z_i=1,\\
\;\;1/n_C, & z_i=-1,
\end{cases}
\qquad
w_i = z_i v_i.
\]
Then $\mathbb{E}[w_i]=0$ for every $i$. 
Let 
\[
S=\sum_{i=1}^n w_i \bx_i = \frac{1}{n_T}\sum_{i\in T}\bx_i-\frac{1}{n_C}\sum_{i\notin T} \bx_i.
\]
Then
\[
\|S\|_2 = O_p\!\left(\sqrt{\frac{n}{n_T n_C}}\right).
\]
In particular, under \Cref{ass:positivity}, $\|S\|_2=O_p(n^{-1/2})$.
\end{lemma}

\begin{proof}
Recall $d_i=1$ if $z_i=+1$ and $d_i=0$ if $z_i=-1$. Therefore $\sum_{i=1}^n d_i=n_T$ and $\mathbb{E}[d_i]=p_T=n_T/n$.
Let $\alpha = \frac{n}{n_T n_C}$. One can verify that $w_i = \alpha (d_i - p_T)$.
Thus
\[
S=\sum_{i=1}^n w_i \bx_i = \alpha \sum_{i=1}^n (d_i-p_T)\bx_i.
\]
Let $\bar \bx = \frac{1}{n}\sum_{i=1}^n \bx_i$ and $\by_i=\bx_i-\bar \bx$, which gives $\sum_{i=1}^n \by_i=0$.
Since $\sum_{i=1}^n (d_i-p_T)=0$, we have
\[
S = \alpha\sum_{i=1}^n (d_i-p_T)\bx_i = \alpha\sum_{i=1}^n (d_i-p_T)\by_i.
\]
Under complete randomization,
\[
\Var(d_i)=p_T(1-p_T),\qquad
\Cov(d_i,d_j)= -\frac{p_T(1-p_T)}{n-1},\quad i\neq j.
\]
Therefore,
\begin{align*}
\mathbb{E}\left[\|S\|_2^2 \right]
&= \alpha^2\,\mathbb{E}\left[ \left\|\sum_{i=1}^n (d_i-p_T)\by_i\right\|_2^2 \right] \\
&= \alpha^2 \sum_{i=1}^n\sum_{j=1}^n \mathbb{E}\big[(d_i-p_T)(d_j-p_T)\big]\; \by_i^\top \by_j \\
&= \alpha^2\left( \sum_{i=1}^n p_T(1-p_T)\|\by_i\|_2^2
+ \sum_{i\neq j}\left(-\frac{p_T(1-p_T)}{n-1}\right)\by_i^\top \by_j \right).
\end{align*}
Since $\sum_i \by_i=\boldsymbol{0}$,
\[
\sum_{i\neq j} \by_i^\top \by_j
= \left\|\sum_{i=1}^n \by_i\right\|_2^2 - \sum_{i=1}^n \|\by_i\|_2^2
= -\sum_{i=1}^n \|\by_i\|_2^2.
\]
Therefore
\[
\mathbb{E}\left[\|S\|_2^2 \right]
= \alpha^2\,p_T(1-p_T)\left(1+\frac{1}{n-1}\right)\sum_{i=1}^n \|\by_i\|_2^2
= \alpha^2\,p_T(1-p_T)\frac{n}{n-1}\sum_{i=1}^n \|\bx_i-\bar \bx\|_2^2.
\]
Finally, $\sum_{i=1}^n \|\bx_i-\bar \bx\|_2^2 \le \sum_{i=1}^n \|\bx_i\|_2^2 \le Cn$, hence
\[
\mathbb{E}\|S\|_2^2
\le
\alpha^2\,p_T(1-p_T)\frac{n}{n-1}\,Cn.
\]
Since $\alpha=\frac{n}{n_T n_C}$ and $p_T(1-p_T)=\frac{n_T n_C}{n^2}$, this simplifies to
\[
\mathbb{E}\|S\|_2^2
\le
\frac{C}{n-1}\cdot \frac{n^2}{n_T n_C}
= O(\frac{n}{n_T n_C}).
\]
Then by Markov's inequality, for any $M>0$,
\[
\mathbb{P}\!\left(\|S\|_2 \ge M\sqrt{\frac{n}{n_T n_C}}\right)
=
\mathbb{P}\!\left(\|S\|_2^2 \ge \frac{M^2n}{n_T n_C}\right)
\le
\frac{\mathbb{E}\|S\|_2^2}{\frac{M^2n}{n_T n_C}} = 
O(\frac{1}{M^2}).
\]
Given any $\varepsilon>0$, we can choose $M:=M(\varepsilon)$ large enough so that the right-hand-side above is smaller than $\varepsilon$.
Therefore $\|S\|_2 = O_p\!\left(\sqrt{\frac{n}{n_T n_C}}\right)$, and under \Cref{ass:positivity}, since $n_T/n$ and $n_C/n$ are bounded away from zero and one,
$\|S\|_2=O_p(n^{-1/2})$.
\end{proof}

\section{Efficiency of LOORA-DM}
\label{sec:lin-variance}

In this section, we study the asymptotic efficiency of the LOORA-DM estimator. Since the estimator of \cite{lin2013agnostic} with interacted terms is asymptotically efficient among linearly adjusted estimators (see \cite{cytrynbaum2024covariate}), it suffices to show that LOORA-DM is asymptotically as efficient as the estimator of \cite{lin2013agnostic}.

The estimator of \cite{lin2013agnostic}, denoted by $\widehat{\tau}_{\textup{Lin}}$, is obtained by regressing the observed outcome vector on $\bd$, $\bX$, and $\bD(\bX - \overline{\bX})$, with an intercept, where $\bD$ denotes the diagonal matrix formed from $\bd$. 
Under complete random assignment (with $n_T$ units assigned to treatment and $n_T/n \to p_T$) and suitable regularity conditions, the asymptotic variance of the estimator is given by
\begin{align*}
    \lim_{n \to \infty}\Var(\sqrt{n}\hat{\tau}_{\textup{Lin}}) &= \frac{1-p_T}{p_T}\lim_{n \to \infty}\sigma_{n}^2 (\by^{(1)}) + \frac{p_T}{1-p_T}\lim_{n \to \infty}\sigma_{n}^2(\by^{(0)})\ +2\lim_{n \to \infty}\sigma_{n}(\by^{(1)}, \by^{(0)}),
\end{align*}
where 
\begin{align*}
& \sigma_{n}^2(\by^{(1)}) = \frac{
  \norm{(\bX - \overline{\bX})\bbeta^{(1)} - (\by^{(1)} - \overline{\by}^{(1)})}{2}^2}{n}, ~ \bbeta^{(1)} = \argmin_{\bb \in \mathbb{R}^{k}} \norm{(\bX - \bar{\bX})\bb - (\by^{(1)} - \overline{\by}^{(1)})}{2}^2,\\
& \sigma_{n}^2(\by^{(0)}) = \frac{
  \norm{(\bX - \overline{\bX})\bbeta^{(0)} - (\by^{(0)} - \overline{\by}^{(0)})}{2}^2}{n},~ \bbeta^{(0)} = \argmin_{\bb \in \mathbb{R}^{k}} \norm{(\bX - \bar{\bX})\bb - (\by^{(0)} - \overline{\by}^{(0)})}{2}^2,\\
  & \sigma_{n}(\by^{(1)},\by^{(0)}) = \frac{(\by^{(1)} - \overline{\by}^{(1)} - (\bX-\bar{\bX})\bbeta^{(1)})^\top(\by^{(0)} - \overline{\by}^{(0)} - (\bX-\bar{\bX})\bbeta^{(0)})}{n}.
  \end{align*}

\cite{lin2013agnostic} considers only the case of complete random assignment and a regression adjustment without regularization (i.e, $\lambda=0$). There are also other differences between our setting and that of \cite{lin2013agnostic}. For example, we do not assume that $n_T/n$ or the variance of $\sqrt{n}\widehat{\tau}_{\textup{LDM}}$ converges; that is, these quantities may fluctuate with $n$.

Therefore, to compare the large-sample asymptotic variance of LOORA-DM with that of \cite{lin2013agnostic}, we restrict attention to the case in which the relevant limits exist and $\lambda = 0$. In particular, we assume the existence of the limit of $\sigma_n^2(\widetilde{\bmu})$, defined as follows.
\[
\sigma_{n}^2 (\widetilde{\bmu}) = \frac{
  \norm{(\bX - \overline{\bX})\bbeta^{(\widetilde{\bmu})} - (\widetilde{\bmu} - \overline{\widetilde{\bmu}})}{2}^2}{n}, ~~ \bbeta^{(\widetilde{\bmu})} =
    \argmin_{\bb \in \R^k}
    \norm{(\bX - \overline{\bX})\bb - (\widetilde{\bmu} - \overline{\widetilde{\bmu}})}{2}^2.
\]
Simple algebra yields
\[
\frac{1-p_T}{p_T}\sigma_n^2(\by^{(1)}) + \frac{p_T}{1-p_T} \sigma_n^2(\by^{(0)}) + 2 \sigma_{n}(\by^{(1)}, \by^{(0)}) = \sigma_{n}^2(\widetilde{\bmu}).
\]
Since $\sigma_{n}^2(\widetilde{\bmu})$ is convergent, we have
\begin{align*}
\lim_{n \to \infty}\sigma_{n}^2 (\widetilde{\bmu}) & = \lim_{n \to \infty}
    \frac{1}{n}\min_{\bb \in \R^k}
    \norm{(\bX - \overline{\bX})\bb - (\widetilde{\bmu} - \overline{\widetilde{\bmu}})}{2}^2.
\end{align*}
Therefore LOORA-DM is as efficient as the estimator of \cite{lin2013agnostic}.\bigskip

\end{appendix}

\end{document}